\documentclass[11pt]{article}

\usepackage{bm, commath}
\usepackage{natbib}
\usepackage{caption}
\usepackage{subcaption}
\usepackage{graphicx}
\usepackage{amsmath, amsfonts, amsthm}
\usepackage{float}
\usepackage{booktabs,siunitx}
\usepackage{url}
\usepackage{multirow}
\usepackage{amssymb}
\usepackage{hyperref}
\usepackage{blkarray}
\usepackage{bbm}  
\usepackage{multicol}
\usepackage{authblk}
\usepackage[shortlabels]{enumitem}
\usepackage[flushleft]{threeparttable}
\usepackage{enumitem}
\usepackage{thmtools}

\makeatletter

\newcommand{\Rmnum}[1]{\expandafter\@slowromancap\romannumeral #1@}
\makeatother

\renewcommand\thmcontinues[1]{Continued}

\usepackage{listings}
\usepackage{bbm}
\usepackage{textcomp}
\usepackage[margin=1in]{geometry}
\usepackage{authblk}
\usepackage[ruled]{algorithm2e}
\SetKwInput{KwParam}{Parameter}
\SetAlgoCaptionLayout{centerline}

\usepackage{sectsty}
\usepackage[doublespacing]{setspace}
\usepackage[utf8]{inputenc}
\usepackage{float}
\usepackage{tikz}
\usetikzlibrary{shapes,decorations,arrows,calc,arrows.meta,fit,positioning}

\usepackage[textsize=scriptsize, textwidth = 2cm, shadow]{todonotes}

\newcommand{\indep}{\mathrel{\perp\mspace{-10mu}\perp}}

\newtheorem{theorem}{Theorem}
\newtheorem{definition}{Definition}

\newtheorem{remark}{Remark}
\newtheorem{example}{Example}
\newtheorem{assumption}{Assumption}
\newtheorem*{assumption*}{Assumption}

\newtheorem{proposition}{Proposition}

\newtheorem{lemma}{Lemma}

\newenvironment{examplecont}[1]
  {\begin{trivlist}
   \item[\hskip \labelsep {\bfseries Example~\ref{#1} {\itshape (continued).}}]
   \itshape}
  {\end{trivlist}}

\usepackage{mathtools}

\usepackage{xcolor}

\makeatletter

\newcommand*\if@single[3]{%
  \setbox0\hbox{${\mathaccent"0362{#1}}^H$}%
  \setbox2\hbox{${\mathaccent"0362{\kern0pt#1}}^H$}%
  \ifdim\ht0=\ht2 #3\else #2\fi
  }
\newcommand*\rel@kern[1]{\kern#1\dimexpr\macc@kerna}
\newcommand*\widebar[1]{\@ifnextchar^{{\wide@bar{#1}{0}}}{\wide@bar{#1}{1}}}
\newcommand*\wide@bar[2]{\if@single{#1}{\wide@bar@{#1}{#2}{1}}{\wide@bar@{#1}{#2}{2}}}
\newcommand*\wide@bar@[3]{%
  \begingroup
  \def\mathaccent##1##2{%
    \let\mathaccent\save@mathaccent
    \if#32 \let\macc@nucleus\first@char \fi
    \setbox\z@\hbox{$\macc@style{\macc@nucleus}_{}$}%
    \setbox\tw@\hbox{$\macc@style{\macc@nucleus}{}_{}$}%
    \dimen@\wd\tw@
    \advance\dimen@-\wd\z@
    \divide\dimen@ 3
    \@tempdima\wd\tw@
    \advance\@tempdima-\scriptspace
    \divide\@tempdima 10
    \advance\dimen@-\@tempdima
    \ifdim\dimen@>\z@ \dimen@0pt\fi
    \rel@kern{0.6}\kern-\dimen@
    \if#31
      \overline{\rel@kern{-0.6}\kern\dimen@\macc@nucleus\rel@kern{0.4}\kern\dimen@}%
      \advance\dimen@0.4\dimexpr\macc@kerna
      \let\final@kern#2%
      \ifdim\dimen@<\z@ \let\final@kern1\fi
      \if\final@kern1 \kern-\dimen@\fi
    \else
      \overline{\rel@kern{-0.6}\kern\dimen@#1}%
    \fi
  }%
  \macc@depth\@ne
  \let\math@bgroup\@empty \let\math@egroup\macc@set@skewchar
  \mathsurround\z@ \frozen@everymath{\mathgroup\macc@group\relax}%
  \macc@set@skewchar\relax
  \let\mathaccentV\macc@nested@a
  \if#31
    \macc@nested@a\relax111{#1}%
  \else
    \def\gobble@till@marker##1\endmarker{}%
    \futurelet\first@char\gobble@till@marker#1\endmarker
    \ifcat\noexpand\first@char A\else
      \def\first@char{}%
    \fi
    \macc@nested@a\relax111{\first@char}%
  \fi
  \endgroup
}

\makeatother

\begin{document}

\sectionfont{\bfseries\large\sffamily}%

\subsectionfont{\bfseries\sffamily\normalsize}%

\title{Generalized projection tests for function-valued parameters with applications to testing structural causal assumptions}

\author[1]{Rui Wang}
\author[1]{Albert Osom}
\author[2]{Bo Zhang\thanks{Assistant Professor of Biostatistics, Vaccine and Infectious Disease Division, Fred Hutchinson Cancer Center. Email: {\tt bzhang3@fredhutch.org}.} }

\affil[1]{Department of Biostatistics, University of Washington}
%\affil[2]{Department of Biostatistics, University of Washington}
\affil[2]{Vaccine and Infectious Disease Division, Fred Hutchinson Cancer Center}

\date{}
\maketitle

\vspace{-0.8cm}

\noindent
\textsf{{\bf Abstract}: Structural assumptions are central to the causal inference literature. In practice, it is often crucial to assess their validity or to test implications that follow from them. In many settings, such tests can be framed as evaluating whether a function-valued parameter equals zero. In this paper, we propose a class of generalized projection tests based on series estimators for function-valued parameters. We establish conditions under which the proposed tests are valid and illustrate their applicability through examples from the data fusion and instrumental variables literature. Our approach accommodates flexible machine learning methods for estimating nuisance parameters. In contrast to many existing approaches, the limiting distributions of the proposed test statistics are straightforward to compute under the null hypothesis. We apply our method to test the equality of conditional COVID-19 risk across vaccine arms in the COVID-19 Variant Immunologic Landscape (COVAIL) trial, and to test an instrumental variable compatibility condition in a study of the effect of family size on female labor supply.}%

\noindent
\textsf{{\bf Keywords}: Causal inference; Debiased machine learning;  Model specification test; Neyman orthogonality; Semiparametric inference; Series estimation}

\section{Introduction}\label{sec: intro}
\subsection{Testing structural causal assumptions; related work}
\label{subsec: intro test causal assumption}
Structural assumptions are ubiquitous in causal inference to relate the observed data to target causal parameters. For instance, alignment assumptions, which specify that certain components of the observed data likelihood align with corresponding components of the target potential outcomes distribution, are often required in the data fusion literature \citep{dahabreh2019generalizing,li2023efficient,He01102024}. In the principal stratification literature \citep{frangakis2002principal}, principal ignorability, which asserts that principal stratum membership is conditionally independent of the potential outcomes given observed covariates, is commonly invoked to generalize the effect from latent principal stratum to the entire population \citep{jo2009use, jiang2022multiply}.

Structural causal assumptions may have testable implications, often expressed in the form of $f_{P} =0$, where $f_{P}$ is a function-valued parameter taking the form of $f_{P} = \mathbb{E}_{P}[g(O;\eta_{P})\mid X=x]$. Here, $P$ is a data-generating distribution, $\eta_{P}$ denotes one or more, possibly function-valued, nuisance parameters, $O$ is a set of observed variables, and $X$ is a subset of $O$. For example, suppose we want to test $H_0: \mathbb{E}_{P}[Y(a)\mid S=1,X=x] = \mathbb{E}_{P}[Y(a)\mid S=0,X=x]$ in a data-fusion problem, where $Y(a)$ is the potential outcome under treatment assignment $a \in\{0,1\}$, $S$ is a binary indicator of data source, and $X$ is a set of baseline covariates. Under the no unmeasured confounding assumption, testing $H_0$ can be reduced to testing $\mathbb{E}_{P}[g(O;\eta_{P})\mid X=x] = 0$, where $g(O;\eta_{P})$ equals $\mathbb{E}_{P}[Y\mid S=1,X=x,A=a] - \mathbb{E}_{P}[Y\mid S=0,X=x,A=a]$, with $\eta_{P} = \{\mathbb{E}_{P}[Y\mid S=1,X=x,A=a],\mathbb{E}_{P}[Y\mid S=0,X=x,A=a]\}$ denoting nuisance functions. 

A common approach to testing a Euclidean-valued parameter is the Wald test, which involves: (1) obtaining a consistent estimate of the target parameter and its standard error, and (2) constructing a test statistic based on estimator's asymptotic distribution after proper scaling. However, applying such a framework to function-valued parameters could be challenging. First, estimating a function-valued parameter is intrinsically difficult, particularly when function's argument is high-dimensional, as this leads to a slower convergence rate. Second, even if a consistent estimator can be obtained, consistently estimating its covariance function is even more complex than in the finite-dimensional case. Third, in many modern causal inference problems, the target function-valued parameter depends on additional nuisance parameters that are themselves function-valued. The presence of these nuisance functions adds another layer of complexity.

To circumvent these challenges, a common strategy is to transform the function-valued parameter into a pathwise-differentiable parameter \citep{bickel1993efficient} and apply tools from the semiparametric theory. Researchers have proposed multiple ways to derive such a pathwise-differentiable parameter. For instance, \citet{wang2024nested} proposed projecting $f_{P}$ onto the linear space spanned by $X$, and constructing a Wald test based on the estimated best linear projection. Another common approach involves evaluating a one-dimensional summary of the discrepancy between $f_p$ and zero. \citet{luedtke2019omnibus} proposed using the maximum mean discrepancy (MMD). Because the first-order influence function of MMD is degenerate under the null hypothesis, they employed the second-order influence function to derive the asymptotic properties of the MMD-based test statistic. Later, \citet{williamson2023general} considered a more general class of discrepancy measures that include the $L^2(P)$ norm of $f_{P}$. 

Another related framework is based on observing that $f_{P} = 0$ implies $\int f_{P} g dP = 0$ for all $g$ in some function class $\mathcal{G}$. Since  $\phi (g):=\int f_{P}gdP$ is usually $n^{1/2}$-estimable, one can first estimate $\phi (g)$, denoted as $\widehat{\phi} (g)$, and then construct test statistics based on some summary statistic of $\{\widehat{\phi} (g),~g\in \mathcal{G}\}$. For instance, if one let the test statistic be the supremum of $\widehat{\phi} (g)$ over $g\in \mathcal{G}$, then the limiting distribution of the test statistic is the supremum of a Gaussian process indexed by $\mathcal{G}$; see, e.g., \citet{hudson2021inference,westling2022nonparametric, jin2024class, dukes2024nonparametric}. One limitation of these approaches is that estimating the quantiles of the resulting test statistic often requires multiplier bootstrap-type procedures, which are computationally expensive, especially in high-dimensional or complex settings \citep{hudson2021inference,westling2022nonparametric, jin2024class}. 

More recently, \citet{luedtke2024one} developed a unified theory for estimation and hypothesis testing for Hilbert-valued parameters. Their proposed test statistic \citep[Equation 30]{luedtke2024one}, when specialized to the problem studied in this article, corresponds closely to an uncentered and unscaled version of our proposed statistic. However, their approach relies on bootstrap methods to approximate the distribution of the test statistic, whereas we explicitly derive and characterize the limiting distribution after properly centering and scaling the test statistic. In addition, we provide a more refined characterization of the asymptotic behavior of the test statistic.

The problem studied in this article is also closely related to the literature on model specification test, which traditionally focuses on assessing whether the conditional mean function lies within a specified parametric or semiparametric model; see, for example, \citet{hong1995consistent} and \citet{fan1996consistent}. More recently, \citet{breunig2015goodness} developed a series-estimator-based goodness-of-fit test for instrument validity within a nonparametric instrumental variable regression framework.

\subsection{Our contribution}
\label{subsec: intro contribution}
In this paper, we develop a \emph{generalized projection test} for function-valued parameters. We establish uniform control of type I error, demonstrate consistency against local alternatives, and illustrate its broad applicability by applying the proposed method to two problems in causal inference. 

The proposed test can be seen as an infinite-dimensional extension of \citet{wang2024nested}, who project $f_P$ onto a space using least squares. In contrast, we project $f_P$ onto a space whose number of basis functions grows with the sample size. This approach leads to substantial power gain under certain alternatives, particularly when $f_P$ lies in the null space of the least squares projection operator in \citet{wang2024nested}.

Our proposed test is also motivated by and extends the model specification tests proposed by \citet{hong1995consistent} and \citet{breunig2015goodness}. Compared to \citet{hong1995consistent} and \citet{breunig2015goodness}, our proposed test accommodates arbitrary function-valued parameters and can also be applied to the parametric conditional mean model specification considered in \citet{hong1995consistent} and the  instrumental variable model specification problem in \citet{breunig2015goodness}. Perhaps more importantly, our proposed test accommodates using modern machine learning tools to flexibly estimate nuisance functions in the target function-valued parameter, aligning with a growing body of work in causal inference that leverages such techniques; see, for example, \citet{2018ChernozhukovDML}, \citet{van2018targeted}, \citet{kennedy2024semiparametric}, and \citet{luedtke2024one}, among others.

Our methodology also relates to recent work on estimating function-valued causal parameters \citep{kennedy2023towards,yang2023forster,semenova2021debiased,luedtke2024one}. While this literature focuses on estimation, we instead address hypothesis testing. Under the null hypothesis, the function-valued parameter is known a priori, which provides additional structure that our approach exploits for inference. As a result, we can construct valid tests without requiring consistent estimation of the target function or the construction of valid confidence bands.

Finally, relative to approaches by \citet{westling2022nonparametric}, \citet{jin2024class}, \citet{dukes2024nonparametric}, and \citet{luedtke2024one}, our proposed test features a simple, analytically tractable limiting distribution under the null and is therefore computationally efficient. 

To the best of our knowledge, there is limited statistical software for conducting hypothesis tests involving function-valued parameters. One exception is the \textsf{npsigtest} function in the \textsf{np} package, which is designed to test a specific function-valued parameter in a particular setting, namely, the conditional mean independence problem \citep{JSSv027i05}. To address this gap, we implement the proposed hypothesis testing framework in statistical software \textsf{R}, and the implementation is publicly available from
{\centering
    \url{https://github.com/wangrui24/Generalized-Projection-tests}.
}

\subsection{Notation}
\label{subsec: intro notation}
We use $\Vert \cdot \Vert$ to denote the $L^2$ norm for vectors and the spectral norm for matrices, and $\Vert \cdot \Vert_F$ for the Frobenius norm of a matrix. The space of square-integrable functions with respect to a probability measure $P$ is denoted by $L^2(P)$, with corresponding norm $\Vert \cdot \Vert_{P,2}$. We write $o(1)$ for a sequence converging to zero and $O(1)$ for a bounded sequence. We write $a_n \asymp b_n$ if $a_n = O(b_n)$ and $b_n = O(a_n)$. Let $\delta_n, \zeta_n, \xi_n, \kappa_n, \Delta_n$ be positive sequences converging to zero. For a positive semidefinite matrix $M$ with spectral decomposition $M = \Gamma^\top \Lambda \Gamma$, we define $M^{1/2} =  \Lambda^{1/2}\Gamma$, then $M = (M^{1/2})^\top M^{1/2}$. Finally, $Q(X;\alpha)$ denotes the $\alpha$-quantile of a random variable $X$.
\section{A generalized projection test for function-valued parameters}\label{sec: general test}
\subsection{A fixed-dimension projection test} \label{subsec: method setup}
Let $\mathcal{P}$ denote a statistical model dominated by measure $\mu$. We consider testing the null hypothesis $H_0: P \in \mathcal{P}_0 :=  \{P\in \mathcal{P}:f_P(X) = 0\}$, where $f_P(X):=\mathbbm{E}_{P}[g(O;\eta_{P})\mid X] = 0$ almost surely $P_{X}$, $g$ is a known function, $O$ denotes the observed data, $X$ is a subset of $O$, $P_{X}$ is the marginal distribution of $X$ implied by $P$, and $\eta_{P}$ denotes one or more Euclidean- or function-valued nuisance parameters under distribution $P$. We further assume that $f_P(X) \in L^2(P)$. 

While the main focus of the paper is to test structural causal assumptions, the above problem setup is generic enough to encompass many well-studied problems in a literature broader than causal inference. We highlight a few notable examples in Supplementary Material \ref{ex:param}.

Before introducing the generalized projection test, we briefly review the projection test and related literature. A projection test is based on deriving a low-dimensional projection of a high-dimensional parameter. In estimation problems, researchers have investigated how to estimate the best least squares projection of a function-valued parameter; see, e.g., \citet{chernozhukov2018generic,ye2023instrumented}. \citet{wang2024nested} studied hypothesis testing for the null hypothesis $\mathbb{E}_{P}[g(O;\eta_{P})\mid X] = 0$ using a least squares projection approach. \citet{wang2024nested} first specify a finite-dimensional function class of $X$, denoted as $\Gamma = \{\gamma(x;\theta);\theta \in \Theta \subset \mathbb{R}^d \}$, satisfying $\gamma(x;\theta) = 0$ if $\theta=0$ and $0\in \Theta$. A typical example is to let $\{\gamma(x;\theta) = x^\top\theta,~\theta \in \mathbb{R}^d\}$ be the class of linear functions. The least squares projection of $f_P$ on $\Gamma$, denoted $\Pi[f_{P}\mid \Gamma]$,  can be written as $X^\top\theta_{\text{LS}}$, where
$\theta_{\text{LS}} = \{\mathbbm{E}_P[X X^\top]\}^{-1}\mathbbm{E}_{P}[Xf(X)] = \{\mathbbm{E}_P[X X^\top]\}^{-1}\mathbbm{E}_{P}[Xg(X;\eta_{P})].$

Under the null hypothesis $f_{P}(X) = 0$, it follows that $\theta_{\text{LS}} = 0$. To test the hypothesis, \citet{wang2024nested} proposed a Wald-type test statistic, defined as $\widehat{W}_n = n\widehat{\theta}_{\text{LS}}^\top\widehat{\Sigma}_{\text{LS}}^{-1}\widehat{\theta}_{\text{LS}}$, where (1) $\widehat{\theta}_{\text{LS}}$ is a $n^{1/2}$-consistent estimate for $\theta_{\text{LS}}$; and (2) $\widehat{\Sigma}_{\text{LS}}$ is a consistent estimator of the variance-covariance matrix $\Sigma_{\text{LS}}$ of $\widehat{\theta}_{\text{LS}}$. With suitable estimators of the nuisance functions $\widehat{\eta}_n$ and treating $g(O_i;\widehat{\eta}_n)$ as a pseudo-outcome, $\widehat{\theta}_{\text{LS}}$ can be obtained via ordinary least squares:
\begin{align*}
    & \widehat{\theta}_{\text{LS}} = \left\{\frac{1}{n}\sum_{i=1}^n X_i X_i^\top\right\}^{-1}\left\{\frac{1}{n}\sum_{i=1}^n X_ig(O_i;\widehat{\eta}_n)\right\},
\end{align*}
$\widehat{\Sigma}_{\text{LS}}$ can be obtained using the Huber-White heteroskedasticity-robust covariance estimator \citep{white1980heteroskedasticity}:
\begin{align*}
    & \widehat{\Sigma}_{\text{LS}} = \left\{\frac{1}{n}\sum_{i=1}^n X_i X_i^\top\right\}^{-1}\left\{\frac{1}{n}\sum_{i=1}^n X_i X_i^\top\left[g(O_i;\widehat{\eta}_n)-X_i^\top \widehat{\theta}_{\text{LS}}\right]^2\right\}\left\{\frac{1}{n}\sum_{i=1}^n X_i X_i^\top\right\}^{-1},
\end{align*}
and the test statistic $\widehat{W}_n$ simplifies to:
\begin{align*}
     \widehat{W}_n = n\underbrace{\left\{\frac{1}{n}\sum_{i=1}^nX_ig(O_i;\widehat{\eta}_n)\right\}^\top}_{\text{Projection term}}\underbrace{\left\{\frac{1}{n}\sum_{i=1}^n X_i X_i^\top\left[g(O_i;\widehat{\eta}_n)-X_i^\top\widehat{\theta}_{\text{LS}}\right]^2\right\}^{-1}}_{\text{A positive semidefinite weighting matrix}} \underbrace{\left\{\frac{1}{n}\sum_{i=1}^nX_ig(O_i;\widehat{\eta}_n)\right\}}_{\text{Projection term}}.
\end{align*}
The first term in $\widehat{W}_n$, referred to as the projection term, is an empirical analogue of the projection $\mathbbm{E}_{P}[Xg(O,\eta_{P})]$. The second term is a positive semidefinite (PSD) weighting matrix. Under suitable regularity conditions,  $\widehat{W}_n$ converges in distribution to a chi-square distribution with $d$ degrees of freedom \citep{wang2024nested}. This projection-based test offers several advantages: it rigorously controls the type I error, is computationally efficient, and can readily incorporate machine-learning-based estimators for nuisance functions. However, if $f_{P}$ is highly nonlinear, or worse, if it lies in the null space of the projection operator, then the test may have little or no power. This limitation motivates the development of a generalized projection test.

\subsection{The generalized projection test}
\label{subsec: method generalized proj test}
Let $B_{n}(X) = (b_{1}(X),...,b_{J_n}(X)^\top$ be a $J_n$-dimensional basis in $L^2(P)$. In contrast to the fixed-dimension projection test, we allow the dimension $J_n$ to grow with the sample size $n$ at a rate to be specified later. We consider the following test statistic as the basis for constructing a generalized projection test:
\begin{align}
\label{eq: test statistic}
     \widehat{S}_n = n\underbrace{\left(\frac{1}{n}\sum_{i=1}^nB_n(X_i)g(O_i;\widehat{\eta}_n)\right)^\top}_{\text{projection term}}{\large\underbrace{\widehat{\Omega}_n}_{\substack{\text{weighting} \\ \text{matrix}}}} \underbrace{\left(\frac{1}{n}\sum_{i=1}^nB_n(X_i)g(O_i;\widehat{\eta}_n)\right)}_{\text{projection term}},
\end{align}
where the projection term is an empirical analogue of $\mathbbm{E}_{P}[B_n(X)g(O;\eta_{P})]$, and $\widehat{\Omega}_n$ is a $J_n \times J_n$-dimensional PSD weighting matrix which can be either prespecified or data-dependent. The generalized projection test will be based on the asymptotic distribution of $\widehat{S}_n$ or a standardized version of it, as we will present formally in Section \ref{sec: asymptotic theory}. \textcolor{black}{We do not a priori impose assumptions on the eigenvalues of the PSD weighting matrix $\widehat{\Omega}_n$ when defining the test statistic $\widehat{S}_n$; however, when studying the asymptotic distribution of $\widehat{S}_n$, we will impose further regularity conditions on $\widehat{\Omega}_n$. } 

To understand why the generalized projection test may have better power under certain alternatives than the projection test reviewed in Section \ref{subsec: method setup}, it is useful to examine the behavior of the test statistic $\widehat{S}_n$ under a generic alternative hypothesis. Suppose first that $\widehat{\Omega}_n$ is an identity matrix. In this case, $\widehat{S}_n$ is the empirical counterpart of $n\sum_{j=1}^\infty \{\mathbbm{E}_P[b_j(X)g(O;\eta_{{P}})]\}^2$, which is nonzero whenever $f \in L^2(P) \backslash \{0\}$, provided that the basis $\{b_j(X)\}_{j=1}^\infty$ is a complete basis for $L^2(P_X)$. 

More generally, suppose $\widehat{\Omega}_n$ is a positive semidefinite matrix that converges to ${\Omega}_{n,P}$ (to be formalized in Assumption \ref{Assumption: weighting matrix} below), and that ${\Omega}_{n,P}$ admits a spectral decomposition of the form $\Gamma_{n,P}^\top\Lambda_{n,P} \Gamma_{n,P}$, where $\Gamma_{n,P}$ is an orthogonal matrix and $\Lambda_{n,P} = \text{diag}(\lambda_{1,P},..,\lambda_{J_n,P})$. In this case, $\widehat{S}_n$ is an empirical counterpart of $n\sum_{j=1}^\infty \left\{\mathbbm{E}_P[\Tilde{b}^{\lambda}_{j,P}(X)g(O;\eta_{{P}})]\right\}^2$, where $\Tilde{b}^{\lambda}_{j,P} = \lambda^{1/2}_{j,P}\tilde{b}_{j,P}$, and $\tilde{b}_{j,P}$ is the $j$-th row of the rotated basis $\Gamma_{n,P}B_n(X)$. This expression remains nonzero as long as  $f \in L^2(P) \backslash \{0\}$ and $\Omega_{n,P}$ is positive definite.

\subsection{Debiasing nuisance function estimation}
\label{subsec: debias nuisance}
A defining feature of many modern causal inference problems is the presence of function-valued nuisance parameters, denoted by $\eta_P$, which appear in the target functional $f_P$. Estimating these nuisance function components can introduce first-order bias into the estimation of the primary statistical functional. To mitigate this and enable valid inference while allowing for flexible, data-driven estimation of $\eta_P$, two key techniques are typically employed: Neyman orthogonality and cross-fitting \citep{2018ChernozhukovDML}. In our setting, the statistical functional of interest is a conditional expectation. Accordingly, we invoke a form of \emph{conditional Neyman orthogonality}; see, e.g., \citet{chernozhukov2024conditional}.

\begin{definition}
    \textbf{(Conditional Neyman orthogonality).} Let $\mathcal{T}$ be a set of perturbation directions. A function $g^\perp(o;\eta)$ is said to be conditional Neyman orthogonal, or Neyman orthogonal given $X$, if for any $\eta_1\in \mathcal{T}$,
    \begin{align*}
        \frac{\partial}{\partial t}\mathbb{E}_{P}\left[g^\perp(O;(1-t)\eta_{P}+t\eta_{1})\mid X\right]\Big\vert_{t=0} = 0.
    \end{align*}
    \label{definiftion: conditional neyman orthogonal}
\end{definition}

% \noindent A function $g^\perp(O;\eta)$ that satisfies the conditional Neyman orthogonal condition can often be obtained by calculating the gradient of the orthogonal function-valued parameter \citep{luedtke2024one,chernozhukov2024conditional}. The conditional Neyman orthogonality condition implies that the first-order bias of estimating nuisance parameter is zero.

{\color{black}

Supplemental Material~\ref{section:neyman orthogonality and conditional neyman orthogonality} establishes formal connections between Neyman orthogonal and conditional Neyman orthogonal functions. Under regularity conditions, the former is also the latter. Thus, for many problems, one may take a Neyman orthogonal function from the literature and directly verify Definition~\ref{definiftion: conditional neyman orthogonal}. An alternative way to construct a candidate conditional Neyman orthogonal function is to add a first-order influence function term to the original $g(o;\eta)$, that is, $g^\perp(o;\eta) = g(o;\eta) + h(o;\eta)$, where $h(\cdot;\eta)$ satisfies $\mathbb{E}_P[h(O;\eta_P) \mid X] = 0$ and can be constructed by computing the pathwise derivative of $\mathbb{E}_P[g(O;\eta) \mid X]$ with respect to $\eta$ \citep{yang2023forster,chernozhukov2024conditional}. Hence, under the null hypothesis, $\mathbb{E}_P[g^\perp(O;\eta) \mid X] = \mathbb{E}_P[g(O;\eta) \mid X] = f(X) = 0$. After constructing a candidate $g^\perp(o;\eta)$ from $g(o;\eta)$, one must still verify that it satisfies the conditional Neyman orthogonality condition.

Our conditional Neyman orthogonality condition is also related to \citet[Condition 3]{luedtke2019omnibus} and \citet[Assumptions of Theorem 2]{yang2023forster}, who assume that the parameter satisfies a conditional version of pathwise differentiability, whereas we formulate this condition in terms of the Gateaux derivative. These two formulations are related, as discussed in \citet{chernozhukov2022automatic, chernozhukov2024conditional, chen2026equivalence}, and recapitulated in Supplemental Material~\ref{subsec: additional discussion on neyman orthogonality}.

In the rest of the article, instead of $g(o; \eta)$, we work with $g^\perp(o; \eta)$, whose conditional Neyman orthogonality ensures that nuisance-parameter estimation contributes no first-order bias.} To construct the proposed statistic, we further employ cross-fitting \citep{klaassen1987consistent,schick1986asymptotically}, which alleviates restrictive assumptions on nuisance-model complexity, such as the Donsker conditions \citep{van2000asymptotic}. Under conditional Neyman orthogonality and using cross-fitting, nuisance functions $\eta$ in $g^\perp(o;\eta)$ can be replaced by its estimator obtained via flexible machine learning tools.

\subsection{\textcolor{black}{Hypothesis testing algorithms; choosing sieve dimensions}}
\label{subsec: Algorithm}
{\color{black}
We propose two hypothesis tests based on unstandardized and standardized test statistics. Let $\widehat{\rho}_n$ and $\widehat{\gamma}_n$ denote the trace and Frobenius norm of $\widehat{\Sigma}_n$, defined as
\[
\widehat{\Sigma}_n = \frac{1}{n} \sum_{i=1}^n {g^\perp}^2(O_i; \widehat{\eta}_n) \left( \widehat{\Omega}_n^{1/2} B_n(X_i) \right) \left( \widehat{\Omega}_n^{1/2} B_n(X_i) \right)^\top.
\]
The procedures are summarized in Algorithms \ref{alg: unstd} and \ref{alg: std}.

\begin{algorithm}[H]
\caption{Hypothesis testing based on the unstandardized test statistic}
\label{alg: unstd}
\textbf{Input:} $J_n$-dimensional basis for each continuous $X$; level $\alpha$\;
\textbf{Step I:} Derive the conditional Neyman orthogonal function $g^\perp(o; \eta)$\;
\textbf{Step II:} Split the data into $K$ folds. For each fold $k \in [K]$, use data from the other $K-1$ folds to estimate the nuisance function $\widehat{\eta}_n$. Then, compute $g^\perp(O_i; \widehat{\eta}_n)$ for all observations $i$ in the $k$-th fold. Repeat this process for all folds so that each observation has an associated value of $g^\perp(O_i; \widehat{\eta}_n)$ based on an out-of-fold estimate\;

\textbf{Step III:} Compute the test statistic $\widehat{S}_n$ as defined in Equation \eqref{eq: test statistic}\;

\textbf{Step IV:} Calculate $p$-value by comparing $\widehat{S}_n$ to the $1 - \alpha$ quantile of $\sum_{j=1}^{J_n} \widehat{\tau}_{j}\chi^2_j(1)$, where $\widehat{\tau}_{j}$ are eigenvalues of $\widehat{\Sigma}_n$, and $\chi^2_j(1)$ are independent chi-square random variables with one degree of freedom.
\end{algorithm}

\begin{algorithm}[H]
\caption{Hypothesis testing based on the standardized test statistic}
\label{alg: std}
\textbf{Input, Step I \& Step II:} Same as in Algorithm \ref{alg: unstd}\;

\textbf{Step III:} Calculate $\widehat{S}_n$ as defined in Equation \eqref{eq: test statistic} and then standardize it as follows:
\[
\widehat{T}_n = \frac{\widehat{S}_n - \widehat{\rho}_n}{\sqrt{2}\widehat{\gamma}_n}.
\]
\textbf{Step IV:} Compare $\widehat{T}_n$ to the $1 - \alpha$ quantile of the standard normal distribution to determine statistical significance.
\end{algorithm}

Both algorithms require specifying the sieve dimension $J_n$. For estimation problems, where the goal is typically to minimize a loss function, the optimal number of basis functions (or other hyperparameters) can be selected using cross-validation. In the hypothesis testing setting, however, the aim in choosing $J_n$ is to maximize the power of the test while controlling its type I error, and no cross-validation procedure is apparently applicable. Rather than selecting a single optimal $J_n$, we may adopt the approach in \citet{horowitz2001adaptive} and \citet{breunig2024adaptive} and combine test results across multiple sieve dimensions on a carefully constructed grid, using a Bonferroni correction to maintain overall type I error control.

Specifically, let $\mathcal{J}_n = \{J=\underline{J}_n 2^{j}, j = 0,...,j_{\max},\underline{J}_n 2^{j_{\max}}\leq J_{\max}\}$ denote an exponentially spaced grid of candidate numbers of basis functions, with  $\underline{J}_n$ and $j_{\max}$ chosen so that each $J_n \in \mathcal{J}_n$ satisfies the rate condition to be specified in Section \ref{sec: asymptotic theory}. For each $J_n \in \mathcal{J}_n$, we construct a test using either Algorithm \ref{alg: unstd} or Algorithm \ref{alg: std} evaluated at level $\alpha/\vert \mathcal{J}_n\vert$, resulting in $\vert \mathcal{J}_n\vert$ tests in total. We reject the null hypothesis if and only if at least one of these tests rejects. This exponentially spaced grid is chosen for two reasons: first, it keeps the number of simultaneous tests small ($\vert \mathcal{J}_n\vert = O(\log n)$); second, the resulting range of dimensions remains wide enough to detect alternatives across various smoothness levels. 

%Unlike estimation problem, where the goal is to minimize the loss and number of basis function can be chosen by using cross-validation, the goal for choosing number of basis functions here is to control type one error while maximize power. Here we adopt the similar method as it in \citet{breunig2024adaptive} to combine different tests. Let $\mathcal{J}_n$ be a set of candidate number of basis. For each $J_n \in \mathcal{J}_n$ we can construct a test using either standardized test or standardized test with level $\alpha/\vert \mathcal{J}_n\vert$, therefore we end up with $\vert \mathcal{J}_n\vert$ tests. We reject the null if and only if there is at least one test reject the null hypothesis. 
}
\section{Theoretical properties of the proposed tests}
\label{sec: asymptotic theory}

\subsection{Limiting behavior under the null hypothesis}
\label{subsec: theory level}
 We present results on uniform type I error control of \textcolor{black}{Algorithms \ref{alg: unstd} and \ref{alg: std} over all $P\in \mathcal{P}_0$.} We first discuss technical assumptions.
 
 Assumption \ref{Assumption: basis} imposes regularity conditions on the basis $B_n(x)$.
\begin{assumption}[Basis]
Let $I_{J_n}$ denote a $J_n\times J_n$ identity matrix and $B_n(x) = (b_1(x),b_2(x),\dots,b_{J_n}(x))^\top$ be a sequence of basis functions satisfying: 
\begin{enumerate}
    \item $C^{-1}\cdot I_{J_n}\leq  \mathbbm{E}_{P}[B_n(X)B^\top_n(X)]\leq C\cdot I_{J_n}$ for some constant $C$ for all $P\in \mathcal{P}$;
    \item For all $j\geq 1$, $\sup_{x\in \mathcal{X}} \vert b_j(x) \vert \leq \xi_n$ with $\xi_n \leq C{J_n}^{1/2}$, and $\sup_{x\in \mathcal{X}}\Vert B_n(x)\Vert_2 \leq \omega_n$ with $\omega_n \leq C{J_n}$.
\end{enumerate}
    \label{Assumption: basis}
\end{assumption}

When $B_n(X)$ consists of the first $J_n$ elements of an orthonormal basis of $L^2(P_{X})$, Assumption \ref{Assumption: basis}.1 holds trivially for $C = 1$. In practice, however, the orthonormal basis of $L^2(P_{X})$ is not available because $P_X$ is unknown. If the support of $X$, denoted $\mathcal{X}$, is a product of closed intervals, one can then take $B_n(X)$ to be the first $J_n$ elements of the orthonormal basis of $L^2({P}^{\text{uniform}}_{X})$, where $L^2({P}^{\text{uniform}}_{X})$ denotes the uniform distribution on $\mathcal{X}$. \citet[Proposition 2.1]{belloni2015some} show that, provided the density of $X$ is bounded, Assumption \ref{Assumption: basis}.1 holds for this choice of $B_n(X)$. 

Assumption \ref{Assumption: basis}.2  holds for many commonly used basis functions on compact intervals, including the Fourier series with $\xi_n \asymp C$ and $\omega_n \asymp {J_n}^{1/2}$, and orthonormal Legendre polynomials with $\xi_n \asymp {J_n}^{1/2}$ and  $\omega_n \asymp {J_n}$ \citep{belloni2015some}.

Assumption \ref{Assumption: weighting matrix} gives conditions on the weighting matrix $\widehat{\Omega}_{n}$ and the convergence rate of its estimator $\widehat{\Omega}_{n}$. When $\widehat{\Omega}_n = I_{J_n}$, Assumption \ref{Assumption: weighting matrix} trivially holds. 

\begin{assumption}[Weighting matrix]
For $P\in \mathcal{P}$, the weighting matrix $\widehat{\Omega}_n$ in $S_n$ and its population counterpart ${\Omega}_{n,P}$ satisfy the following:
\begin{enumerate}
    \item $\Omega_{n,P}$ is a positive definite matrix with eigenvalues $\{\lambda_{P,1},...,\lambda_{P,J_n}\}$; furthermore, there exists a constant $C$ such that  $\max\{\lambda_{P,1},...,\lambda_{P,J_n}\}\leq C$ holds for all $n$.
    \item With probability at least $1-\Delta_n$,
        $\Vert \widehat{\Omega}_n-\Omega_{n,P} \Vert \leq   \delta_n/{J_n}^{1/2}$.
\end{enumerate}
\label{Assumption: weighting matrix}
\end{assumption}

Assumption \ref{Assumption: nuisance} specifies the convergence rate for nuisance function estimation and serves as a key condition enabling the use of machine learning methods for estimating nuisance functions in the target statistical parameter.
\begin{assumption}[Nuisance function estimation]
The following conditions hold for $P\in \mathcal{P}$:
     \begin{enumerate}
         % \item Under the null hypothesis, $\mathbb{E}_{P_0}[h(O;\eta_{0})|X] = 0$;
         \item There exists a positive constant $C$, such that $\mathbb{E}_{P}[{g^\perp}^4(O;\eta_{P})\mid X] \leq C$ and $\mathbb{E}_{P}[{g^\perp}^2(O;\eta_{P})\mid X] \geq 1/C$;
         \item $g^{\perp}(O;\eta)$ is conditional Neyman orthogonal;
         \item The nuisance estimate $\widehat{\eta}_n$ belongs to a realization set $\mathcal{T}_n$ with probability {at least} $1-o(1)$, where $o(1)$ is some fixed sequence converging to $0$. Define the statistical rates \[
         r_{n,P}: = \sup_{\eta\in \mathcal{T}_n}\left(\mathbb{E}_{P}\left[\left\vert g^{\perp}(O;\eta)-g^{\perp}(O;\eta_{P})\right\vert^2\right]\right)^{1/2}
         \] and 
         \[r'_{n,P} : = \sup_{t\in[0,1],\eta_{P_1}\in \mathcal{T}_n}\mathbbm{E}_P\left\vert \frac{\partial^2}{\partial t^2}\mathbb{E}_{P}\left[g^{\perp}(O;(1-t)\eta_{P}+t\eta_{P_1})\mid X\right]\right\vert.\]
         We assume that for $\eta \in \mathcal{T}_n$,
         \begin{align*}
             &\mathbb{E}_{P}[(g^\perp)^4(O;\eta)\mid X] \leq C,\\
&\mathbbm{E}_P[(g^{\perp}(O;\eta)-g^{\perp}(O;\eta_P))^2(g^{\perp}(O;\eta)+g^{\perp}(O;\eta_P))^2]\leq C\mathbbm{E}_P[(g^{\perp}(O;\eta)-g^{\perp}(O;\eta_P))^2],\\
&
r_{n,P} \leq \zeta_n, r'_{n,P} \leq \zeta_n/{n}^{1/2},
         \end{align*}
         where 
         \begin{enumerate}
             \item $\zeta_n(\xi_n+\omega_n) = o(1/{J_n}^{1/2})$ for the standardized generalized projection test;
             \item $\zeta_n\xi_n = o(1/{J_n})$ for the unstandardized generalized projection test;
             \item $\zeta_n\xi_n = o\left(1/\text{tr}(\Omega_{n,P})\right)$ for the unstandardized generalized projection test and when $\widehat{\Omega}_n = \Omega_{n,P} = \text{diag}(\lambda_1,...,\lambda_{J_n})$ with known $\lambda_1,...,\lambda_{J_n}$.
         \end{enumerate}
     \end{enumerate}
    \label{Assumption: nuisance}
\end{assumption}

The condition $1/C \leq \sup_{\eta \in \mathcal{T}_n}\mathbb{E}_{P}[{g^\perp}^4(O;\eta)\mid X] \leq C$ is typical for controlling the conditional moment. A sufficient condition for $\mathbbm{E}_P[(g^\perp(O;\eta)-g^\perp(O;\eta_P))^2(g^\perp(O;\eta)+g^\perp(O;\eta_P))^2]\leq C\mathbbm{E}_P[(g^\perp(O;\eta)-g^\perp(O;\eta_P))^2]$ is that $\sup_{\eta \in \mathcal{T}_n}\vert g^\perp(o;\eta)\vert$ is uniformly bounded. 

\textcolor{black}{
The required convergence rates of the nuisance estimators depend mainly on the growth rate of the number of basis functions and on the choice of basis. For the hypothesis testing procedure based on the standardized test statistic (Algorithm \ref{alg: std}), a crude rate for nuisance parameter estimation is $o\left(1/((\xi_n+\omega_n)^{1/2}{(nJ_n)}^{1/4})\right)$, which is a stronger requirement than the typical rates in the double machine learning literature \citep{2018ChernozhukovDML}. Later we will let $J_n = o(n^{1/3})$; thus, a sufficient condition is that the nuisance estimator achieves a convergence rate faster than $1/(n^{(5/12-\epsilon/2)})$, assuming $J_n = n^{1/3-\epsilon}$ for some $0<\epsilon <1/3$, $\xi_n\asymp C$ and $\omega_n \asymp {J_n}^{1/2}$. When $\epsilon$ is close to $0$, the nuisance convergence rate condition is close to the parametric rate. When $\epsilon$ is close to $1/3$, the convergence rate condition is close to $n^{-1/4}$, which is the usual rate requirement for a debiased machine learning estimator \citep{chernozhukov2018generic} and can be achieved by assorted nonparametric and semiparametric estimators. In Supplemental Material \ref{sec: supp nuisance rates}, we give more detailed discussion on nuisance rates. Stronger convergence-rate requirements for nuisance parameter estimation in the presence of a growing number of basis functions or moment conditions have been noted before; see, e.g. \citet{semenova2021debiased} and \citet{wang2025gmm}.} Assumption \ref{Assumption: nuisance} also implies that $\left\Vert\frac{1}{n}\sum_{i=1}^n \left\{g^\perp(O_i;\widehat{\eta}_n)-g^\perp(O_i;\eta_{P})\right\}B_n(X_i)\right\Vert = o_p(1/{n}^{1/2})$, which ensures that nuisance estimation does not introduce first-order bias into the test statistic. When $\widehat{\Omega}_n$ is a known diagonal matrix with $\lim_{J_n \rightarrow \infty}\sum_{j=1}^{J_n}\lambda_{j} <\infty$, we only need $\zeta_n\xi_n = o(1)$ for unstandardized test. If we further assume $\xi_n \asymp C$, we only need $\zeta_n = o(1)$ , which then coincides with the classical rate requirement for nuisance estimation in the double machine learning literature \citep{chernozhukov2018generic}.

%Similar rate conditions arise in settings where the number of basis functions or moment conditions increases with sample size; see, e.g., \citet{semenova2021debiased}.

Theorem \ref{Theorem: asymptotic distribution chi-sqaure approximation} states that the unstandardized test statistic $\widehat{S}_n$ can be approximated by a weighted sum of $J_n$ independent chi-square random variables with one degree of freedom.

%where $J_n$ increases with the sample size.

\begin{theorem}[Weighted chi-square approximation of the unstandardized test statistic]
    Suppose the null hypothesis holds. Under Assumptions \ref{Assumption: basis} to \ref{Assumption: nuisance},  $\frac{J_n^{5/4}\omega_n}{{n}^{1/2}}\rightarrow 0$, $\delta_n\sqrt{J}_n = o(1)$, and $\inf_{P\in \mathcal{P}}\lambda_{\max}(\Omega_{n,P})>1/C$, we have
    \begin{align*}
        \sup_{P\in \mathcal{P}_0}\sup_{t\in \mathbb{R}} \left\vert P(\widehat{S}_n \leq t)-P\left(\sum_{j=1}^{J_n}\tau_j\chi^2_j(1)\leq t\right)\right\vert \rightarrow 0,
    \end{align*} 
    where $\tau_1,...,\tau_{J_n}$ are eigenvalues of 
    \begin{align*}
        \Sigma_{n,P}   = \mathbbm{E}_{P} \left[{g^\perp}^2(O_i; {\eta}_P) \left( {\Omega}_{n,P}^{1/2} B_n(X_i) \right) \left( {\Omega}_{n,P}^{1/2} B_n(X_i) \right)^\top\right].
    \end{align*}
    \label{Theorem: asymptotic distribution chi-sqaure approximation}
\end{theorem}

\begin{proof}
    All proofs in this paper can be found in the Supplementary Materials \ref{sec:technical lemma} and \ref{sec:main results proof}. 
\end{proof}

Theorem \ref{Theorem: asymptotic distribution chi-sqaure approximation} implies that an asymptotically uniformly size-$\alpha$ test can be constructed as 
$$\phi_{\alpha, n} = \mathbbm{1}\left\{\widehat{S}_n\geq Q\left(\sum_{j=1}^{J_n}\tau_j\chi^2_j(1),1-\alpha\right)\right\}.$$ Since $\{\tau_j\}_{j=1}^{J_n}$ are typically unknown, they are replaced in practice by their estimators $\{\widehat{\tau}_j\}_{j=1}^{J_n}$, which correspond to the eigenvalues of $\widehat{\Sigma}_n$.

\begin{remark}[Comparison to \citet{breunig2015goodness}]
    Theorem \ref{Theorem: asymptotic distribution chi-sqaure approximation} differs from \citet[Theorem 2.2]{breunig2015goodness} in several important respects. First, our result establishes uniform control of the type I error rate, supported by a novel Berry Esseen bound with explicit constants \citep{raivc2019multivariate}. In contrast, \citet[Theorem 2.2]{breunig2015goodness} guarantees valid type I error control only for a fixed data-generating distribution. Second, Theorem \ref{Theorem: asymptotic distribution chi-sqaure approximation} does not require the limit $\lim_{J_n\rightarrow \infty}\sum_{j=1}^{J_n}\tau_j \chi_j^2(1)$ to be well defined, whereas \citet[Theorem 2.2]{breunig2015goodness} imposes the condition $\sum_{j=1}^{J_n}\tau_j = O(1)$ to ensure convergence of $\sum_{j=1}^{\infty}\tau_j \chi_j^2(1)$. Third, we impose a high-level rate condition on the nuisance estimators that accommodates a broad class of nonparametric and machine learning methods, while \citet{breunig2015goodness} focus on series estimators.
\end{remark}

%For example, suppose $\widehat{\Omega}_n = [\frac{1}{n}\sum_{i=1}^n B_n(X_i)B_n(X_i)^\top g^2(O;\eta_{\widehat{P}})]^{-1}$, then $\Omega_n = \mathbbm{E}_{P}[B_n(X_i)B_n(X_i)^\top g^2(O;\eta_{P})]$, $\Sigma = I_{J_n}$, $\tau_1 = \tau_2 = ...=\tau_{J_n}=1$. However, the distribution function of $\sum_{j=1}^{J_n}\tau_j \chi_j^2$ does not converge.

% \begin{theorem}
%     Under Assumptions \ref{Assumption: basis} to \ref{Assumption: nuisance}, furthermore, suppose that $\sum_{j=1}^{J_n}\lambda_n = O(1)$, then under the the null hypothesis,
%     \begin{align*}
%         nS_n \rightsquigarrow \sum_{j=1}^\infty \tau_j\chi_{1j}^2,
%     \end{align*}
%     where $\tau_j$ is the ordered eigen-values of the matrix $\mathbbm{E}_{P_0}[g^2(O_i;\eta_0)B_n(X)\Omega_nB^\top_n(X)]$, $\{\chi^2_{1j}\}_{j=1}^\infty$ is a sequence of independent chi-sqaure random variables with degree of freedom $1$.
% \end{theorem}

{\color{black} Theorem \ref{Theorem: asymptotic distribution chi-sqaure approximation} may be cumbersome to apply in practice for two reasons. First, $\sum_{j=1}^{J_n} \tau_j \chi_j^2(1)$ is an unknown distribution, since the $\tau_j$ are unknown and must be estimated via singular value decomposition of $\widehat{\Sigma}_n$; this can be more computationally intensive (with computation time $O(J_n^3)$) than computing trace and Frobenius norm (with computation time $O(J_n^2)$). Second, the null distribution must be simulated, which introduces additional Monte Carlo computation and simulation error. In contrast, the standardized test avoids both the singular value decomposition and the Monte Carlo approximation. Therefore, we will primarily focus on the procedure based on the standardized $\widehat{S}_n$.}

Theorem \ref{Theorem: asymptotic distribution under the null} characterizes the asymptotic distribution of the standardized test statistic under the null hypothesis and outlines how to construct an asymptotically valid size-$\alpha$ test.

\begin{theorem}[Uniform type I error control of the standardized test]
Under Assumptions \ref{Assumption: basis} to \ref{Assumption: nuisance}, and two additional conditions: (1) $\sum_{j=1}^{J_n}\lambda_{j,P} \leq c_n $, where $ c_n = o(n^{1/3})$, $J_n \rightarrow \infty$; and (2) $1/C \leq \min\{\lambda_{1,P},...,\lambda_{J_n,P}\}$ for some constant $C$, we have
\begin{align*}
    \sup_{P\in \mathcal{P}_0}\sup_{t\in \mathbb{R}}\left\vert P\left(\frac{\widehat{S}_n-\widehat{\rho}_n}{{2}^{1/2}\widehat{\gamma}_n} \leq t\right)-\Phi(t)\right\vert\rightarrow 0,
\end{align*}
where $\Phi(t)$ denotes the standard normal CDF. Therefore, the test \[
T_{\alpha,n} \equiv \mathbbm{1}\left\{\frac{\widehat{S}_n-\widehat{\rho}_n}{{2}^{1/2}\widehat{\gamma}_n}>z_{1-\alpha}\right\}
\]
is asymptotically uniformly size-$\alpha$ over $\mathcal{P}_0$.
\label{Theorem: asymptotic distribution under the null}
\end{theorem}

Theorem \ref{Theorem: asymptotic distribution under the null} applies to general function-valued statistical parameters, whereas the results of \citet{breunig2015goodness} and \citet{hong1995consistent} focus specifically on regression and instrumental variable regression settings. Another key distinction is that \citet{breunig2015goodness} and \citet{hong1995consistent} rely on parametric models or series estimators for nuisance estimation, while our framework accommodates flexible, potentially high-dimensional machine learning methods for nuisance estimation. Furthermore, our theory allows for a data-adaptive weighting matrix $\widehat{\Omega}_n$. From a technical perspective, our theorem establishes uniform control of the type I error rate over $\mathcal{P}_0$ by leveraging a Berry--Esseen-type bound for martingale sums \citep{haeusler1988rate}, rather than the martingale central limit theorem used in prior work; see Supplementary Materials S2 and S3 for details. As a final remark, conditions (1) and (2) in Theorem \ref{Theorem: asymptotic distribution under the null} imply $J_n = o(n^{1/3})$.

\subsection{Limiting behavior under fixed and local alternatives}
\label{subsec: theory power}
We next study the limiting behavior of the proposed test $T_{\alpha, n}$ in Theorem \ref{Theorem: asymptotic distribution under the null} under a sequence of local alternatives $P_{1,n}\in \mathcal{P}$ of the following form:
\begin{align}\label{eq: local alternatives}
    \mathbbm{E}_{P_{1,n}}[g^\perp(O;\eta_{P_{1,n}})\mid X] = ({\gamma_{n,P_{1,n}}}/n)^{1/2}\phi(X),
\end{align}
where $\mathbbm{E}_{P_{1,n}}[\vert \phi(X)\vert^4]<C$ for some constant $C$. Theorem \ref{Theorem: local power analysis} derives the asymptotic distribution of the standardized test statistic under this local alternative.

\begin{theorem}
    Under the Assumptions of Theorem \ref{Theorem: asymptotic distribution under the null}, suppose further that $\mathbb{E}_{P_{1,n}}[\phi(X)\tilde{b}^{\lambda}_j(X)] < \infty$. Then, we have:
    \begin{align*}
        \frac{\widehat{S}_n-\widehat{\rho}_n}{{2}^{1/2}\widehat{\gamma}_n} -1/{2}^{1/2}\sum_{j=1}^{J_n} (\mathbbm{E}_{P_{1,n}}[\phi(X)\tilde{b}^{\lambda}_j(X)] )^2\rightsquigarrow N\left(0,1\right).
    \end{align*}    
    \label{Theorem: local power analysis}
\end{theorem}
\noindent Suppose that $\Xi_n :=\sum_{j=1}^{J_n} (\mathbbm{E}_{P_{1,n}}[\phi(X)\tilde{b}^{\lambda}_j(X)] )^2$ converges to some positive value $\Xi$, then by Slutsky’s theorem we know that $\frac{\widehat{S}_n-\widehat{\rho}_n}{{2}^{1/2}\widehat{\gamma}_n} \rightsquigarrow N\left(1/{2}^{1/2}\Xi,1\right)$. Suppose $\phi(X)$ is an element of the space spanned by $(b_j(X))_{j=1}^\infty$, since $\sum_{j=1}^{J_n} (\mathbbm{E}_{P_{1,n}}[\phi(X)\tilde{b}^{\lambda}_j(X)] )^2$ is always positive as long as $\phi(X)\neq 0$ almost surely, Theorem \ref{Theorem: local power analysis} implies that the proposed test exhibits nontrivial asymptotic power under local alternatives of the form \eqref{eq: local alternatives}. In particular, its asymptotic power under such alternatives exceeds that of a trivial test that rejects the null with probability $\alpha$.

Finally, we establish the uniform consistency result under a set of alternatives. To this end, define the collection of distributions
\begin{align*}
    \mathcal{P}_{1,n}(q) = \left\{P:\sum_{j=1}^{J_n}(\mathbbm{E}_{P}[g^\perp(O;\eta_P)\tilde{b}^{\lambda}_j(X)])^2>q\gamma_{n,P}/n\right\},
\end{align*}
where $\gamma_{n,P} = \Vert \Sigma_{n,P} \Vert_F$. Our proposed test $T_{\alpha, n}$ in Theorem \ref{Theorem: asymptotic distribution under the null} has asymptotic power one under any fixed alternative distribution in $\mathcal{P}_{1,n}(q)$, as formalized in Theorem \ref{Theorem: consistency}.
{\color{black}
\begin{theorem}
Under the Assumptions of Theorem \ref{Theorem: asymptotic distribution under the null}, we have that for any $\epsilon \in (0,1)$ and $\alpha\in (0,1)$, there is a $q$ such that
    \begin{align*}
    \liminf_{n\rightarrow \infty}\inf_{P\in\mathcal{P}_{1,n}(q)} P \left(\frac{\widehat{S}_n-\widehat{\rho}_n}{{2}^{1/2}\widehat{\gamma}_n} >z_{1-\alpha}\right)\geq 1-\epsilon.
\end{align*}

Let $q_n$ be a sequence such that $q_n\rightarrow \infty$, then we have
\begin{align*}
    \lim_{n\rightarrow \infty}\inf_{P\in\mathcal{P}_{1,n}(q_n)} P \left(\frac{\widehat{S}_n-\widehat{\rho}_n}{{2}^{1/2}\widehat{\gamma}_n} >z_{1-\alpha}\right)=1.
\end{align*}
\label{Theorem: consistency}
\end{theorem}
}

\section{Testing structural assumptions in causal inference: two examples}\label{sec: example}
\label{sec: examples}
%We demonstrate how to apply the proposed testing framework to the causal inference literature by examining two prominent use cases.
\subsection{Testing the mean exchangeability assumption in data fusion problems}\label{subsec: example mean exchangeability}
We consider the problem of testing the mean exchangeability assumption in the data fusion literature \citep{dahabreh2019generalizing}. Let $Y(a)$ denote the potential outcome under a binary treatment $A = a \in \{0, 1\}$, $X$ the vector of observed covariates, and $S$ an indicator of the data source.  For instance, $S=1$ may indicate an observational study and $S=0$ a randomized clinical trial. Combining data from these two sources can improve the efficiency of average causal effect estimation compared to using either source alone, provided certain alignment conditions hold \citep{li2023efficient}. One such condition is the following mean exchangeability assumption:
\begin{align*}
    \mathbb{E}_{P}[Y(a)\mid X,S=1] =  \mathbb{E}_{P}[Y(a)\mid X,S=0], ~ a\in \{0,1\},
\end{align*}
which states that the conditional means of the potential outcomes are equal across two data sources. Assuming no unmeasured confounding within each data source, i.e., $Y(a)\indep A\mid X,S=s,~s\in\{0,1\}$, $\mathbb{E}_{P}[Y(a)\mid X,S=1]$ can be identified from the observed data as $\mathbb{E}_{P}[Y\mid A=a,X,S=s]$. {\textcolor{black}{Thus, testing the mean exchangeability assumption reduces to testing the equality of two conditional means functions:} $\mathbb{E}_{P}[\tilde{g}(O;\tilde{\eta}_{P})\mid X]=0$, with $\tilde{g}(O;\eta_{P}) = \mathbb{E}_{P}[Y\mid A=a,X,S=1]-\mathbb{E}_{P}[Y\mid A=a,X,S=0]$, and $\tilde{\eta}_{P} = \{\mathbb{E}_{P}[Y\mid A=a,X,S=1],\mathbb{E}_{P}[Y\mid A=a,X,S=0]\}$.

Note that $\tilde{g}$ above does not satisfy the conditional Neyman orthogonality condition. As a result, it is necessary to construct an orthogonalized version so that nuisance functions can be estimated with flexible machine learning tools. The orthogonalized version of $\tilde{g}$, denoted ${g}_{\text{ME}}^{\perp}(O;\eta_{P})$, is:
\begin{equation}\label{eq: g for mean exchangeability}
\begin{split}
   {g}_{\text{ME}}^{\perp}(O;\eta_{P}) = &\frac{I\{A=a,S=1\}}{P(A=a,S=1\mid X)}\left(Y-\mathbb{E}_{P}[Y\mid A=a,X,S=1]\right)+\mathbb{E}_{P}[Y\mid A=a,X,S=1]\\
    -&\frac{I\{A=a,S=0\}}{P(A=a,S=0\mid X)}\left(Y-\mathbb{E}_{P}[Y\mid A=a,X,S=0]\right)-\mathbb{E}_{P}[Y\mid A=a,X,S=0],
    \end{split}
\end{equation}
where the nuisance functions consist of four elements: \[
\eta_{P} = \{P(A=a,S=1\mid X),P(A=a,S=0\mid X),\mathbb{E}_{P}[Y\mid A=a,X,S=1],\mathbb{E}_{P}[Y\mid A=a,X,S=0]\}.
\]
%The orthogonalized version still satisfies $\mathbb{E}_{P}[g^\perp_{\text{ME}}(O;{\eta}_{P})\mid X] =\mathbb{E}_{P}[\tilde{g}(O;{\eta}_{P})\mid X]= 0$ under the null hypothesis, so the mean exchangeability assumption can be tested using methods proposed in Section \ref{subsec: Algorithm} with $g(O;\eta_{P}) = {g}_{\text{ME}}^{\perp}(O;\eta_{P})$. 
\noindent Sufficient conditions for Assumption \ref{Assumption: nuisance} to hold are given in Supplementary Material \ref{sec: revisiting examples, supp}. {\color{black}Some other null hypotheses can also be formulated as the equality of two conditional counterfactual mean functions. For example, testing whether the conditional average treatment effect is zero can be written as $\mathbbm{E}_P[Y(1)|X] = \mathbbm{E}_P[Y(0)|X]$. Similar methods can be applied in such settings.}

%{\color{black}As a final remark, other causal testing problems may share the same statistical structure as the mean exchangeability example. For instance, one may be interested in testing whether a conditional treatment effect is zero, which amounts to testing equality of two conditional potential-outcome means. The first real-data application in Section \ref{sec: application 1} provides an example of such a problem.}

\subsection{Testing a compatibility condition in instrumental variable problems}\label{subsec: example IV}
Let $Z = (Z_1,Z_2)$, $Z_1,Z_2 \in \{0,1\}$, be a vector of binary instrumental variables. Let $X$ denote the vector of baseline covariates, $D$ the treatment uptake, and $Y$ the outcome of interest. Under standard IV assumptions, the conditional average treatment effect among $Z_1$-compliers, defined as those with $\{D(Z_1 = 0)=0, D(Z_1 = 1)=1\}$, can be identified as follows \citep{angrist1996identification,frolich2007nonparametric,angrist2010extrapolate}:
\[
\Psi_{Z_1,P}(X) = \frac{\mathbb{E}_{P}[Y\mid Z_1=1,X]-\mathbb{E}_{P}[Y\mid Z_1=0,X]}{\mathbb{E}_{P}[D\mid Z_1=1,X]-\mathbb{E}_{P}[D\mid Z_1=0,X]}.
\]

The conditional average treatment effect among $Z_2$-compliers can be analogously identified: 
\[
\Psi_{Z_2,P}(X) = \frac{\mathbb{E}_{P}[Y\mid Z_2=1,X]-\mathbb{E}_{P}[Y\mid Z_2=0,X]}{\mathbb{E}_{P}[D\mid Z_2=1,X]-\mathbb{E}_{P}[D\mid Z_2=0,X]}.
\]

\citet{angrist2010extrapolate} discussed the \emph{compatibility condition}: $\Psi_{Z_1,P}(X) = \Psi_{Z_2,P}(X).$ If the compatibility condition holds, then the difference between the complier average treatment effect under $Z_1$ and that under $Z_2$ is driven solely by differences in the observed characteristics of
compliers \citep{angrist2010extrapolate}.

Testing the compatibility assumption can be reduced to testing $\mathbb{E}_{P}[\tilde{g}(O;\tilde{\eta}_{P})\vert \boldsymbol{X}] = 0$, with $\tilde{g}(O;\tilde{\eta}_{P}) = \Psi_{Z_1,P}(X) - \Psi_{Z_2,P}(X)$ and $\tilde{\eta}_P = \{\mathbb{E}_{P}[D\mid Z_1 = z_1,X], \mathbb{E}_{P}[Y\mid Z_1=z_1,X], \mathbb{E}_{P}[D\mid Z_2 = z_2,X], \mathbb{E}_{P}[Y\mid Z_2=z_2,X],~z_1,z_2 \in \{0, 1\}\}$. Using the influence function results in \citet{wang2024nested}, the orthogonalized version of $\tilde{g}(O;\tilde{\eta}_{P})$, denoted as $\tilde{g}_{\text{COM}}^\perp(O;\eta_{P})$, is ${g}_{\text{COM}}^\perp(O;\eta_{P}) = g_{Z_1}(O;\eta_{P})-g_{Z_2}(O;\eta_{P})$,
where, for $j\in \{1,2\}$,
{\small\begin{align*}
    g_{Z_j}(o;\eta_{P})&=
(z,\boldsymbol{x},d,y)\mapsto\frac{1}{\pi_{j,P}(\boldsymbol{x})}\Bigg\{\frac{1\{z_j=1\}}{P(Z_j=1\mid\boldsymbol{X} = \boldsymbol{x})}\Big[y-\mathbb{E}_P[Y\mid Z_j=1,\boldsymbol{X} = \boldsymbol{x}]\Big] \\
        &-\frac{1\{z_j=0\}}{P(Z_j=0\mid\boldsymbol{X} = \boldsymbol{x})}\Big[y-\mathbb{E}_P[Y\mid Z_j = 0,\boldsymbol{X} = \boldsymbol{x}]\Big]\Bigg\}\\
        &- \frac{\mu_{j,P}(\boldsymbol{x})}{[\pi_{j,P}(\boldsymbol{x})]^2}\Bigg\{\frac{1\{Z_j = 1\}}{P(Z_j=1\mid\boldsymbol{X} = \boldsymbol{x})}\Big[d-\mathbb{E}_P[D\mid Z_j=1,\boldsymbol{X} = \boldsymbol{x}]\Big] \\
        &-\frac{1\{Z_j=0\}}{P(Z_j=0\mid\boldsymbol{X} = \boldsymbol{x})}\Big[d-\mathbb{E}_P[D\mid Z_j=0,\boldsymbol{X} = \boldsymbol{x}]\Big]\Bigg\}+\Psi_{Z_j}(X),\\
       \pi_{j,P} & = \mathbb{E}_{P}[D\mid Z_j = 1,X=x] - \mathbb{E}_{P}[D\mid Z_j = 0,X=x], \\
       \mu_{j,P} & = \mathbb{E}_{P}[Y\mid Z_j = 1,X=x] - \mathbb{E}_{P}[Y\mid Z_j = 0,X=x],
\end{align*}
}
and the nuisance functions consist of six elements: 
\[
\eta_{P} = \left\{\mathbb{E}_{P}[D\mid Z_j=z,X=x],~~\mathbb{E}_{P}[Y\mid Z_j=z,X=x], ~~\mathbb{E}_{P}[Z_j=z\mid X=x], j \in \{1,2\}\right\}.
\]

%\noindent\citeauthor{angrist2010extrapolate}'s \citeyearpar{angrist2010extrapolate} compatibility condition can then be tested using methods proposed in Section \ref{subsec: Algorithm} with $g(O;\eta_{P}) = {g}_{\text{COM}}^{\perp}(O;\eta_{P})$.  
\noindent Sufficient conditions for Assumption \ref{Assumption: nuisance} to hold are given in Supplementary Material \ref{sec: revisiting examples, supp}.

\section{Simulation}
\label{sec: simulation}

\subsection{Goal and structure}
\label{subsec: simulation overview}
{\color{black} We evaluate the size and power of the proposed generalized projection tests under different choices of the sieve dimension $J_n$ in finite-sample settings. We specify the basis function for each covariate separately and use ${J}^*_n$ to denote the number of basis function we use for each continuous covariate. In addition, we evaluate the size and power of an adaptive test that combines multiple tests with sieve dimensions on an exponentially spaced grid, and use ${\mathcal{J}}^*_n$ to denote the collection of candidate sieve dimensions for each continuous covariate. Throughout the simulation, we set $\mathcal{J}^*_n = \{J=\underline{J}^*_n, \underline{J}^*_n 2, \ldots, \underline{J^*}_n 2^{j_{\max}}\}$, with $\underline{J}^*_n = \lfloor 2 \sqrt{\log n} \rfloor$ and $\underline{J^*}_n 2^{j_{\max}} < \lfloor 6 n^{1/4}\rfloor$.} Section \ref{subsec:simulation mean exchangeability} and \ref{subsec:simulation IV} examine the performance of our proposed tests for testing the mean exchangeability assumption and the IV compatibility condition, respectively. {\color{black} For both simulation studies, we report results for the standardized tests using the Fourier basis in the main article. Parallel results for the unstandardized tests and based on the orthonormal Legendre basis are relegated to the Supplemental Material \ref{subsec:additional simulation}.}

\subsection{Testing the mean exchangeability assumption}
\label{subsec:simulation mean exchangeability}
We generate data from $N$ study participants. The first factor we vary is the sample size:
\begin{description}
    \item[Factor 1: Sample size.] $N = 250$, $500$, $1000$, and $1500$.
\end{description}
For each participant, we generate two observed covariates $(X_1, X_2)$, both are independently drawn from Uniform[-1, 1]. Each participant is from one of the two data sources, e.g., observational database or randomized controlled trial. We let $S = 1$ or $0$ denote the data source membership, and the binary indicator $S$ is generated according to
$P(S=1 \mid X) = \text{expit}(X_1-X_2).$
Then, within each data source, a participant is assigned to treatment ($A = 1$) or control ($A = 0$), and the treatment assignment follows $P(A=1\mid X,S) = S\text{expit}(1.5X_1-0.5X_2)+(1-S)\text{expit}(X_1+0.5X_2)$
in either data source. 

In this simulation study, we only consider testing $\mathbbm{E}_{P}[Y(0)\mid X,S=0] = \mathbbm{E}_{P}[Y(0)\mid X,S=1].$ Testing $\mathbbm{E}_{P}[Y(1)\mid X,S=0] = \mathbbm{E}_{P}[Y(1)\mid X,S=1]$ is analogous and hence omitted. We do not consider testing the mean exchangeability of the causal contrast, that is, $\mathbbm{E}_{P}[Y(1)-Y(0)\mid X,S=0] = \mathbbm{E}_{P}[Y(1)-Y(0)\mid X,S=1],$ because the method of \citet{racine2006testing} does not apply.

Lastly, for each participant, we simulate two potential outcomes, $Y(0)$ and $Y(1)$, as follows:
\begin{align*}
    Y(0)& = \mathbbm{1}\{S=0\}\cdot(X_1+X_2+\text{expit}(X_1))
    \\
    &+\mathbbm{1}\{S=1\}\cdot\left\{X_1+X_2+\text{expit}(X_1)+\alpha_1\{\cos(\pi X_1)+\cos(\pi X_2)\} + \alpha_2(X_1+X_2)\right\}+0.5\epsilon,\\
    Y(1)& = Y(0) + 2X_1 - 2X_2,
\end{align*}
where $\epsilon \sim N(0, 1)$ is a standard normal error. The observed outcome $Y$ satisfies $Y = A \cdot Y(1) + (1-A)\cdot Y(0)$ in either data source.

The second factor we vary is $(\alpha_1,\alpha_2)$, which controls the degree of alignment in $Y(0)$ between the two sources, $S = 1$ and $S = 0$:
\begin{description}
    \item[Factor 2: Degree of misalignment $(\alpha_1,\alpha_2)$.] Scenario I: $(\alpha_1,\alpha_2) = (0, 0)$ corresponds to perfect alignment, i.e., the null hypothesis holds; Scenario II: $(\alpha_1,\alpha_2) = (0.2, 0)$ corresponds to a weak misalignment, with a nonlinear difference; Scenario III: $(\alpha_1,\alpha_2) = (0, 0.2)$ corresponds to a weak misalignment, with a linear difference; Scenario IV: $(\alpha_1,\alpha_2) = (0.4, 0.2)$ corresponds to a strong misalignment, with a mostly nonlinear difference; Scenario V: $(\alpha_1,\alpha_2) = (0.2, 0.4)$ corresponds to a strong misalignment, with a mostly linear difference.
\end{description}
\noindent For each of the $4 \times 5 = 20$ data-generating processes, we generate $1000$ datasets and calculate the proportion of times the null hypothesis $\mathbbm{E}_{P}[Y(0)\mid X,S=0] = \mathbbm{E}_{P}[Y(0)\mid X,S=1]$ is rejected based on the generalized projection (GP) test described in Section \ref{subsec: example mean exchangeability}. When implementing the tests, we estimate all nuisance parameters using the \texttt{SuperLearner} package in \textsf{R} \citep{Polley2024superlearner}. For estimating the models of $Y$, we include both \texttt{glm} and \texttt{randomForest} in the \texttt{SuperLearner} library. For estimating the models of $(A,S)$, we only use \texttt{glm} to maintain the stability of inverse probability weights when the sample size is small. {\color{black} We compare the GP test---using different choices of the sieve dimension (denoted $J^\ast$), as well as an adaptive combination across dimensions---with the fixed-dimension projection test with $\{\gamma(x;\theta) = x^\top\theta,~\theta \in \mathbb{R}^d\}$ (referred to as the linear projection test, or LP test), and the conditional independence test of \citet{racine2006testing} (implemented in the \textsf{R} package \textsf{np}) .}

{\color{black}
Panel A of Table~\ref{tab:simulation_results} summarizes the simulation results for standardized GP tests using the Fourier basis. The performance of the GP test depends on both the sample size $N$ and the number of basis functions $J^\ast$ for each covariate. When $N$ is relatively small (e.g., $N = 250$), $J^\ast = 4$ provides good type I error control, whereas $J^\ast \in \{8, 16\}$ appears conservative. When $N$ is larger (e.g., $N = 1000$), $J^\ast=5 $ yields slightly inflated type I error, while $J^\ast = 10$ continues to provide good type I error control and $J^\ast = 20$ is  conservative. These results suggest that, in general, using more basis functions improves type I error control at the cost of reduced power. The adaptive test that combines multiple GP tests maintains good type I error control across almost all data-generating processes while preserving good power, making it our recommended default test. 

Between the two projection-based tests, when the misalignment is only nonlinear (Scenario II), the LP test has minimal power, whereas the GP test remains consistent, with its power approaching one as $N$ increases. When the misalignment is mostly linear (Scenarios III and V), the LP test exhibits higher power. When the misalignment contains both linear and nonlinear components but is predominantly nonlinear (Scenario IV), the GP test is more powerful than the LP test. 

Finally, the conditional independence test of \citet{racine2006testing} appears slightly anti-conservative, consistent with previous findings reported in \citet{racine2006testing} and \citet{luedtke2019omnibus}. In terms of power, the \citet{racine2006testing} test is either more powerful than or comparable to the projection tests across the simulation settings considered here.

In Supplemental Material \ref{subsec:high-dimensional simulation}, we conduct additional simulations with a 10-dimensional 
$X$. The GP test continues to control the type I error rate, although it becomes conservative. In Supplemental Material \ref{subsec:computation time}, we compare computation times across methods and settings; the GP test scales well in both covariate dimension and sample size.
}

\subsection{Testing the IV compatibility condition}
\label{subsec:simulation IV}

We generate data from $N$ study participants. The first factor we vary is the sample size:
\begin{description}
    \item[Factor 1: Sample size.] $N = 1000$, $2000$, $3000$ and $5000$.
\end{description}
For each participant, we generate two observed covariates $(X_1, X_2)$, both independently drawn from Uniform[-1, 1]. We then generate two binary instrumental variables, $Z_1$ and $Z_2$ independently of the other as follows: 
{
\small
\begin{align*}
    &P(Z_1 = 1\mid X) = \text{expit}(0.5+0.5I\{X_1>0\}-0.5\{X_2>0\}),\\
    &P(Z_2 = 1\mid X) = \text{expit}(0.5+0.5I\{X_1>0\}+0.5\{X_2>0\}).
\end{align*}}
\noindent Next, We generate an unmeasured confounder $U \sim N(-0.3,0.3^2)$. Each participant belongs to one of the following $5$ principal strata defined by $D(Z_1, Z_2)$ \citep{mogstad2021causal}: Always-Never-Takers (ANT): $\{D(Z_1=0,Z_2=0)=D(Z_1=1,Z_2=0)=D(Z_1=0,Z_2=1)=D(Z_1=1,Z_2=1)=0\}$; Strict $Z_1$-Compliers (SCO1): $\{D(Z_1=0,Z_2=0)=D(Z_1=0,Z_2=1)=0,D(Z_1=1,Z_2=0)=D(Z_1=1,Z_2=1)=1\}$; Strict $Z_2$-Compliers (SCO2): 
$\{D(Z_1=0,Z_2=0)=D(Z_1=1,Z_2=0)=0,D(Z_1=0,Z_2=1)=D(Z_1=1,Z_2=1)=1\}$; reluctant compliers (RCO): $\{D(Z_1=0,Z_2=0)=D(Z_1=1,Z_2=0)=D(Z_1=0,Z_2=1)=0,D(Z_1=1,Z_2=1)=1\}$; and eager compliers (ECO): $\{D(Z_1=0,Z_2=0)=0,D(Z_1=1,Z_2=0)=D(Z_1=0,Z_2=1)=D(Z_1=1,Z_2=1)=1\}$. 

We generate each participant's principal stratum membership $S$ as follows.  Let $X^*_1 = 1\{X_1>0\}$ and $X^*_2 = 1\{X_2>0\}$. Define
{\begin{equation*}
\begin{split}
    &g_{\text{ANT}} = \exp(1-X^*_2+0.3\mathbbm{1}\{U>0\}), \quad g_{\text{SCO1}} = \exp(3.5+0.5X_1^*+X^*_2+0.3\mathbbm{1}\{U>0\}),\\
    &g_{\text{SCO2}} = \exp(3.5+0.5X_1^*+X^*_2+0.3\mathbbm{1}\{U>0\}), \quad
    g_{\text{RCO}} = \exp(2+X_1^*+X^*_2+0.3\mathbbm{1}\{U>0\}),\\
    &g_{\text{ECO}} = \exp(2+X_1^*+0.5X^*_2+0.3\mathbbm{1}\{U>0\}).
    \end{split}
\end{equation*}}
We set $P(S_i = j) = \frac{g_j}{\sum_{j}g_j},$ $j\in \{\text{ANT}, \text{SCO1}, \text{SCO2}, \text{RCO},\text{ECO}\}$.

The instrumental variable assignment $(Z_1, Z_2)$ and the principal stratum membership $S$ together determine the treatment received $D$ for each participant. Finally, we simulate the potential outcomes, $Y(0)$ and $Y(1)$, based on each participant's principal stratum membership as follows:
\begin{equation*}
\begin{split}
Y(0) = &1+X_1+X_2+U+\epsilon, \\
    Y(1) = &\mathbbm{1}\{S\in\{ \text{ANT}\}\}\cdot(1+X_1+X_2+U)+\mathbbm{1}\{S\in\{ \text{SCO1}, \text{RCO}, \text{ECO}\}\}\cdot(1-X_1+X_2+U)\\
    &+\mathbbm{1}\{S= \text{SCO2}\}\cdot\left(1-X_1+X_2+ U + \beta_1\left\{\cos(\pi X_1)+\cos(\pi X_2)\right\}+\beta_2(X_1+X_2)\right\}+\epsilon,
\end{split}
\end{equation*}
where $\epsilon \sim N(0, 1)$. Under this data-generating process, the average treatment effect varies across principal strata. The second factor we vary is $(\beta_1, \beta_2)$, which controls the extent to which the compatibility condition is violated---that is, the degree to which the conditional average treatment effect among $Z_1$-compliers differs from that among  $Z_2$-compliers:
\begin{description}
    \item[Factor 2: Degree of misalignment $(\beta_1,\beta_2)$.] Scenario I: $(\beta_1,\beta_2) = (0, 0)$, that is, the compatibility condition holds; Scenario II: $(\beta_1,\beta_2) = (0.3, 0)$ corresponds to a mild violation of the compatibility condition, with a nonlinear difference between two conditional principal effects;  Scenario III: $(\beta_1,\beta_2) = (0, 0.3)$ corresponds to a mild violation of the compatibility condition, with a linear difference between two conditional principal effects; Scenario IV: $(\beta_1,\beta_2) = (0.6, 0.3)$ corresponds to a strong violation, with a mostly nonlinear difference between two conditional principal effects; Scenario V: $(\beta_1,\beta_2) = (0.3, 0.6)$ corresponds to a strong violation, with a mostly linear difference between two conditional principal effects.
\end{description}

For each of the $4 \times 5 = 20$ data-generating processes, we generate 1000 datasets. 
We use  \texttt{glm}, \texttt{randomForest} and \texttt{gam} in \texttt{SuperLearner} library for estimating for the models of $Y$ and use \texttt{glm} and \texttt{randomForest} \texttt{SuperLearner} library for estimating for the models of $Z$ and $D$.  The generalized projection test is compared with the fixed-dimension projection test, 
but not with the conditional independence test of \citet{racine2006testing}, as the 
latter is not applicable in this setting.

{\color{black}
Panel B of Table~\ref{tab:simulation_results} summarizes the simulation results for standardized GP tests using the Fourier basis. We observe patterns similar to those in Panel A. In general, the GP tests, when the number of basis functions is not too small relative to the sample size, and the adaptive test that combines GP tests across sieve dimensions appear to have good type I error control in most cases. In terms of power, GP tests tend to outperform LP tests when the alternative has a strong nonlinear components, but underperform when the alternative is mostly linear. Again, similar to results in Panel A, the GP tests remain powerful across all types of alternatives considered here, while the LP test has essentially no power when the alternative is purely nonlinear. Additional simulations with a 10-dimensional $X$ and a comparison of computational cost are reported in Supplemental Materials \ref{subsec:high-dimensional simulation} and \ref{subsec:computation time}.
}

\begin{table}[htbp]
\centering
\tiny
\caption{%
\textbf{Panel A}: Proportion of times the mean exchangeability assumption was
rejected across different parameter settings, based on the fixed-dimension, linear
projection (LP) test, the generalized projection (GP) tests, and the conditional
independence test of \citet{racine2006testing}.
\textbf{Panel B}: Proportion of times the IV compatibility condition was rejected
across different parameter settings, based on the LP test and the GP tests.
For the GP test, results correspond to the standardized tests using the Fourier basis.
$J^*$ denotes the number of basis functions for each covariate. Scenario I corresponds
to the null hypothesis case; Scenario II corresponds to a weak nonlinear misalignment;
Scenario III corresponds to a weak linear misalignment; Scenario IV corresponds to a
strong nonlinear misalignment; Scenario V corresponds to a strong linear misalignment.
}
\label{tab:simulation_results}
\resizebox{\textwidth}{!}{%
\begin{tabular}{cclccccc}
\toprule
& \begin{tabular}[c]{@{}c@{}}Sample \\ size\end{tabular}
& \begin{tabular}[c]{@{}c@{}}Testing \\ method\end{tabular}
& \begin{tabular}[c]{@{}c@{}}Scenario \\ I\end{tabular}
& \begin{tabular}[c]{@{}c@{}}Scenario \\ II\end{tabular}
& \begin{tabular}[c]{@{}c@{}}Scenario \\ III\end{tabular}
& \begin{tabular}[c]{@{}c@{}}Scenario \\ IV\end{tabular}
& \begin{tabular}[c]{@{}c@{}}Scenario \\ V\end{tabular} \\
\midrule
\multirow{24}{*}{Panel A}
& \multirow{6}{*}{250}
& LP test & 0.069 & 0.089 & 0.246 & 0.233 & 0.592 \\
& & GP test ($J^*=4$) & 0.051 & 0.128 & 0.088 & 0.486 & 0.325 \\
& & GP test ($J^*=8$) & 0.029 & 0.065 & 0.060 & 0.298 & 0.194 \\
& & GP test ($J^*=16$) & 0.012 & 0.019 & 0.015 & 0.112 & 0.087 \\
& & GP test (Combined) & 0.025 & 0.083 & 0.056 & 0.382 & 0.234 \\
& & Racine & 0.092 & 0.206 & 0.367 & 0.833 & 0.909 \\
\cmidrule{2-8}
& \multirow{6}{*}{500}
& LP test & 0.062 & 0.059 & 0.428 & 0.393 & 0.922 \\
& & GP test ($J^*=4$) & 0.052 & 0.319 & 0.160 & 0.916 & 0.809 \\
& & GP test ($J^*=8$) & 0.043 & 0.181 & 0.108 & 0.818 & 0.659 \\
& & GP test ($J^*=16$) & 0.024 & 0.087 & 0.061 & 0.620 & 0.447 \\
& & GP test (Combined) & 0.039 & 0.230 & 0.117 & 0.875 & 0.726 \\
& & Racine & 0.081 & 0.332 & 0.610 & 0.990 & 0.996 \\
\cmidrule{2-8}
& \multirow{6}{*}{1000}
& LP test & 0.065 & 0.062 & 0.739 & 0.695 & 0.996 \\
& & GP test ($J^*=5$) & 0.067 & 0.690 & 0.374 & 0.999 & 0.997 \\
& & GP test ($J^*=10$) & 0.045 & 0.505 & 0.257 & 0.996 & 0.989 \\
& & GP test ($J^*=20$) & 0.024 & 0.294 & 0.155 & 0.973 & 0.932 \\
& & GP test (Combined) & 0.053 & 0.592 & 0.283 & 0.996 & 0.994 \\
& & Racine & 0.082 & 0.696 & 0.871 & 1.000 & 1.000 \\
\cmidrule{2-8}
& \multirow{6}{*}{1500}
& LP test & 0.061 & 0.057 & 0.911 & 0.892 & 1.000 \\
& & GP test ($J^*=5$) & 0.064 & 0.903 & 0.561 & 1.000 & 1.000 \\
& & GP test ($J^*=10$) & 0.044 & 0.751 & 0.478 & 1.000 & 1.000 \\
& & GP test ($J^*=20$) & 0.039 & 0.550 & 0.299 & 0.999 & 0.996 \\
& & GP test (Combined) & 0.041 & 0.846 & 0.495 & 1.000 & 1.000 \\
& & Racine & 0.079 & 0.897 & 0.968 & 1.000 & 1.000 \\
\midrule
\multirow{22}{*}{Panel B}
& \multirow{5}{*}{1000}
& LP test & 0.039 & 0.036 & 0.072 & 0.076 & 0.157 \\
& & GP test ($J^*=5$) & 0.063 & 0.077 & 0.073 & 0.204 & 0.129 \\
& & GP test ($J^*=10$) & 0.060 & 0.062 & 0.061 & 0.116 & 0.108 \\
& & GP test ($J^*=20$) & 0.047 & 0.046 & 0.039 & 0.074 & 0.072 \\
& & GP test (Combined) & 0.057 & 0.054 & 0.060 & 0.162 & 0.111 \\
\cmidrule{2-8}
& \multirow{5}{*}{2000}
& LP test & 0.048 & 0.042 & 0.123 & 0.126 & 0.410 \\
& & GP test ($J^*=5$) & 0.064 & 0.129 & 0.110 & 0.550 & 0.354 \\
& & GP test ($J^*=10$) & 0.058 & 0.099 & 0.082 & 0.408 & 0.294 \\
& & GP test ($J^*=20$) & 0.053 & 0.070 & 0.060 & 0.263 & 0.183 \\
& & GP test (Combined) & 0.063 & 0.096 & 0.082 & 0.464 & 0.308 \\
\cmidrule{2-8}
& \multirow{6}{*}{3000}
& LP test & 0.038 & 0.055 & 0.175 & 0.200 & 0.662 \\
& & GP test ($J^*=5$) & 0.056 & 0.220 & 0.132 & 0.765 & 0.588 \\
& & GP test ($J^*=10$) & 0.049 & 0.145 & 0.108 & 0.638 & 0.482 \\
& & GP test ($J^*=20$) & 0.041 & 0.101 & 0.090 & 0.476 & 0.347 \\
& & GP test ($J^*=40$) & 0.039 & 0.081 & 0.060 & 0.296 & 0.234 \\
& & GP test (Combined) & 0.047 & 0.158 & 0.105 & 0.691 & 0.506 \\
\cmidrule{2-8}
& \multirow{6}{*}{5000}
& LP test & 0.045 & 0.038 & 0.359 & 0.385 & 0.917 \\
& & GP test ($J^*=5$) & 0.056 & 0.349 & 0.190 & 0.968 & 0.868 \\
& & GP test ($J^*=10$) & 0.060 & 0.260 & 0.133 & 0.923 & 0.776 \\
& & GP test ($J^*=20$) & 0.057 & 0.190 & 0.109 & 0.799 & 0.654 \\
& & GP test ($J^*=40$) & 0.047 & 0.130 & 0.086 & 0.604 & 0.483 \\
& & GP test (Combined) & 0.055 & 0.263 & 0.146 & 0.957 & 0.811 \\
\bottomrule
\end{tabular}%
}
\end{table}

\section{Two case studies}

\subsection{COVID-19 incidence in the COVAIL study}
\label{sec: application 1}
The COVID-19 Variant Immunologic Landscape (COVAIL) trial [NCT05289037] is a Phase 2 randomized study \citep{branche2023comparison,branche2023immunogenicity}. The trial proceeded in four sequential stages, with most participants enrolling in Stage 1 (variant Moderna mRNA-1273 vaccine) or Stage 2 (variant Pfizer--BioNTech BNT162b2 vaccine); within each stage, participants were randomized to different variant formulations. COVAIL's primary objective is to evaluate the immunogenicity and safety of a second COVID-19 booster using monovalent or bivalent variant vaccines \citep{branche2023comparison,branche2023immunogenicity}. The trial also collected clinical endpoints enabling assessment of cumulative COVID-19 incidence across vaccine recipients. \citet{DIEMERT2025127718} recently reported that Omicron-containing boosters were more protective against COVID-19 than the Prototype vaccine among Stage 2 (Pfizer--BioNTech) recipients, but not among Stage 1 (Moderna) recipients. Figure \ref{fig: km curve} displays Kaplan-Meier curves for Omicron-containing and Prototype vaccine recipients within each stage.

%Figure \ref{fig: km curve} displays Kaplan-Meier curves for four groups: Stage 1 recipients of an Omicron-containing vaccine (Omicron + Beta, Omicron + Delta, Omicron, Omicron + Prototype; $n = 407$), Stage 1 recipients of the Prototype vaccine ($n = 97$), Stage 2 recipients of an Omicron-containing vaccine (Omicron + Beta, Omicron, Omicron + Prototype; $n = 156$), and Stage 2 recipients of the Prototype vaccine ($n = 47$).

{\color{black}
We assess whether the conditional probability of acquiring COVID-19 by Month 6 differed between Omicron-containing and Prototype vaccine recipients within each stage:
\[
P(Y(A = \text{Omicron-containing}) = 1 \mid X, \text{Stage} = s) - P(Y(A = \text{Prototype}) = 1 \mid X, \text{Stage} = s) = 0,
\]
for $s = 1$ and $ 2$, where $X$ includes: (1) a risk score from an ensemble learning method \citep{zhang2025neutralizing,fong2025neutralizing}, (2) a binary indicator for prior COVID-19 infection, (3) an external force of infection score, and (4) baseline neutralizing antibody levels against Omicron BA.1.

We applied the adaptive generalized projection test, combining candidate sieve dimensions $J^\ast \in \{4, 8, 16\}$, as described in Sections~\ref{subsec: example mean exchangeability} and \ref{subsec:simulation mean exchangeability}. Among Stage 2 participants, the conditional probability of acquiring COVID-19 differed significantly by vaccine insert (combined-test $p=0.002$; constituent $p$-values $<0.001$, $0.002$, and $0.021$ for $J^\ast=4,8,16$, respectively). Among Stage 1 participants, we failed to reject the null hypothesis (combined-test $p=1$; constituent $p$-values $0.578$, $0.550$, and $0.751$ for $J^\ast=4,8,16$, respectively). These results corroborate the findings of \citet{DIEMERT2025127718}, who studied the marginal, but not conditional, cumulative incidence of COVID-19 among Prototype- and Omicron-containing-vaccine recipients.}

%The number of basis functions was set to $J^\ast = 3,6,12$ for each continuous covariate, as this choice yielded the most favorable simulation results for a sample size of approximately $200$ (Panel A of Table \ref{tab: simulation results}).  The results indicate that the probability of acquiring COVID-19 differed significantly between participants with the same baseline covariates who received the Pfizer--BioNTech Prototype vaccine and those who received the Pfizer--BioNTech Omicron-containing vaccine  ($p<0.001$ for both $J^*=3$ and $J^*=6$ tests, $p=0.029$ for $J^*=12$ test, $p<0.001$ for the combined test). In contrast, we failed to reject the null hypothesis of equal conditional COVID-19 incidence among participants who received the Moderna Prototype vaccine and those who received the Moderna Omicron-containing vaccine ($p=0.686$ for $J^*=3$,$p=0.471$ for $J^*=6$, $p=0.555$ for $J^*=12$ test, $p = 1$ for the combined test). 

%Lastly, we tested the null hypothesis that the probability of acquiring COVID-19 differed among Omicron-based vaccinees between two stages. The null hypothesis was forcefully rejected based on our proposed generalized projection test ($p$-value $< 0.001$). 
\begin{figure}[ht]
    \centering
    \includegraphics[width=0.95\linewidth]{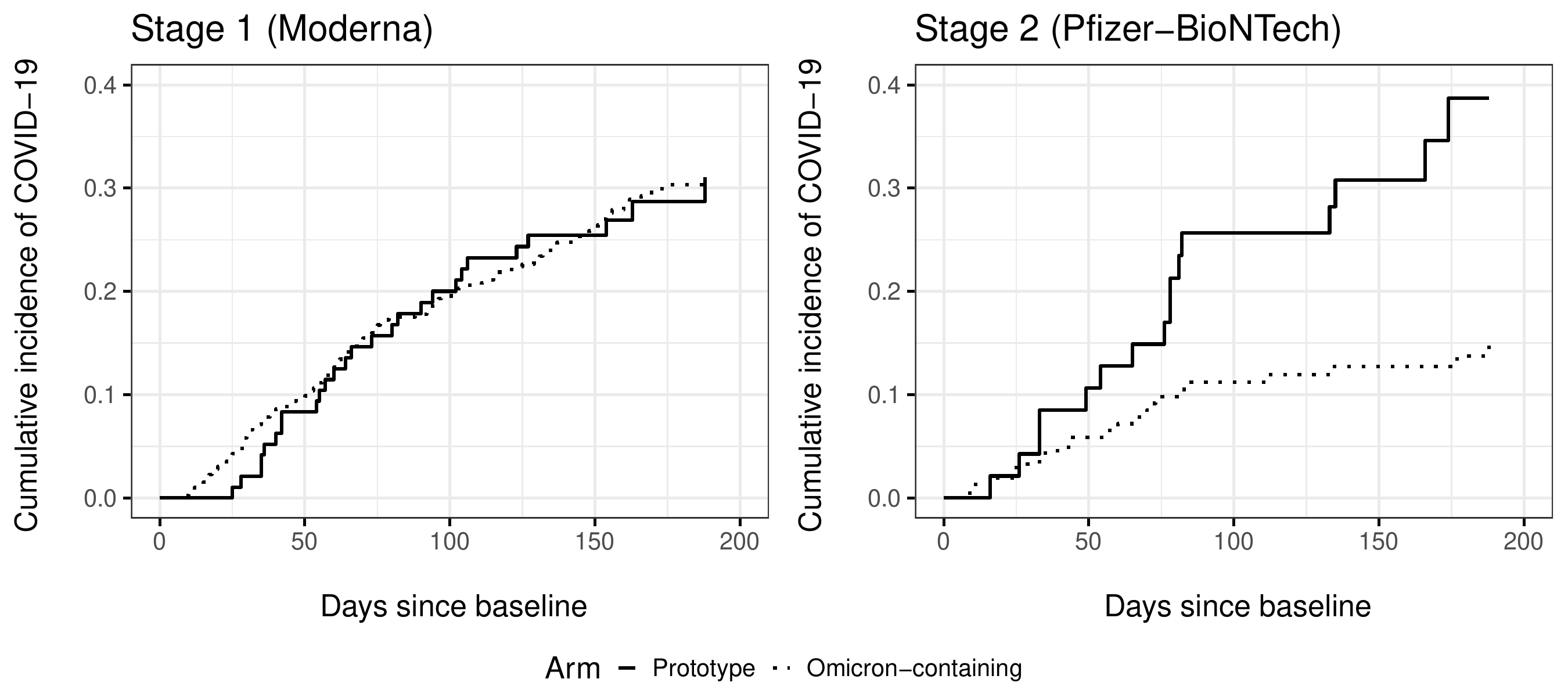}
\captionsetup{width=\textwidth}
    \caption{Kaplan-Meier estimates of the cumulative incidence among Stage 1 Moderna participants who received an Omicron-containing vaccine (Omicron + Beta, Omicron + Delta, Omicron, Omicron + Prototype; $n = 407$ in total), Stage 1 Moderna participants who received a Prototype vaccine ($n = 97$), Stage 2 Pfizer--BioNTech participants who received an Omicron-containing vaccine (Omicron + Beta, Omicron, Omicron + Prototype; $n = 156$ in total), and Stage 2 Pfizer--BioNTech participants who received a Prototype vaccine ($n = 47$).}
    \label{fig: km curve}
\end{figure}

{\color{black}
\subsection{The effects of family size on female labor supply}
\label{sec: application 2}
We revisited the empirical study of \citet{angrist2010extrapolate}, which examines the effect of family size on female labor supply. Specifically, \citet{angrist2010extrapolate} studied the effect of a third birth on mothers' number of working weeks, using two instrumental variables for childbearing: twin births and sibling sex composition. Our goal here is to test whether the IV compatibility assumption holds for these two IVs within Washington State, controlling for mother's age, age at first birth, ages of the first two children, and a binary indicator for race (White versus other), as in \citet{angrist2010extrapolate}. Using $N=7649$ samples from Washington State, we found no evidence against the compatibility condition (combined-test $p=1$; constituent $p$-values $0.752$, $0.775$, $0.725$, and $0.722$ for $J^\ast=5,10,20,40$, respectively). Failing to reject the IV compatibility assumption implies no evidence against \citeauthor{angrist2010extrapolate}'s conditional effect ignorability assumption for both IVs \citeyearpar{angrist2010extrapolate}, a key assumption in their original data analysis.

}

\section*{Acknowledgements}
We thank the Editor, Associate Editor, and two reviewers for their constructive feedback. We thank the participants, site staff, and investigators of the COVAIL trial sponsored by the Division of Microbiology and Infectious Diseases at  the National Institute Of Allergy and Infectious Diseases of the National Institutes of Health (NIAID NIH). This work is funded by the NIAID/NIH (R01AI192632 to BZ). The content is solely the responsibility of the authors and does not necessarily represent the official views of the National Institutes of Health.

{\singlespacing
\bibliographystyle{apalike}
\bibliography{paper-ref}
}

\newpage
\begin{center}
    \Large Supplemental Materials to ``\emph{Generalized projection tests for function-valued parameters with applications to testing structural causal assumptions}" \\by Rui Wang, Albert Osom and Bo Zhang 
\end{center}
\renewcommand{\thesection}{S\arabic{section}}
\setcounter{page}{1}
\setcounter{section}{0}

\renewcommand{\thefigure}{S.\arabic{figure}}

\setcounter{figure}{0}

\renewcommand{\thetable}{S\arabic{table}}
\setcounter{table}{0}

\renewcommand{\thelemma}{S\arabic{lemma}}
\setcounter{lemma}{0}
\renewcommand{\thedefinition}{S\arabic{definition}}
\setcounter{definition}{0}

\renewcommand{\theexample}{S\arabic{example}}
\setcounter{example}{0}

\renewcommand{\theproposition}{S\arabic{proposition}}
\setcounter{proposition}{0}

\renewcommand{\thetheorem}{S\arabic{theorem}}
\setcounter{theorem}{0}

\renewcommand{\theremark}{S\arabic{remark}}
\setcounter{remark}{0}

\setcounter{figure}{0}

Throughout the supplemental material, we will omit subscript $n$ or $P$ when there is no ambiguity. For example, $\Omega_{n,P}$ may be written as $\Omega_{n}$,  $\widehat{\Omega}_n$ may be written as $\widehat{\Omega}$ for simplicity. For nuisance estimator $\widehat{\eta}_n$, we will also write it as ${\eta}_{\widehat{P}_n}$. To simplify notation, we will write $g$ rather than $g^\perp$ in the proof when there is no confusion.

\section{Examples outside of causal inference}
\begin{example}[\textbf{Parametric model specification}]\label{ex:param}
     A classical topic in econometrics is to test model specification \citep{hong1995consistent}. Let $O=(X,Y)$, where $Y$ is the outcome of interest and $X$ is a vector of predictors. Let $\{h(x;\eta);\eta \in T \subset \mathbb{R}^p\}$ be a class of functions indexed by a Euclidean-valued parameter $\eta$. A model specification test concerns testing $H_0: \mathbbm{E}_{P}[Y - h(X;\eta_{P}) \mid X] = 0$, where $\eta_{P} = \arg\min_{\eta \in T} \mathbbm{E}_{P}\{[Y-h(X;\eta)]^2\}$. By letting $g(O;\eta_{P}) = Y-h(X;\eta_{P})$, the null hypothesis can be rewritten as testing $\mathbbm{E}_{P}[g(O;\eta_{P})\mid X] = 0$.
\end{example}

\begin{example}[\textbf{Conditional covariance}]\label{ex:conditional covariance}
    Let $O=(X,Y,Z)$. A conditional covariance test concerns testing if $Y$ is uncorrelated with $Z$ given $X$, that is, $H_0: \text{cov}_{P}(Y,Z\mid X) = 0$, which is a testable implication of conditional independence \citep{lundborg2024projected}. By letting $g(O;\eta_{P}) = (Y-\mathbbm{E}_P[Y\mid X])\cdot Z$, the null hypothesis can be rewritten as testing $\mathbbm{E}_{P}[g(O;\eta_{P})\mid X] = 0$, where  $\eta_{P} = \mathbbm{E}_P[Y\mid X]$ is a function-valued nuisance parameter.
\end{example}

\iffalse
\begin{examplecont}{ex:param}
Suppose the influence function of $\eta_P$ is $\psi_p(o)$, then an orthogonal version of $g$ in $H_0$, denoted as $g^\perp$, is
\begin{align*}
    g^\perp(O;\eta_P) = Y-h(X;\eta_P)-\left[\frac{\partial h(X;\eta_P)}{\partial \eta_P}\right]^\top \psi_p(O),
\end{align*}
where $\psi_p(o)$ can be derived from standard M-estimation theory. Since $\mathbbm{E}_P[\psi_p(O)\mid X] = 0$, we have $\mathbbm{E}_P[g^\perp(O;\eta_P)\mid X] = \mathbbm{E}_P[g(O;\eta_P)\mid X]$. 
Therefore, to test $H_0: \mathbbm{E}_P[g(O;\eta_P)\mid X]=0$ is equivalent to testing $H'_0: \mathbbm{E}_P[g^\perp(O;\eta_P)\mid X] = 0$.
\end{examplecont}

\begin{examplecont}{ex:conditional covariance}
The orthogonal version of $g$ is
\begin{align*}
    g^\perp(O;\eta_P) = (Y-\mathbbm{E}_P[Y|X])(Z-\mathbbm{E}_P[Z|X]).
\end{align*}
In this case, an additional nuisance parameter, $\mathbbm{E}_P[Z|X]$ needs to be estimated.
\end{examplecont}
\fi

\section{Revisiting the two examples in Section \ref{sec: example}}
\label{sec: revisiting examples, supp}
We specialize the high-level conditions into the two examples we studied in Section \ref{sec: examples}. For the mean exchangeability example in Section \ref{subsec: example mean exchangeability}, we have the following conditions
\begin{proposition}
    Let $\mu_{P,a,s}(x) = \mathbbm{E}_P[Y|A=a,X=x,S=s]$, $\pi_{P,a,s}(x) = P(A=a,S=s|X=x)$. Let $\mu_{\widehat{P},a,s}$ and $\pi_{\widehat{P},a,s}$ be the estimators for $\mu_{P,a,s}(x)$, $\pi_{P,a,s}(x)$, respectively. Suppose the following conditions hold for some constants $C,\epsilon$: (a) Each coordinate of $X$ and $Y$ are bounded by the absolute constant $C$, $\mathbbm{E}_P[g^\perp(O;\eta_P)^2|X]\geq 1/C$. (b) For all $P\in \mathcal{P}$, $1-\epsilon>\pi_{P,a,s}(x)>\epsilon$. (c) Let $I\subset \{1,...,n\}$ be a random subset of size $n/L$. Suppose that the nuisance parameter estimator $\eta_{\widehat{P}} = \{\mu_{\widehat{P},a,s},\pi_{\widehat{P},a,s}:s\in \{0,1\}\}$ is constructed using only observations outside of $I$, satisfying the following conditions: with $P$-probability at least $1-\Delta_n$,$\epsilon<\pi_{\widehat{P},a,s}<1-\epsilon,\Vert \mu_{{P},a,s}-\mu_{\widehat{P},a,s}\Vert_{P,\infty}\leq C$,  $\Vert \mu_{{P},a,s}-\mu_{\widehat{P},a,s}\Vert_{P,2}\leq \zeta_n$,  $\Vert \pi_{{P},a,s}-\pi_{\widehat{P},a,s}\Vert_{P,\infty}\leq C$, $\Vert \pi_{{P},a,s}-\pi_{\widehat{P},a,s}\Vert_{P,2}\leq \zeta_n$, $\Vert \mu_{{P},a,s}-\mu_{\widehat{P},a,s}\Vert_{P,2}\times \Vert \pi_{{P},a,s}-\pi_{\widehat{P},a,s}\Vert_{P,2}\leq \zeta_n/{n}^{1/2}$. Assuming one of the conditions for $\zeta_n\xi_n$ listed in Assumption \ref{Assumption: nuisance}.3 hold. Then Assumption \ref{Assumption: nuisance} holds with $g=g^\perp_{\text{ME}}$. 
    \label{prop: example 1}
\end{proposition}

For the IV compatibility example in section \ref{subsec: example IV}, we have the following conditions.

\begin{proposition}
    Let $\mu^{Y,(j)}_{P,z_j}(x) = \mathbbm{E}_P[Y|Z_j=z_j,X=x]$,$\mu^{D,(j)}_{P,z_j}(x) = \mathbbm{E}_P[D|Z_j=z_j,X=x]$ $\pi^{(j)}_{P}(x) = P(Z_j=1|X=x)$. Let $\mu^{Y,(j)}_{\widehat{P},z_j}$, $\mu^{D,(j)}_{\widehat{P},z_j}$ and $\pi^{(j)}_{\widehat{P}}$ be the estimators for $\mu^{Y,(j)}_{{P},z_j}$, $\mu^{D,(j)}_{{P},z_j}$ and $\pi^{(j)}_{{P}}$ , respectively. Suppose the following conditions hold for some constants $C>0$ and $\epsilon\in (0,1)$: (a) Each coordinate of $X$ and $Y$ are bounded by the absolute constant $C$, $\mathbbm{E}[g(O;\eta_P)^2|X]\geq 1/C$. (b)  There is a $\epsilon \in (0,1)$ such that $\epsilon<\pi^{(j)}_{P}<1-\epsilon$, $ \mu^{D,(j)}_{P,1}(x)-\mu^{D,(j)}_{{P},0}(x) >\epsilon$. (c) Let $I\subset \{1,...,n\}$ be a random subset of size $n/L$. Suppose that the nuisance parameter estimator $\eta_{\widehat{P}} = \{\mu^{Y,(j)}_{\widehat{P},s}, \mu^{D,(j)}_{\widehat{P},s} ,\pi^{(j)}_{\widehat{P}}:s\in \{0,1\}, j \in \{1,2\}\}$ is constructed using only observations outside of $I$, satisfying the following conditions: with $P$-probability at least $1-\Delta_n$,$1-\epsilon/2>\pi^{(j)}_{\widehat{P}}>\epsilon/2, \mu^{D,(j)}_{\widehat{P},1}(x)-\mu^{D,(j)}_{\widehat{P},0}(x)>\epsilon/2,\Vert \mu^{Y,(j)}_{P,z_j}-\mu^{Y,(j)}_{\widehat{P},z_j}\Vert_{P,\infty}\leq C$, $\Vert \mu^{D,(j)}_{P,z_j}-\mu^{D,(j)}_{\widehat{P},z_j}\Vert_{P,\infty}\leq C$, $\Vert \pi^{(j)}_{{P}}-\pi^{(j)}_{\widehat{P}}\Vert_{P,\infty}\leq C$,$\Vert \mu^{Y,(j)}_{P,z_j}-\mu^{Y,(j)}_{\widehat{P},z_j}\Vert_{P,2}\leq q_n$, $\Vert \mu^{D,(j)}_{P,z_j}-\mu^{D,(j)}_{\widehat{P},z_j}\Vert_{P,2}\leq q_n$, $\Vert \pi^{(j)}_{{P}}-\pi^{(j)}_{\widehat{P}}\Vert_{P,2}\leq q_n$, $q_n \leq \zeta_n,q_n^2\leq \zeta_n/n^{1/2}$. Assuming one of the conditions for $\zeta_n\xi_n$ listed in Assumption \ref{Assumption: nuisance}.3 holds. Then Assumption \ref{Assumption: nuisance} holds with $g=g^\perp_{\text{COM}}$. 
     \label{prop: example 2}
\end{proposition}

\section{Additional discussions on Definition \ref{definiftion: conditional neyman orthogonal}}
\label{sec: additional discussion on Def 1}

{\color{black}
\subsection{Additional discussion on conditional Neyman orthogonality}
\label{subsec: additional discussion on neyman orthogonality}
Our formulation of Conditional Neyman orthogonality is based on a directional (Gateaux) derivative with respect to the nuisance parameter.  For a perturbation direction \(\eta_1-\eta_P\), we consider the path $\eta_t=(1-t)\eta_P+t\eta_1$ and require the derivative $\frac{\partial}{\partial t}
\mathbb{E}_P\{g(O;\eta_t)\mid X\}
|_{t=0}$ to exist for the relevant perturbation directions \(\eta_1\in\mathcal T\). A simple sufficient condition is that, for each such direction, \(t\mapsto g(O;\eta_t)\) is differentiable in a neighborhood of zero and its derivative is dominated by an integrable envelope, so that differentiation and conditional expectation can be interchanged. Under such conditions,

$$
\left.
\frac{\partial}{\partial t}
\mathbb{E}_P\{g(O;\eta_t)\mid X\}
\right|_{t=0}
=
\mathbb{E}_P
\left[
\left.
\frac{\partial}{\partial t}
g(O;\eta_t)
\right|_{t=0}
\,\middle|\,X
\right].
$$

More generally, our assumptions can be stated directly in terms of existence of the above directional derivative, without requiring pointwise differentiability of \(g\).

Our condition is closely related to, but distinct from, Condition 3 of \citet{luedtke2019omnibus}. Their condition is formulated in terms of conditional pathwise differentiability of a statistical functional. Specifically, for each regular parametric submodel \(P_t\) with score \(S\), they require the derivative of the conditional parameter, denoted as $T_P(X)$, along \(P_t\) to admit a conditional score representation of the form

$$
\left.
\frac{\partial}{\partial t}T_{P_t}(X)
\right|_{t=0}
=
\mathbb{E}_P\{D_P^T(O)S(O)\mid X\},
$$
where \(D_P^T\) is conditionally mean zero, together with additional boundedness conditions. Thus, their formulation perturbs the entire data-generating distribution through regular parametric submodels, whereas ours perturbs the nuisance parameter directly while holding the underlying distribution \(P\) fixed. Consequently, our definition requires differentiability only along the collection of nuisance directions \(\mathcal T\) that are relevant for nuisance estimation, rather than along all regular parametric submodels of a specified statistical model.

The assumption in Theorem 2 of \citet{yang2023forster} is more directly related to our construction. They consider a regular parametric submodel \(P_t\), with corresponding nuisance parameter \(\eta_{P_t}\), and assume that

$$
\left.
\frac{\partial}{\partial t}
\mathbb{E}_{P}\{g(O;\eta_{P_t})\mid X\}
\right|_{t=0}
$$
admits a conditional inner-product representation involving the score of the submodel. Their formulation is therefore also pathwise: the nuisance perturbation \(\eta_{P_t}\) is induced by a perturbation of the underlying distribution \(P_t\). In contrast, our definition directly considers nuisance paths such as $\eta_t=(1-t)\eta_P+t\eta_1, \eta_1\in\mathcal T,$ which is analogous to the Gateaux-derivative formulation of Neyman orthogonality in \citet{2018ChernozhukovDML}. Hence, our condition does not require specifying a statistical model or verifying differentiability along every regular parametric submodel; it only requires differentiability along the nuisance perturbations relevant to our analysis.

These two formulations are closely connected when perturbations of the nuisance parameter can be generated by sufficiently rich regular parametric submodels, but neither imply the other without additional regularity conditions between the statistical model and the nuisance parameter space. Recent work by \citet{chen2026equivalence} establishes formal equivalence results between Neyman orthogonality formulated through Gateaux derivatives and pathwise differentiability under additional local product-structure conditions. To our knowledge, a corresponding general equivalence result for the conditional formulations considered here has not yet been established.

Finally, in Supplemental Material \ref{section:neyman orthogonality and conditional neyman orthogonality}, we provide sufficient conditions connecting ordinary and conditional Neyman orthogonality. These results allow us, in many of the examples considered in this paper, to start from a previously derived Neyman-orthogonal score and verify that it also satisfies Definition \ref{definiftion: conditional neyman orthogonal}, rather than deriving a new conditional score from scratch. Our formulation of the assumption is also easy to verify, see Supplemental Material \ref{section:neyman orthogonality and conditional neyman orthogonality}. We also want to mention that our condition is actually easy to verify. For each $x$, define $m_x(t) = \mathbbm{E}_P[g^\perp(O;\eta_t)|X=x]$, note that for each $x$, $m_x(t)$ is just a scalar function. Then the condition is naturally interpreted as $m'_x(0) = 0$ for $P_X$ a.e. $x$. Consequently, verifying the condition amounts to checking the usual one-dimensional directional derivative with respect to $t$, pointwise in $x$, rather than establishing pathwise differentiability along regular parametric submodels of the full statistical model.

\subsection{Additional results on the relationship between conditional Neyman orthogonality and Neyman orthogonality}
\label{section:neyman orthogonality and conditional neyman orthogonality}

In this section, we establish connections between conditional Neyman orthogonality and Neyman orthogonality. Define $\tilde{\mathcal{T}} = \mathcal{T}-\eta_P$. We first restate the definition of Neyman orthogonality and conditional Neyman orthogonality below:
\begin{definition}
    \textbf{(Conditional Neyman orthogonality).} A function $g^\perp(o;\eta)$ is said to be conditional Neyman orthogonal, or Neyman orthogonal given $X$, if for any $v\in \tilde{\mathcal{T}}$,
    \begin{align*}
        \frac{\partial}{\partial t}\mathbb{E}_{P}\left[g^\perp(O;\eta_{P}+tv)\mid X\right]\Big\vert_{t=0} = 0.
    \end{align*}
\end{definition}

\begin{definition}
    \textbf{(Neyman Orthogonality).} A function $g^\perp(o;\eta)$ is said to be Neyman orthogonal, if for any $v\in \tilde{\mathcal{T}}$,
    \begin{align*}
        \frac{\partial}{\partial t}\mathbb{E}_{P}\left[g^\perp(O;\eta_{P}+tv)\right]\Big\vert_{t=0} = 0.
    \end{align*}
\end{definition}

Next we state a formal result on the
equivalence between Conditional Neyman orthogonality and Neyman orthogonality. To ease notation, for every $P\in \mathcal{P}$ we define $m_P(\eta)(X) = \mathbbm{E}_P[g^\perp(O;\eta)|X]$. Denote the Gateaux derivative of $m_P(\eta)(X)$ at $\eta_P$ with direction $v$ as   $\dot{m}_{\eta_P}[v](X)$, which is defined as
\begin{align*}
    \dot{m}_{\eta_P}[v] := \lim_{t\rightarrow 0}\frac{m_P(\eta_P+tv)-m_P(\eta_P)}{t}.
\end{align*}

\begin{theorem} (Equivalence between Conditional Neyman orthogonality and Neyman orthogonality). \label{theorem: equivalance between conditional neyman orthogonality and neyman orthogonality}
Suppose the following conditions hold. (a) For every  $v\in \tilde{\mathcal{T}}$ , the Gateaux derivative $\dot{m}_{\eta_P}[v](X)$ exists in $L_1(P_X)$. (b) For every $v\in \tilde{\mathcal{T}}$ and every scalar function $a:\mathcal{X}\rightarrow \mathbb{R}$ which satisfies $\vert a(x)\vert \leq 1$, there exists a constant $c>0$, which may depend on $v$ and $a$, such that $cav\in \tilde{\mathcal{T}}$. (c) Suppose the nuisance function $\eta:\mathcal{X}\rightarrow \mathbb{R}^d$. There exists a map $\widebar{m}_P:\mathcal{X}\times  \mathbb{R}^d \rightarrow \mathbb{R}$, such that for every $v\in \tilde{\mathcal{T}}$ and all sufficiently small $t$, $m_P(\eta_P+tv)(x) = \widebar{m}_P(x,\eta_P(x)+tv(x))$ for $P_X$ almost every $x$, Moreover, the map $u \mapsto \widebar{m}_P(x,u)$ is differentiable at $u = \eta_P(x)$. 

% (d) We assume the following differentiation and expectation can be exchanged:
% \begin{align*}
%     \mathbbm{E}_P\left[\frac{\partial}{\partial t}m_P(\eta_P+tv)\right] = \frac{\partial}{\partial t}\mathbbm{E}_P[m_P(\eta_P+tv)]
% \end{align*}
Then the conditional Neyman orthogonality and Neyman orthogonality are equivalent.
\end{theorem}

Theorem \ref{theorem: equivalance between conditional neyman orthogonality and neyman orthogonality} implies that a Neyman orthogonal moment function is also a conditional Neyman orthogonal moment function under regularity conditions.  Assumption (a) is a weak assumption on Gateaux differentiability in $L_1(P_X)$. The definition of conditional Neyman orthogonality requires such Gateaux derivative not only exists but also equals 0. By inspecting the proof, we can see that Assumption (a) can be replaced by pointwise Gateaux differentiability plus a condition on exchangeability of derivative and expectation. Assumption (b) is a local richness condition on the nuisance direction set, requiring it to be stable, up to rescaling, under pointwise sign changes. It is naturally satisfied in sufficiently rich nonparametric nuisance models. Assumption (c) is a condition on the conditional moment with respect to the nuisance functions. It requires that, conditional on $X=x$, the moment \(m_P(\eta)(x)=E_P[g^\perp(O;\eta)\mid X=x]\) depends on the nuisance function \(\eta\) only through its pointwise value \(\eta(x)\). This condition is natural in many of the settings considered here because the nuisance parameters are themselves conditional regression or probability functions indexed by \(X\), and the target moment is conditioned on the same \(X\).

The primary purpose of Theorem \ref{theorem: equivalance between conditional neyman orthogonality and neyman orthogonality} is to clarify the relationship between ordinary and conditional Neyman orthogonality, rather than to introduce an additional set of conditions that must be verified whenever our testing framework is applied. In particular, verifying Assumption (a) typically requires analyzing the Gateaux derivative of the conditional moment explicitly; in many examples, once this derivative has been calculated, one can directly verify whether it is equal to zero and hence establish conditional Neyman orthogonality without appealing to the theorem. The theorem is therefore most useful conceptually: it shows that, under suitable regularity conditions, ordinary Neyman orthogonality implies its conditional counterpart. Consequently, in applications one may often begin with a Neyman-orthogonal score or moment function already established in the literature and then directly verify the conditional Neyman orthogonality condition required by our framework.

\begin{proof}
    We first establish a linear representation result for $\dot{m}_{\eta_P}[v]$. Denote $A_P(x) = \frac{\partial}{\partial u}\widebar{m}_P(x,\eta_P(x))$, which exists by (c). Fix $v\in \tilde{\mathcal{T}}$, by (c) we have $m_P(\eta_p+tv)(X) = \widebar{m}_P(X,\eta_P(X)+tv(X))$. Therefore, for $P_X$ almost all $X$,
    \begin{align*}
        \frac{m_P(\eta_P+tv)(X)-m_P(\eta_P)(X)}{t} = \frac{\widebar{m}(X,\eta_P+tv(X))-\widebar{m}(X,\eta_P)}{t} \rightarrow A_P(X)v(X).
    \end{align*}

    On the other hand, (a) states that the same object converges to  $\dot{m}_{\eta_P}[v]$, hence it also converges in probability to  $\dot{m}_{\eta_P}[v]$. Therefore, we must have with probability one $\dot{m}_{\eta_P}[v](X) = A_P(X)v(X)$.

    let $M_P(\eta) = \mathbbm{E}_P[m_P(\eta)(X)]$. Then 
    \begin{align*}
        \frac{d}{dt}M_P(\eta_P+tv) \Big\vert_{t = 0}= \mathbbm{E}_P[\dot{m}_{\eta_P}[v](X)].
    \end{align*}
    This interchange between differentiation and expectation is a result of Assumption (a). This identity also directly shows that conditional Neyman orthogonality implies Neyman orthogonality. Next we aim to show that Neyman orthogonality implies conditional orthogonality.

    Let $a(X)$ be arbitrary scalar function with $\vert a(x)\vert \leq 1$, by (b), there exists a constant $c$ such that $cav \in \tilde{\mathcal{T}}$. The Neyman orthogonality gives
    \begin{align*}
        \frac{d}{dt}M_P(\eta_P+tv)\Big\vert_{t = 0} = 0,
    \end{align*}
    which implies $\mathbbm{E}_P[\dot{m}_{\eta_P}[v](X)] = 0$. We apply this to the direction $cav$, then we have $\mathbbm{E}_P[\dot{m}_{\eta_P}[cav](X)] = 0$. Since $ \dot{m}_{\eta_P}[cav](X) = cA_P(X) a(X)v(X) =c a(X)\dot{m}_{\eta_P}[v](X)$, we have $\mathbbm{E}_P[a(X)\dot{m}_{\eta_P}[v](X)]=0$ for any $a$ with $\vert a(x)\vert \leq 1$. Take $a = \text{sign}(\dot{m}_{\eta_P}[v](X))$, we conclude that $\dot{m}_{\eta_P}[v](X) = 0$, therefore the conditional Neyman orthogonality holds.
\end{proof}

We now verify Conditions (b) and (c) of Theorem~\ref{theorem: equivalance between conditional neyman orthogonality and neyman orthogonality} for the usual AIPW
estimating function
\[
\psi(O;\eta)
=
\frac{\mathbbm{1}(A=1)}{\mu_a(X)}
\{Y-\mu_y(X)\}
+\mu_y(X),
\]
where $\eta=(\mu_a,\mu_y)$, $\mu_{a,P}(x)=P(A=1\mid X=x)$, $\mu_{y,P}(x)=\mathbbm{E}_P[Y\mid A=1,X=x]$. The moment function in example 1 takes the similar form. We first specify the nuisance parameter space. Let $\mathcal T
=
\mathcal T_a\times\mathcal T_y$,
where for some constants $\epsilon\in(0,1/2)$ and $K<\infty$, $\mathcal T_a
=
\left\{
\mu_a:\epsilon\leq \mu_a(x)\leq 1-\epsilon
\right\},$ and $\mathcal T_y
=
\left\{
\mu_y:|\mu_y(x)|\leq K
\right\}.$
Suppose that the true nuisance functions are uniformly in the
interior of $\mathcal T$: there exist constants
$\kappa_a>0$ and $\kappa_y>0$ such that $\epsilon+\kappa_a
\leq
\mu_{a,P}(x)
\leq
1-\epsilon-\kappa_a$ and $|\mu_{y,P}(x)|
\leq
K-\kappa_y$ for $P_X$-almost every $x$. Recall that $\widetilde{\mathcal T}
=
\mathcal T-\eta_P
=
\{\,\eta-\eta_P:\eta\in\mathcal T\,\}.$

\noindent\textbf{Verification of Condition (a).}
Fix an arbitrary direction
$v=(v_a,v_y)\in\widetilde{\mathcal T}$ and let
$\eta_t=\eta_P+tv$. Direct calculation gives
\begin{align*}
m_P(\eta)(x)
&=
\mathbbm{E}_P
\left[
\frac{\mathbbm{1}(A=1)}{\mu_a(x)}
\{Y-\mu_y(x)\}
+\mu_y(x)
\;\middle|\;
X=x
\right]\\
&=
\frac{1}{\mu_a(x)}
\mathbbm{E}_P
\left[
\mathbbm{1}(A=1)\{Y-\mu_y(x)\}
\mid X=x
\right]
+\mu_y(x)\\
&=
\frac{\mu_{a,P}(x)}{\mu_a(x)}
\{\mu_{y,P}(x)-\mu_y(x)\}
+\mu_y(x).
\end{align*}
we have
\begin{align*}
m_P(\eta_P+tv)(x)
&=
\frac{\mu_{a,P}(x)}
{\mu_{a,P}(x)+tv_a(x)}
\{-tv_y(x)\}
+
\mu_{y,P}(x)+tv_y(x).
\end{align*}
Therefore,
\begin{align*}
\frac{
m_P(\eta_P+tv)(x)-m_P(\eta_P)(x)
}{t}
&=
\frac{
t v_a(x)v_y(x)
}{
\mu_{a,P}(x)+tv_a(x)
}.
\end{align*}
Since the true propensity score is uniformly bounded away from
zero and $v_a$ is uniformly bounded, for all sufficiently small
$|t|$, $\mu_{a,P}(x)+tv_a(x)
\geq
\epsilon+\kappa_a/2$ for $P_X$-almost every $x$. Moreover, $v_a$ and $v_y$ are uniformly bounded because $v\in\widetilde{\mathcal T}$. Hence
\begin{align*}
\left\|
\frac{
m_P(\eta_P+tv)-m_P(\eta_P)
}{t}
\right\|_{P_X,1}
&\leq
\frac{
|t|\|v_a\|_\infty\|v_y\|_\infty
}{
\epsilon+\kappa_a/2
}
\longrightarrow 0.
\end{align*}
Thus the Gateaux derivative exists in $L_1(P_X)$ for every
$v\in\widetilde{\mathcal T}$ (and, in this example, is equal to
zero). Therefore Condition (a) is satisfied.

\noindent\textbf{Verification of Condition (b).}
Fix an arbitrary $v=(v_a,v_y)\in\widetilde{\mathcal T}$ and an arbitrary scalar function $a:\mathcal X\to\mathbb R$
satisfying $|a(x)|\leq1$.
Since $v\in\widetilde{\mathcal T}$, there exists
$\eta_1=(\mu_{a,1},\mu_{y,1})\in\mathcal T$ such that $v_a=\mu_{a,1}-\mu_{a,P},
v_y=\mu_{y,1}-\mu_{y,P}.$ In particular, $v_a$ and $v_y$ are uniformly bounded. Indeed, $\|v_a\|_\infty\leq 1,
\|v_y\|_\infty\leq 2K.$ Choose $c>0$ sufficiently small such that $c\|v_a\|_\infty\leq \kappa_a,c\|v_y\|_\infty\leq \kappa_y.$
For example, one may take $c
=
\min\left\{
1,\,
\frac{\kappa_a}{\|v_a\|_\infty},\,
\frac{\kappa_y}{\|v_y\|_\infty}
\right\}.$

Since $|a(x)|\leq 1$, $|ca(x)v_a(x)|\leq \kappa_a,
|ca(x)v_y(x)|\leq \kappa_y.$ Therefore, $\epsilon
\leq
\mu_{a,P}(x)+ca(x)v_a(x)
\leq
1-\epsilon$ and $\left|
\mu_{y,P}(x)+ca(x)v_y(x)
\right|
\leq K$ for almost every $x$. Hence $\eta_P+cav\in\mathcal T,$ which is equivalent to $cav\in\widetilde{\mathcal T}.$ Thus Condition~(b) is satisfied.

\noindent\textbf{Verification of Condition (c).}
For a generic nuisance value
$\eta=(\mu_a,\mu_y)\in\mathcal T$, define $m_P(\eta)(x)
=
\mathbbm{E}_P[\psi(O;\eta)\mid X=x].$
Direct calculation gives
\begin{align*}
m_P(\eta)(x)
&=
\mathbbm{E}_P
\left[
\frac{\mathbbm{1}(A=1)}{\mu_a(x)}
\{Y-\mu_y(x)\}
+\mu_y(x)
\;\middle|\;
X=x
\right]\\
&=
\frac{1}{\mu_a(x)}
\mathbbm{E}_P
\left[
\mathbbm{1}(A=1)\{Y-\mu_y(x)\}
\mid X=x
\right]
+\mu_y(x)\\
&=
\frac{\mu_{a,P}(x)}{\mu_a(x)}
\{\mu_{y,P}(x)-\mu_y(x)\}
+\mu_y(x).
\end{align*}
Define $\bar m_P(x,u_a,u_y)
=
\frac{\mu_{a,P}(x)}{u_a}
\{\mu_{y,P}(x)-u_y\}
+u_y, u_a>0.$ Then $m_P(\eta)(x)
=
\bar m_P\{x,\eta(x)\}.$
In particular, for every $v\in\widetilde{\mathcal T}$ and all sufficiently
small $t$ such that $\eta_P+tv\in\mathcal T$,
\[
m_P(\eta_P+tv)(x)
=
\bar m_P\{x,\eta_P(x)+tv(x)\}.
\]
Thus the conditional moment is local in the nuisance function: conditional
on $X=x$, it depends on the entire nuisance function $\eta$ only through its
pointwise value $\eta(x)$.

Moreover, because the propensity-score component is uniformly bounded away
from zero, the map $(u_a,u_y)\mapsto \bar m_P(x,u_a,u_y)$ is a scalar function on $\mathbb{R}^2$, which is also continuously differentiable in a neighborhood of $\eta_P(x)
=
\{\mu_{a,P}(x),\mu_{y,P}(x)\}.$ 
Therefore Condition~(c) is satisfied. The conditions can be shown to  hold for the IV compatibility example in a similar way.

\section{Additional details on the nuisance estimation rates}
\label{sec: supp nuisance rates}
\begin{table}[htbp]

\centering
\caption{Crude sufficient nuisance convergence-rate requirements for the standardized generalized projection test under various growth rates of $J_n$ when $\xi_n\asymp C$ and $\omega_n\asymp J_n^{1/2}$. Here, we report a common $L_2(P)$ convergence rate $q_n$ for each nuisance component under the product-rate conditions described in the Supplementary Materials. We use $\epsilon>0$ to denote a small positive constant.}
\begin{tabular}{lcc}

\toprule
Growth of $J_n$
& $\xi_n\asymp C,\omega_n\asymp J_n^{1/2}$
& $\xi_n\asymp J_n^{1/2}, \omega_n\asymp J_n$ \\
\midrule
$n^{1/3-\epsilon}$
& $o(n^{-5/12+\epsilon/2})$
& $o(n^{-1/2+3\epsilon/4})$ \\

$n^{1/6}$
& $o(n^{-1/3})$
& $o(n^{-3/8})$\\

$n^{1/12}$
& $o(n^{-7/24})$
& $o(n^{-5/16})$ \\

$n^{\epsilon}$
& $o(n^{-1/4-\epsilon/2})$
& $o(n^{-1/4-3\epsilon/4})$ \\
\bottomrule
\end{tabular}
\label{table: nuisance rates}
\end{table}

Table \ref{table: nuisance rates} makes explicit the tradeoff between the growth of the sieve dimension and the required nuisance convergence rate. When $J_n$ grows close to the maximal rate allowed by our theory for the standardized test, $J_n=n^{1/3-\epsilon}$, the requirement for nuisance convergence rate is relatively strong and can, for example, be satisfied by correctly specified regular parametric estimators, which typically converge at the $n^{-1/2}$ rate. As $J_n$ grows more slowly, the nuisance convergence requirement becomes weaker. For example, with $\xi_n\asymp C$ and $\omega_n\asymp J_n^{1/2}$, taking $J_n=n^{1/6}$ requires $q_n=o(n^{-1/3})$, while taking $J_n=n^{1/12}$ requires only $q_n=o(n^{-7/24})$. Such rates may also be attainable by suitable nonparametric or semiparametric estimators like the additive models \citep{stone1985additive} and the highly adaptive lasso \citep{bibaut2019fast}, under appropriate smoothness or structural assumptions. When $J_n=n^\epsilon$ with $\epsilon$ close to zero, the sufficient rate approaches the familiar $n^{-1/4}$ rate in double/debiased machine learning.

\section{The causal interpretation of the local alternatives}
\label{sec:causal interpretation of the local alternatives}
For the mean exchangeability example in Section \ref{subsec: example mean exchangeability}, the causal restriction of interest is
\[
\mathbbm{E}_P[Y(a)\mid X,S=1]
-
\mathbbm{E}_P[Y(a)\mid X,S=0]
=
0.
\]
Accordingly, a sequence of \textbf{causal} local violations of mean exchangeability can be written as
\[
\mathbbm{E}_{P_{1,n}}[Y(a)\mid X,S=1]
-
\mathbbm{E}_{P_{1,n}}[Y(a)\mid X,S=0]
=
\left(
\frac{\gamma_{n,P}}{n}
\right)^{1/2}
\phi(X).
\]
Under the treatment exchangeability assumptions used for identification in Section 4.1, we have $\mathbbm{E}_P[Y(a)\mid X,S=s]
=
\mathbbm{E}_P[Y\mid A=a,X,S=s]$. Moreover, the orthogonal score $g^\perp_{\mathrm{ME}}$ constructed in Equation (3) satisfies $\mathbbm{E}_P
\left[
g^\perp_{\mathrm{ME}}(O;\eta_P)
\mid X
\right]
=
\mathbbm{E}_P[Y\mid A=a,X,S=1]
-
\mathbbm{E}_P[Y\mid A=a,X,S=0].$ Therefore, the above local \textbf{causal} violation implies
\[
\mathbbm{E}_{P_{1,n}}
\left[
g^\perp_{\mathrm{ME}}(O;\eta_{P_{1,n}})
\mid X
\right]
=
\left(
\frac{\gamma_{n,P}}{n}
\right)^{1/2}
\phi(X),
\]
which is exactly the form of the local alternative considered we formulated using observed data in the first example.

Similarly, for the IV compatibility example in Section 4.2, the causal restriction compares the conditional treatment effects among the complier populations associated with the two instruments.
A sequence of causal local violations can be represented as
\[
\mathbbm{E}_{P_{1,n}}
\left[
Y(1)-Y(0)
\mid Z_1\text{-complier},X
\right]
-
\mathbbm{E}_{P_{1,n}}
\left[
Y(1)-Y(0)
\mid Z_2\text{-complier},X
\right]
=
\left(
\frac{\gamma_n}{n}
\right)^{1/2}
\phi(X).
\]
Under the IV identification assumptions invoked in Section 4.2, these two conditional complier average treatment effects are identified respectively by $\Psi_{Z_1,P}(X)$ and $\Psi_{Z_2,P}(X)$.
The orthogonal score $g^\perp_{\mathrm{COM}}$ is constructed such that $\mathbbm{E}_P
\left[
g^\perp_{\mathrm{COM}}(O;\eta_P)
\mid X
\right]
=
\Psi_{Z_1,P}(X)-\Psi_{Z_2,P}(X)$.
Therefore, the local causal violation above translates directly into the following local alternative formulated using observed data:
\[
\mathbbm{E}_{P_{1,n}}
\left[
g^\perp_{\mathrm{COM}}(O;\eta_{P_{1,n}})
\mid X
\right]
=
\left(
\frac{\gamma_{n,P}}{n}
\right)^{1/2}
\phi(X),
\]
which is again the local alternative analyzed in example 2.

Thus, although the hypothesis test and Theorem 3 are expressed in terms of identified observed-data functionals, under the identification assumptions specific to each application these functionals coincide with the corresponding causal discrepancies. The local power result can therefore be interpreted as describing power against increasingly small violations of the underlying causal restrictions.}

\section{Technical lemmas}
\label{sec:technical lemma}
In this section, we describe a collection of useful lemmas which will be used in the following proofs.

\begin{lemma}
    Let $X_1,..,X_n$ be i.i.d. $\mathbbm{R}^d$-valued random vectors with mean $0$ and positive definite covariance matrix $\Sigma$. Then for $W\sim N(0,\Sigma)$ and any convex sets $A\subset \mathbbm{R}^d$, we have
    \begin{align*}
      \left \vert P\left(\frac{1}{{n}^{1/2}}\sum_{i=1}^n X_i\in A\right)-P\left(W\in A\right)\right\vert\leq \frac{42d^{1/4}+16}{{n}^{1/2}} \mathbbm{E}_{P}\Vert \Sigma_{\text{sym}}^{-1/2}X_i\Vert^3,
    \end{align*}
    \label{lemma: Berry--Esseen}
    where, for a PD matrix $\Sigma$ with decomposition $\Sigma = \Gamma^\top\Lambda\Gamma$, we define $\Sigma^{1/2}_{\text{sym}} = \Gamma^\top\Lambda^{1/2}\Gamma$.
\end{lemma}

\begin{proof}
    Let $Y_i = \Sigma^{-1/2}_{\text{sym}}X_i$, then by Theorem 1.1 of \citet{raivc2019multivariate}, we have for $Z\sim N(0,I_d)$,
    \begin{align*}
         \left \vert P\left(\frac{1}{{n}^{1/2}}\sum_{i=1}^n Y_i\in A\right)-P\left(Z\in A\right)\right\vert\leq \frac{42d^{1/4}+16}{{n}^{1/2}} \mathbbm{E}_{P}\Vert Y_i\Vert^3.
    \end{align*}

    Since $\{x:\Sigma_{\text{sym}}^{-1/2}x  \in A\} =\{x: x  \in \Sigma_{\text{sym}}^{1/2} A\}$, where $\Sigma_{\text{sym}}^{1/2} A = \{\Sigma_{\text{sym}}^{1/2}x:x\in A\}$, we have
    \begin{align*}
      \left \vert P\left(\frac{1}{{n}^{1/2}}\sum_{i=1}^n X_i\in A\right)-P\left(W\in A\right)\right\vert\leq \frac{42d^{1/4}+16}{{n}^{1/2}} \mathbbm{E}_{P}\Vert Y_i\Vert^3,
    \end{align*}
    which proves this result.
\end{proof}

\begin{lemma}
Let $\Sigma_n = \mathbbm{E}_P[g^2(O;\eta_P)\tilde{B}^\lambda_n[\tilde{B}_n^\lambda]^\top ]$, let $\tau_1,...,\tau_{J_n}$ be the eigenvalues of $\Sigma_n$, $Z_1,...,Z_{J_n}$ be $J_n$ independent $\chi^2$ distributed random variables with degree of freedom 1. Furthermore, suppose the null hypothesis is true and $\frac{J_n^{5/4}\omega_n}{{n}^{1/2}}\rightarrow 0$, $\mathbbm{E}[g(O;\eta_P)^2|X]\geq c$, $\mathbbm{E}[g(O;\eta_P)^4|X]\leq C$, $\mathbbm{E}_P[B_n^\top B_n]\succeq cI$,  then we have
    \begin{align*}
        \sup_{P\in \mathcal{P}_0}\sup_{t\in \mathbb{R}} \left\vert P(W^\top W \leq t)-P\left(\sum_{j=1}^{J_n}\tau_jZ_j\leq t\right)\right\vert \rightarrow 0,
    \end{align*}
    where
    \begin{align*}
        W = \left(\frac{1}{{n}^{1/2}}\sum_{i=1}^ng(O_i;\eta_P)\tilde{b}^\lambda_1,...,\frac{1}{{n}^{1/2}}\sum_{i=1}^ng(O_i;\eta_P)\tilde{b}^\lambda_{J_n}\right)
    \end{align*}
    \label{lemma: chi-square approximation}
\end{lemma}

\begin{proof}
    $W$ can be written as $\frac{1}{{n}^{1/2}}\sum_{i=1}^n Y_i$, where
    \begin{align*}
        Y_i = \left(g(O_i;\eta_P)\tilde{b}^\lambda_1,...,g(O_i;\eta_P)\tilde{b}^\lambda_{J_n}\right)^\top.
    \end{align*}
    Therefore, $Y_1,...,Y_n$ are i.i.d. random $\mathbbm{R}^{J_n}$-valued random vectors with covariance matrix $\Sigma$. By Lemma \ref{lemma: Berry--Esseen}, we have for any convex subset $A$,
    \begin{align*}
      \left \vert P\left(W\in A\right)-P\left(\tilde{W}\in A\right)\right\vert\leq \frac{42d^{1/4}+16}{{n}^{1/2}} \mathbbm{E}_{P}\Vert \Sigma_{n,\text{sym}}^{-1/2}Y_i\Vert^3,
    \end{align*}
    where $\tilde{W}$ is a $N(0,\Sigma)$ random vector. Since the sets in the form $\{x:x^\top x\leq t\}$ are convex sets, we have
    \begin{align}
     \sup_{t\in \mathbb{R}} \left \vert P\left(W^\top W\leq t\right)-P\left(\tilde{W}^\top \tilde{W}\leq t\right)\right\vert\leq \frac{42{J_n}^{1/4}+16}{{n}^{1/2}} \mathbbm{E}_{P}\Vert \Sigma^{-1/2}Y_i\Vert^3. \label{eq: BE bound}
    \end{align}
    Furthermore, since $\tilde{W}^\top \tilde{W}$ has the same distribution as $\sum_{j=1}^{J_n}\tau_jZ_j$,we only need to prove the right hand side of \eqref{eq: BE bound} converges to $0$ uniformly in $P\in \mathcal{P}$.

Define $H_n = \mathbbm{E}_P[g^2(O;\eta_P)B_n(X)B_n^T(X)]$, $A_n = \Lambda_n^{1/2}\Gamma_n$, so that $\tilde{B}_n = A_nB_n$, $Y_i = A_n g_i B_i$, $\Sigma_n = A_n H_n A_n^\top$. Therefore,
\begin{align*}
    \Sigma_{n,\text{sym}}^{-1/2}Y_i = \Sigma_{n,\text{sym}}^{-1/2}A_n g_i B_i = \Sigma_{n,\text{sym}}^{-1/2}A_n H_{n,\text{sym}}^{1/2}H_{n,\text{sym}}^{-1/2}g_i B_i.
\end{align*}

Note that
\begin{align*}
    (\Sigma_{n,\text{sym}}^{-1/2}A_n H_{n,\text{sym}}^{1/2})(\Sigma_{n,\text{sym}}^{-1/2}A_n H_{n,\text{sym}}^{1/2})^\top = 1,
\end{align*}
we have 
\begin{align*}
    \Vert \Sigma_{n,\text{sym}}^{-1/2}Y_i \Vert  = \Vert H^{-1/2}_{n,\text{sym}} g_i B_i\Vert.
\end{align*}

Now it remains to derive an explicit upper bound of $\Vert H^{-1/2}_{n,\text{sym}} g_i B_i\Vert$. Since $H_n = \mathbbm{E}_P[\mathbbm{E}_P[g^2|X]B_nB_n^\top] \succeq 1/C\mathbbm{E}_P[B_nB_n^\top] \succeq 1/C^2 I_{J_n}$, we have $\Vert H^{-1}_n\Vert \leq C^2$. Therefore,
\begin{align*}
    &\mathbbm{E}_P[\Vert H^{-1/2}_{n,\text{sym}} g_i B_i\Vert^3] \leq \Vert H_{n,\text{sym}}^{-1/2}\Vert ^{3} \mathbbm{E}_P[\vert g_i\vert^3\Vert B_i\Vert^3]\\
    \lesssim &\mathbbm{E}_P[\vert g_i\vert^3\Vert B_i\Vert^3] \lesssim \mathbbm{E}_P[\Vert B_i\Vert^3] \lesssim \omega_n \mathbbm{E}_P\Vert B_n(X)\Vert ^2 \lesssim \omega_n J_n
\end{align*}
    
%     \begin{align*}
%         \mathbbm{E}_{P}\Vert \Sigma^{-1/2}Y_i\Vert^3
%         \leq  \Vert \Sigma\Vert^{-3/2} \mathbbm{E}_P\Vert Y_i\Vert^3
%     \end{align*}
%     We first consider how to bound $\Vert \Sigma \Vert^{-3/2}$. The key here is to give a lower bound for $\Vert \Sigma \Vert$.
%     \begin{align*}
%         &\Vert \Sigma \Vert\\
% =&\left\Vert \mathbbm{E}_{P}\left[g^2(O_i; {\eta}_P) \left( {\Omega}_n^{1/2} B_n(X_i) \right) \left( {\Omega}_n^{1/2} B_n(X_i) \right)^\top \right] \right\Vert\\
% \gtrsim & \left\Vert{\Omega}_n^{1/2} \mathbbm{E}_{P}\left[ B_n(X_i)   B_n(X_i) ^\top \right] ({\Omega}_n^{1/2})^\top \right\Vert\\
% \gtrsim & \left\Vert{\Omega}_n^{1/2} ({\Omega}_n^{1/2})^\top \right\Vert\\
% = & \lambda_{\max}
%     \end{align*}
%     Therefore we have $\Vert \Sigma \Vert^{-3/2}\leq \lambda_{\max}^{-3/2}$. Next, we give a bound on $\mathbbm{E}_P\Vert Y_i\Vert^3$:
%     \begin{align*}
%         & \mathbbm{E}_P\Vert Y_i\Vert^3\\
%         = & \mathbbm{E}_P\left[\left(\sum_{j=1}^{J_n}\lambda_jg^2(O_i;\eta_P)\tilde{b}^2_j(X)\right)^{3/2}\right] \\
%         \lesssim & \mathbbm{E}_P[g^3(O_i;\eta_P)](\text{tr}(\Omega))^{3/2}\\
%         \lesssim & (\text{tr}(\Omega))^{3/2}
%     \end{align*}

    Therefore, the right-hand side of \eqref{eq: BE bound} reduces to 
    \begin{align*}
        C\frac{(42J_n^{1/4}+16)\omega_n J_n}{{n}^{1/2}},
    \end{align*}
    which converges to $0$ by the assumption that $\frac{J_n^{5/4}\omega_n}{{n}^{1/2}}\rightarrow 0$.
\end{proof}

Next, we show a uniform version of Slutsky's theorem.
\begin{lemma}
    Let $X_n$ be a sequence of random variables that converges  uniformly in distribution over a set of distribution $\mathcal{P}$ to the distribution of a sequence of random variables $\tilde{X}_n$  in the sense that $\sup_{P\in \mathcal{P}}\sup_{t\in \mathbbm{R}}\vert P(X_n\leq t)-P(\tilde{X}_n\leq t)\vert \rightarrow 0$, and $Y_n$ be a sequence of random variables converges in probability uniformly to $0$ in the sense that $\sup_{P\in\mathcal{P}}P(\vert Y_n\vert>\epsilon)\rightarrow 0$ for any $\epsilon>0$. Suppose the following condition holds
    \begin{align}
        \lim_{\epsilon\downarrow 0} \limsup_{n\rightarrow \infty} \sup_{P\in \mathcal{P}}\sup_{t} \{P(\tilde{X}_n\leq t+\epsilon) - P(\tilde{X}_n\leq t-\epsilon)\} = 0. \label{eq: equicontinuity}
    \end{align}
    Then we have $X_n+Y_n$ converges  uniformly in distribution over a set of distributions $\mathcal{P}$ to the distribution of $\tilde{X}_n$. \label{lemma: uniform slusky}
\end{lemma}
\begin{proof}
    For any fixed $\epsilon>0$ and $t\in \mathbb{R}$, for any $P\in \mathcal{P}$, we have $\{X_n+Y_n\leq t\}\subset \{X_n\leq t+\epsilon\}\cup\{ \vert Y_n\vert >\epsilon\}$, therefore $P(X_n+Y_n\leq t)\leq P(X_n\leq t+\epsilon)+P(\vert Y_n\vert >\epsilon)$. Similarly, we have $P(X_n\leq t-\epsilon)\leq P(X_n+Y_n\leq t)+P(\vert Y_n\vert >\epsilon)$. Combining those two bounds, we have 
    \begin{equation*}
        P(X_n\leq t-\epsilon)-P(\vert Y_n\vert >\epsilon)\leq P(X_n+Y_n\leq t)\leq P(X_n\leq t+\epsilon)+P(\vert Y_n\vert >\epsilon).
    \end{equation*}
    Let 
    \begin{align*}
        &A_{n,\epsilon}=\sup_{P\in \mathcal{P}}\sup_{t\in \mathbb{R}}\max\{\vert P(X_n\leq t+\epsilon)-P(\tilde{X}_{n}\leq t)\vert,\vert P(X_n\leq t-\epsilon)-P(\tilde{X}_{n}\leq t)\vert\}\\
        &B_{n,\epsilon}=\sup_{P\in \mathcal{P}}P(\vert Y_n\vert >\epsilon).
    \end{align*}
    We have $\sup_{P\in \mathcal{P}}\sup_{t\in \mathbb{R}}\vert P(X_n+Y_n\leq t)-P(X\leq t)\vert \leq A_{n,\epsilon}+B_{n,\epsilon}$. Since $Y_n$ converges uniformly to $0$, we have $\lim_{n\rightarrow \infty}B_{n,\epsilon} = 0$. Furthermore, $A_{n,\epsilon} \leq \sup_{P\in \mathcal{P}}\sup_{t\in R}\vert P(X_n\leq t)-P(\tilde{X}_n \leq t)\vert +\sup_{P\in \mathcal{P}}\sup_{t\in R}\vert P(\tilde{X}_n\leq t+\epsilon)-P(\tilde{X}_n\leq t)\vert+\sup_{P\in \mathcal{P}}\sup_{t\in R}\vert P(\tilde{X}_n\leq t-\epsilon)-P(\tilde{X}_n\leq t)\vert$. The first term converges to $0$ by the lemma assumptions, the second term goes to $0$ when $\epsilon \rightarrow 0$ due to \eqref{eq: equicontinuity}. 
\end{proof}

\begin{remark}
    It's clear that if $\tilde{X}_n\sim N(0,1)$, condition \eqref{eq: equicontinuity} holds. Furthermore, suppose $\inf_{n,P}\max_{j} \tau_{j,n,P}>1/C$, then $G_{n,P} := \sum_{j=1}^{J_n}\tau_{j,n,P}\chi_j^2$ also satisfy the condition, where $\chi_j^2$ are independent $\chi(1)$ distributed random variables. To see this, let $\tau_* = \max_{j}\tau_{j,n,P}$, and write $G_{n,P} = \tau_*\chi_1^2+R$, where $\tau_*\chi_1^2 \indep R$. For any interval of length $2\epsilon$, 
    \begin{align*}
        P(G_{n,P} \in I) = &\mathbbm{E}[P(\tau_*\chi_1^2\in I-R|R)]\\
        \leq &P(\tau_*\chi_1^2\in [s,s+2\epsilon]).
    \end{align*}
    Furthermore, the density of $\chi_1^2$ is decreasing on $(0,\infty)$, therefore the largest probability of an interval of a fixed length $h$ occurs near zero, and
    \begin{align*}
        \sup_{s}P(\chi_1^2\in [s,s+h])\leq P(\chi_1^2\leq h) \lesssim \sqrt{h}.
    \end{align*}
    Therefore,
    \begin{align*}
        \sup_{t}P(t-\epsilon<G_{n,P}\leq t+\epsilon) \lesssim \sqrt{\epsilon/\tau_*}.
    \end{align*}
    The right hand side converges to $0$ when $\epsilon$ converges to $0$ and $\tau_*$ is bounded below by a constant $1/C$.
    \label{remark:sufficient conditions}
\end{remark}

\begin{lemma}
    Let $X_1,...,X_n$ be a fourth-moment-integrable mean zero martingale difference sequence with respect to the filtration $\{\mathcal{F}_n\}$. Let $\Phi(\cdot)$ be the distribution function of standard normal distribution. Then we have
    \begin{align*}
        \sup_{t\in \mathbb{R}}\left\vert P\left(\sum_{i=1}^n X_i\leq t\right)-\Phi(t)\right\vert \leq C(L_{n}+N_{n})^{1/5},
    \end{align*}
    where $C$ is an absolute constant,
    \begin{align*}
        &L_n = 3\sum_{i=1}^n \mathbbm{E}_P[X_i^4],\\
        &N_n = 2\mathbbm{E}_P\left[\left(\sum_{i=1}^n X_i^2-1\right)^2\right].\\
    \end{align*}
    \label{lemma: martingale bound}
\end{lemma}

\begin{proof}
    By Theorem 1 of \citet{haeusler1988rate}, we know that for any $\delta>0$,
    \begin{align*}
        \sup_{t\in \mathbb{R}}\left\vert P\left(\sum_{i=1}^n S_n\leq t\right)-\Phi(t)\right\vert \leq C_\delta(L_{n,2\delta}+N_{n,2\delta})^{1/5},
    \end{align*}
    where $C$ is an absolute constant,
    \begin{align*}
        &L_{n,2\delta} = \sum_{i=1}^n \mathbbm{E}_P[\vert X_i\vert^{2+2\delta}],\\
        &N_{n,2\delta} = \mathbbm{E}_P\left[\left(\sum_{i=1}^n \mathbbm{E}_P[X_i^2|\mathcal{F}_{i-1}]-1\right)^{1+\delta}\right].\\
    \end{align*}

    Let $\delta = 1$, then $L_{n,2\delta}=L_n$,
    \begin{align*}
        &N_{n,2\delta} = \mathbbm{E}_P\left[\left(\sum_{i=1}^n \mathbbm{E}_P[X_i^2|\mathcal{F}_{i-1}]-1\right)^{2}\right]\\
        \leq & 2\mathbbm{E}_P\left[\left(\sum_{i=1}^n X_i^2-1\right)^{2}\right]+2\mathbbm{E}_P\left[\left(\sum_{i=1}^n ( X_i^2-\mathbbm{E}_P [X_i^2|\mathcal{F}_{i-1}])\right)^{2}\right]\\
        = & 2\mathbbm{E}_P\left[\left(\sum_{i=1}^n X_i^2-1\right)^{2}\right]+2\sum_{i=1}^n\mathbbm{E}_P\left[\left(  X_i^2-\mathbbm{E}_P [X_i^2|\mathcal{F}_{i-1}]\right)^{2}\right] \\
        \leq & 2\mathbbm{E}_P\left[\left(\sum_{i=1}^n X_i^2-1\right)^{2}\right]+2\sum_{i=1}^n\mathbbm{E}_P\left[X_i^4\right].
    \end{align*}
    We use the elementary inequality $(x+y)^2\leq 2(x^2+y^2)$ for the second line. For the third line, we use the fact that $\mathbbm{E}_P[(X_i^2-\mathbbm{E}_P [X_i^2|\mathcal{F}_{i-1}])(X_j^2-\mathbbm{E}_P [X_j^2|\mathcal{F}_{j-1}])] = 0$ for $i\neq j$. The last line follows from the fact that $\mathbbm{E}_P\left[\left(  X_i^2-\mathbbm{E}_P [X_i^2|\mathcal{F}_{i-1}]\right)^{2}\right] = \text{Var}(X_i^2|\mathcal{F}_{i-1})\leq \mathbbm{E}_P[X_i^4]$. Therefore we have proved the result.
\end{proof}

\begin{lemma}
    Let $B_n(X)$ be a $J_n-$dimensional basis function such that $\mathbbm{E}_P[B_n(X)B^\top _n(X)]\prec CI_{J_n}$. Then for any function $f(X)$ such that $\mathbbm{E}_P[f^2(X)]\leq C$, and any positive semi-definite matrix $\Omega_n$ such that the maximal and minimal eigenvalues are bounded by $C$ and $1/C$, we have for all $n$,
    \begin{align*}
        \sum_{j=1}^{J_n} (\mathbbm{E}_P[f(X)\tilde{b}^\lambda_j(X)])^2\leq C'.
    \end{align*}
    for some constants $C'$, where we assume $\Omega_n$ yields spectral decomposition $\Gamma^\top _n\Lambda_n\Gamma_n$, $\tilde{B}^\lambda_n = \Lambda^{1/2}_n\Gamma_nB_n$, $\tilde{b}^\lambda_j(X)$ is the $j-$th element of $\tilde{B}^\lambda_n$.
    \label{lemma: bounded projection}
 \end{lemma}

\begin{proof}
    Let $m_j:=\mathbbm{E}_P[f(X)b_j(X)]$, $m_n = (m_1,...,m_{J_n})$. Let $c_n(X) = \sum_{j=1}^{J_n}m_jb_j(X)$. Then we have $\mathbbm{E}_P c^2_n(X) =  m_n^\top \mathbbm{E}_{P}[B_n(X)B^\top _n(X)]m_n \leq C\Vert m_n\Vert^2$. 
    \begin{align*}
        &\sum_{j=1}^{J_n}(\mathbbm{E}_{P}[f(X)b_j(X)])^2 \\
    =& \mathbbm{E}_P\left[f(X)\sum_{j=1}^{J_n}\mathbbm{E}_P[f(X)b_j(X)]b_j(X)\right] \\
    =& \mathbbm{E}_P[f(X)c_n(X)]\\
    \leq & ({\mathbbm{E}_P[f^2(X)]\mathbbm{E}_P[c^2_n(X)]})^{1/2}\\
    \lesssim & \Vert m_n\Vert 
    \end{align*}
    Therefore, we have $\sum_{j=1}^{J_n}(\mathbbm{E}_{P}[f(X)b_j(X)])^2$ can be bounded by an absolute constant. Then we have
    \begin{align*}
        &\sum_{j=1}^{J_n} (\mathbbm{E}[f(X)\tilde{b}^\lambda_j(X)])^2\\
        =&\Vert \mathbbm{E}_P[f(X)\tilde{B}^\lambda_n(X)] \Vert^2\\
        =&\Vert \Lambda^{1/2}_n\Gamma_n\mathbbm{E}_P[f(X){B}_n(X)] \Vert^2\\
        \leq &\Vert \Lambda^{1/2}_n\Gamma_n \Vert^2 \Vert\mathbbm{E}_P[f(X){B}_n(X)] \Vert^2
    \end{align*}
    Since the operator norm of $ \Lambda^{1/2}_n\Gamma_n$ is bounded by a constant $C$, we have $\sum_{j=1}^{J_n} (\mathbbm{E}_P[f(X)\tilde{b}^\lambda_j(X)])^2$ is bounded by a constant.
\end{proof}

\begin{lemma}
    Under Assumptions \ref{Assumption: basis} to \ref{Assumption: nuisance}, we have
    \begin{align*}
         \left\Vert \frac{1}{n}\sum_{i=1}^nB_n(X_i)g(O_i;\eta_{{P}}) - \frac{1}{n}\sum_{i=1}^nB_n(X_i)g(O_i;{\eta}_{\widehat{P}_n}) \right\Vert= O_{\mathcal{P}}(\zeta_n\xi_n({J_n/n})^{1/2})
    \end{align*}
    \label{lemma: second order impact of nuisance}
\end{lemma}

\begin{proof}
    $\Vert \frac{1}{n}\sum_{i=1}^nB_n(X_i)[g(O_i;\eta_{{P}})-g(O_i;\eta_{\widehat{P}_n})]\Vert$ can be written as
    \begin{align*}
        &\left\Vert \frac{1}{n}\sum_{i=1}^nB_n(X_i)[g(O_i;\eta_{{P}})-g(O_i;\eta_{\widehat{P}_n})]\right\Vert \\
        =&\left\Vert \sum_{l=1}^L\frac{n_l}{n}\frac{1}{n_l}\sum_{i \in I_l}B_n(X_i)[g(O_i;\eta_{{P}})-g(O_i;\eta_{\widehat{P}_n})]\right\Vert  \\
        \leq &\sum_{l=1}^L \left\Vert \frac{n_l}{n}\frac{1}{n_l}\sum_{i \in I_l}B_n(X_i)[g(O_i;\eta_{{P}})-g(O_i;\eta_{\widehat{P}_n})]\right\Vert.
    \end{align*}
    Therefore, we only need to derive the bound for each fold.
    \begin{align*}
       &\left\Vert \frac{1}{n_l}\sum_{i \in I_l}B_n(X_i)[g(O_i;\eta_{{P}})-g(O_i;\eta_{\widehat{P}_n})]\right\Vert\\
       =&\left\Vert P_{n,l}[B_ng(\eta_{{P}})-B_ng(\eta_{\widehat{P}_n})]\right\Vert\\
       \leq &\left\Vert (P_{n,l}-P)[B_ng(\eta_{{P}})-B_ng(\eta_{\widehat{P}_n})]\right\Vert + \left\Vert P[B_ng(\eta_{{P}})-B_ng(\eta_{\widehat{P}_n})]\right\Vert
    \end{align*}
    Since $\eta_{\widehat{P}_{n,l}}$ is fitted using the data outside the fold $l$, it can be treated as fixed (otherwise the following derivation should be interpreted as conditioning on the other folds). For the first term, we have with probability approaching one, 
    \begin{align*}
        &\mathbbm{E}_P\left\Vert (P_{n,l}-P)[B_ng(\eta_{{P}})-B_ng(\eta_{\widehat{P}_{n,l}})]\right\Vert^2\\
        =&\sum_{j=1}^{J_n} \mathbbm{E}_P[(P_{n,l}-P)b_j[g(\eta_P)-g(\eta_{\widehat{P}})]]^2 \\
        = & \sum_{j=1}^{J_n}\text{Var}(P_{n,l}b_j(X)[g(O;\eta_P)-g(O;\eta_{\widehat{P}})])\\
        = & \frac{1}{n_l}\sum_{j=1}^{J_n}\text{Var}(b_j(X)[g(O;\eta_P)-g(O;\eta_{\widehat{P}})])\\
        \leq & \frac{1}{n_l}\sum_{j=1}^{J_n}\mathbbm{E}_Pb^2_j(X)[g(O;\eta_P)-g(O;\eta_{\widehat{P}})]^2\\
        \leq & \frac{1}{n_l}\sum_{j=1}^{J_n}\mathbbm{E}_P[g(O;\eta_P)-g(O;\eta_{\widehat{P}})]^2\xi_n^2\\
        \leq & \frac{J_n}{n_l}\sup_{\eta\in \mathcal{T}_n}\mathbbm{E}_P[g(O;\eta_P)-g(O;\eta)]^2\xi_n^2
        \\
        \leq & \frac{J_n}{n_l} \zeta_n^2\xi_n^2
    \end{align*}
    Where the second inequality is due to boundedness of $b_j(X)$, the third inequality is due to the definition of the realization set. Therefore, we have 
    \begin{align*}
        \left\Vert (P_{n,l}-P)[B_ng(\eta_{{P}})-B_ng(\eta_{\widehat{P}_n})]\right\Vert = O_{\mathcal{P}}(\zeta_n\xi_n({J_n/n})^{1/2}).
    \end{align*}
    For the second term, we have with probability approaching one,
    \begin{align*}
        &\left\Vert P[B_ng(\eta_{{P}})-B_ng(\eta_{\widehat{P}_n})]\right\Vert\\
        =& \left({\sum_{j=1}^{J_n}(P[b_jg(\eta_{{P}})-b_jg(\eta_{\widehat{P}_n})])^2}\right)^{1/2} \\
        \leq & {(\sum_{j=1}^{J_n}\sup_{\eta\in \mathcal{T}_n}(\mathbbm{E}_P[b_j(X)g(O;\eta_{{P}})-b_j(X)g(O;\eta)])^2)^{1/2}}
    \end{align*}
    For a fixed $\eta\in \mathcal{T}_n$. Define the function $h_{j,P}(t;\eta) = \mathbbm{E}[b_j(X)g(O;(1-t)\eta_P+t\eta)], t \in [0,1]$, then we have 
    \begin{align*}
        &h_{j,P}(1;\eta) = \mathbbm{E}_P[b_j(X)g(O;\eta)]\\
        &h_{j,P}(0;\eta) = \mathbbm{E}_P[b_j(X)g(O;\eta_P)]\\
        &h_{j,P}(1;\eta) = h_{j,P}(0;\eta)+ \frac{\partial}{\partial t}h_{j,P}(0;\eta) + \frac{1}{2}\frac{\partial^2}{\partial t^2}h_{j,P}(\tilde{t};\eta),
    \end{align*}
    for some  $\tilde{t}\in [0,1]$. Then we have
    \begin{align*}
        &\frac{\partial}{\partial t}h_{j,P}(0;\eta) \\
        =&\frac{\partial}{\partial t}\mathbbm{E}_P[b_j(X)g(O;(1-t)\eta_P+t\eta)]\Big\vert_{t=0}\\
        =&\mathbbm{E}_P\left[b_j(X)\frac{\partial}{\partial t}\mathbbm{E}_P[g(O;(1-t)\eta_P+t\eta)|X]\Big\vert_{t=0}\right] = 0,
    \end{align*}
    and
    \begin{align*}
        &\left\vert \frac{\partial^2}{\partial t^2}h_{j,P}(\tilde{t};\eta)\right\vert \\
        =&\left\vert \mathbbm{E}_P\left[b_j(X)\frac{\partial^2}{\partial t^2}\mathbbm{E}_P[g(O;(1-t)\eta_P+t\eta)|X]\Big\vert_{t=\tilde{t}}\right]\right\vert \leq \zeta_n\xi_n/{n}^{1/2}
    \end{align*}
    by Assumption \ref{Assumption: nuisance}. Therefore,
    \begin{align*}
        & \sup_{\eta\in \mathcal{T}_n}(\mathbbm{E}_P[b_j(X)g(O;\eta_{{P}})-b_j(X)g(O;\eta)])^2\leq  \zeta_n^2\xi^2_n/n
    \end{align*}
    and
    \begin{align*}
        \left\Vert P[B_ng(\eta_{{P}})-B_ng(\eta_{\widehat{P}_n})]\right\Vert \leq \zeta_n\xi_n {(J_n/n)^{1/2}}.
    \end{align*}

    Therefore, we have proved that $\Vert \widehat{a}_n-a_n\Vert  = O_{\mathcal{P}}(\zeta_n\xi_n {(J_n/n)^{1/2}})$.
\end{proof}

% \begin{lemma}
%     Under Assumption \ref{Assumption: basis} to \ref{Assumption: nuisance}, we have
%     \begin{align*}
%          \left\Vert \frac{1}{n}\sum_{i=1}^nB_n(X_i)g(O_i;\eta_{{P}}) - \frac{1}{n}\sum_{i=1}^nB_n(X_i)g(O_i;{\eta}_{\widehat{P}_n}) \right\Vert_\infty = o_{\mathcal{P}}(\zeta_n\xi_n{/n^{1/2}})
%     \end{align*}
    
%     \label{lemma: bounding the sup norm}
% \end{lemma}
% \begin{proof}
%     From the previous proof, we know that with probability approaching 1, we have
%     \begin{align*}
%         \left\vert\frac{1}{n}\sum_{i=1}^nb_j(X_i)g(O_i;\eta_{{P}}) - \frac{1}{n}\sum_{i=1}^nb_j(X_i)g(O_i;{\eta}_{\widehat{P}_n}) \right\vert\leq C\zeta_n\xi_n/{n}^{1/2}
%     \end{align*}
%     where $C$ is a constant which does not depend on $j,P$ and $n$. Since $B_n(X) = (b_1(X),...,b_{J_n}(X))^\top$, we know that
%     \begin{align*}
%          \left\Vert \frac{1}{n}\sum_{i=1}^nB_n(X_i)g(O_i;\eta_{{P}}) - \frac{1}{n}\sum_{i=1}^nB_n(X_i)g(O_i;{\eta}_{\widehat{P}_n}) \right\Vert_\infty \leq C\zeta_n\xi_n/{n}^{1/2},
%     \end{align*}
%     therefore the argument holds.

% \end{proof}

\begin{lemma}
    Under the null hypothesis, Assumptions \ref{Assumption: basis} to \ref{Assumption: nuisance}, we have
    \begin{align*}
        \left\Vert \frac{1}{n}\sum_{i=1}^n B_n(X_i)g(O_i;\eta_P)\right\Vert = O_{\mathcal{P}_0}({(J_n/n)^{1/2}})
    \end{align*}
    \label{Lemma: a_n}
\end{lemma}
\begin{proof}
Direct calculation gives
    \begin{align*}
        &\mathbbm{E}_P\left\Vert \frac{1}{n}\sum_{i=1}^n B_n(X_i)g(O_i;\eta_P)\right\Vert^2\\
        =&\frac{1}{n}\mathbbm{E}_P g^2(O_i;\eta_P)B^\top_n(X_i)B_n(X_i)\\
        \lesssim &\frac{1}{n}\mathbbm{E}_PB^\top_n(X_i)B_n(X_i)\\
        \lesssim & J_n/n,
    \end{align*}
    where the first equality is due to the fact that under the null hypothesis, $B_n(X_i)g(O_i;\eta_P)$ is mean zero, the second inequality is because $\mathbbm{E}_P[g^2(O_i;\eta_P)|X]<C$, the third inequality is due to uniform boundedness of the basis functions. Therefore we have $\left\Vert \frac{1}{n}\sum_{i=1}^n B_n(X_i)g(O_i;\eta_P)\right\Vert = O_p({(J_n/n)^{1/2}})$.
\end{proof}

% The proof of the next lemma is similar to the proof of Lemma \ref{lemma: bounding the sup norm}, thus omitted.
% \begin{lemma}
%     Under \ref{Assumption: basis} to \ref{Assumption: nuisance}, we have
%     \begin{align*}
%         \left\Vert \frac{1}{n}\sum_{i=1}^n B_n(X_i)g(O_i;\eta_P)\right\Vert_\infty = O_p({\xi_n(1/n)^{1/2}})
%     \end{align*}
% \end{lemma}

\begin{lemma}
    Under Assumption \ref{Assumption: basis} to \ref{Assumption: nuisance} and the null hypothesis, let 
    \begin{align*}
        S_n = n{\left(\frac{1}{n}\sum_{i=1}^nB_n(X_i)g(O_i;\eta_{{P}})\right)^\top}{{\Omega}_n}{\left(\frac{1}{n}\sum_{i=1}^nB_n(X_i)g(O_i;\eta_{{P}})\right)}
    \end{align*}
    we have
    \begin{enumerate}
        \item For $\zeta_n\xi_n = o(1/J_n)$, $\delta_n{J^{1/2}_n} = o(1)$,we have $\widehat{S}_n-S_n = o_{\mathcal{P}_0}(1)$.
        \item Suppose $\widehat{\Omega} = \Omega_n = \text{diag}(\lambda_1,...,\lambda_{J_n})$ and $\zeta_n\xi_n = o(1/\text{tr}(\Omega_n))$, we have $\widehat{S}_n-S_n = o_{\mathcal{P}_0}(1)$.
        \item For $\zeta_n\xi_n = o(1/{J^{1/2}_n})$, $\delta_n = o(1)$, we have $\widehat{S}_n-S_n = o_{\mathcal{P}_0}({J_n}^{1/2})$.
    \end{enumerate}
    \label{lemma: convergence rate of S_n}
\end{lemma}

\begin{proof}
    For the convenience of notation, let $a_n = \frac{1}{n}\sum_{i=1}^nB_n(X_i)g(O_i;\eta_{{P}})$ and $\widehat{a}_n = \frac{1}{n}\sum_{i=1}^nB_n(X_i)g(O_i;{\eta}_{\widehat{P}_n})$. Then
    \begin{align*}
        &\widehat{a}_n^\top\widehat{\Omega}_n\widehat{a}_n - {a}_n^\top{\Omega}_n{a}_n\\
         = & (\widehat{a}^\top_n-a^\top_n)(\widehat{\Omega}_n-\Omega_n)(\widehat{a}_n-a_n)+(\widehat{a}_n^\top -a_n)\Omega_n(\widehat{a}_n-a_n)\\
         &+2(\widehat{a}^\top_n-a^\top_n)(\widehat{\Omega}_n-\Omega_n)a_n+2(\widehat{a}^\top_n-a^\top_n)\Omega_na_n+a_n^\top (\widehat{\Omega}_n-\Omega_n)a_n.
    \end{align*}

    Therefore, the rate is dominated by the last two terms. We have with probability approaching one,
    \begin{align*}
        &(\widehat{a}^\top_n-a^\top_n)\Omega_n a_n \leq \Vert \widehat{a}^\top_n-a^\top_n\Vert \Vert \Omega_n\Vert \Vert a_n\Vert  = O_{\mathcal{P}_0}(\zeta_n \xi_n J_n/n) \\
 & a_n^\top (\widehat{\Omega}_n-\Omega_n)a_n \leq \Vert a_n\Vert^2 \Vert \widehat{\Omega}_n-\Omega_n\Vert = O_{\mathcal{P}_0}(\delta_n J^{1/2}_n/n)
    \end{align*}
    Therefore, when $\zeta_n\xi_n = o(1/{J_n})$ and $\delta_n = o(1/J^{1/2}_n)$, we have $n(\widehat{a}_n^\top\widehat{\Omega}_n\widehat{a}_n - {a}_n^\top{\Omega}_n{a}_n) = o(1)$, when $\zeta_n\xi_n = o(1/{J_n}^{1/2})$ and $\delta_n = o(1)$, we have $n(\widehat{a}_n^\top\widehat{\Omega}_n\widehat{a}_n - {a}_n^\top{\Omega}_n{a}_n) = o({J_n}^{1/2})$.

    When $\widehat{\Omega}_n = {\Omega}_n$ is a diagonal matrix with positive entries, the leading term is only $(\widehat{a}^\top_n-a^\top_n)\Omega_na_n$, we have
    \begin{align*}
        (\widehat{a}^\top_n-a^\top_n)\Omega_na_n \leq [(\widehat{a}^\top_n-a^\top_n)\Omega_n(\widehat{a}^\top_n-a^\top_n)]^{1/2}(a_n^\top\Omega_na_n)^{1/2}.
    \end{align*}
    Under the null $\mathbbm{E}_P(a_n^\top\Omega_na_n) = \frac{1}{n}\sum_{j=1}^{J_n}\lambda_j\mathbbm{E}[g^2b_j^2]\lesssim \frac{\text{trace}(\Omega_n)}{n}$. Furthermore, we have $(\widehat{a}^\top_n-a^\top_n)\Omega_n(\widehat{a}^\top_n-a^\top_n) = O_P(\frac{\zeta_n^2\xi_n^2\text{trace}(\Omega_n)}{n})$.
    Therefore, when $\zeta_n\xi_n = o(1/\text{trace}(\Omega_n))$, we have $n(\widehat{a}_n^\top\widehat{\Omega}_n\widehat{a}_n - {a}_n^\top{\Omega}_n{a}_n) = o_\mathcal{P}(1)$
\end{proof}

\begin{lemma}[Convergence rates of centering and scaling parameters]
\label{lem:centering-scaling}
Under the Assumptions of Theorem \ref{Theorem: asymptotic distribution under the null},
\[
\frac{\widehat{\rho}_n-\rho_{n,P}}{\gamma_{n,P}}
=o_{\mathcal P}(1),
\qquad
|\widehat{\gamma}_n-\gamma_{n,P}|
=o_{\mathcal P}(1).
\]
\end{lemma}

\begin{proof}
    Recall that $\Sigma^{1/2}_{n,P}
=
\Omega_{n,P}^{1/2}
G_{n,P}
(\Omega_{n,P}^{1/2})^\top,
\widehat{\Sigma}_n
=
\widehat{\Omega}_n^{1/2}
\widehat G_n
(\widehat{\Omega}_n^{1/2})^\top$, where 
\begin{align*}
    &G_{n,P}
=
\mathbbm{E}_P\left[
g^2(O;\eta_P)B_n(X)B_n(X)^\top
\right]\\
& \widehat G_n
=
\frac1n\sum_{i=1}^n
g^2(O_i;\widehat\eta_{n,l(i)})
B_n(X_i)B_n(X_i)^\top.
\end{align*}
Here $l(i)$ denotes the fold containing observation $i$, and
$\widehat\eta_{n,l}$ is estimated using observations outside fold $l$. We first show that
\begin{equation}
\label{eq:G-convergence}
\|\widehat G_n-G_{n,P}\|_F
=
o_{\mathcal P}(1).
\end{equation}

Let $\mathcal I_1,\ldots,\mathcal I_L$ denote the cross-fitting folds,
with $n_l=|\mathcal I_l|$ and $n_l\asymp n$, where $L$ is fixed.
Let $P_{n,l}$ denote the empirical measure over fold $\mathcal I_l$. For notational convenience, write $g_P(O)=g(O;\eta_P)$ and $\widehat g_l(O)=g(O;\widehat\eta_{n,l})$. Then
\begin{align}
\widehat G_n-G_{n,P}
={}&
\sum_{l=1}^L\frac{n_l}{n}
(P_{n,l}-P)
\left[
(\widehat g_l^2-g_P^2)B_nB_n^\top
\right]
\label{eq:G-decomp1}\\
&+
(P_n-P)
\left[
g_P^2B_nB_n^\top
\right]
\label{eq:G-decomp2}\\
&+
\sum_{l=1}^L\frac{n_l}{n}
P\left[
(\widehat g_l^2-g_P^2)B_nB_n^\top
\right].
\label{eq:G-decomp3}
\end{align}
We bound the three terms separately.

First consider \eqref{eq:G-decomp1}. By Assumption \ref{Assumption: nuisance}, with probability
tending to one uniformly over $P\in\mathcal P_0$, all fold-specific
nuisance estimators belong to $\mathcal T_n$. Conditional on the
training sample outside fold $l$, $\widehat\eta_{n,l}$ is fixed and the
observations in $\mathcal I_l$ are independent of
$\widehat\eta_{n,l}$. Hence, on this event,
\begin{align*}
&\mathbbm{E}_P\left[
\left\|
(P_{n,l}-P)
\left[
(\widehat g_l^2-g_P^2)B_nB_n^\top
\right]
\right\|_F^2
\,\middle|\,
\widehat\eta_{n,l}
\right]\\
&\qquad
=
\frac1{n_l}
\sum_{j,k=1}^{J_n}
\text{Var}_P\left(
(\widehat g_l^2-g_P^2)b_j(X)b_k(X)
\,\middle|\,
\widehat\eta_{n,l}
\right)\\
&\qquad
\le
\frac1{n_l}
\sum_{j,k=1}^{J_n}
\mathbbm{E}_P\left[
(\widehat g_l-g_P)^2
(\widehat g_l+g_P)^2
b_j^2(X)b_k^2(X)
\,\middle|\,
\widehat\eta_{n,l}
\right]\\
&\qquad
\le
\frac{J_n^2\xi_n^4}{n_l}
\mathbbm{E}_P\left[
(\widehat g_l-g_P)^2
(\widehat g_l+g_P)^2
\,\middle|\,
\widehat\eta_{n,l}
\right]\\
&\qquad
\lesssim
\frac{J_n^2\xi_n^4\zeta_n^2}{n_l},
\end{align*}
where the last inequality follows from Assumption \ref{Assumption: nuisance}. Therefore,
uniformly over $P\in\mathcal P_0$,
\[
\left\|
(P_{n,l}-P)
\left[
(\widehat g_l^2-g_P^2)B_nB_n^\top
\right]
\right\|_F
=
O_{\mathcal P_0}\left(
\frac{J_n\xi_n^2\zeta_n}{\sqrt n}
\right).
\]
Moreover,
\[
\frac{J_n\xi_n^2\zeta_n}{\sqrt n}
=
\left(\sqrt{J_n}\zeta_n\xi_n\right)
\left(\frac{\sqrt{J_n}\xi_n}{\sqrt n}\right)
=o(1),
\]
because Assumption \ref{Assumption: nuisance}(a) gives
$\sqrt{J_n}\zeta_n\xi_n=o(1)$, Assumption \ref{Assumption: basis} gives
$\xi_n\lesssim J_n^{1/2}$, and the assumptions of Theorem \ref{Theorem: asymptotic distribution under the null} imply
$J_n=o(n^{1/3})$. Since $L$ is fixed,
\begin{equation}
\label{eq:G-term1}
\left\|
\sum_{l=1}^L\frac{n_l}{n}
(P_{n,l}-P)
\left[
(\widehat g_l^2-g_P^2)B_nB_n^\top
\right]
\right\|_F
=
o_{\mathcal P_0}(1).
\end{equation}

Next consider the oracle empirical-process term \eqref{eq:G-decomp2}.
By independence,
\begin{align*}
&\mathbbm{E}_P\left\|
(P_n-P)
[g_P^2B_nB_n^\top]
\right\|_F^2\\
&=
\frac1n
\sum_{j,k=1}^{J_n}
\text{Var}_P\left(
g_P^2b_j(X)b_k(X)
\right)\\
&\le
\frac1n
\mathbbm{E}_P\left[
g_P^4
\|B_n(X)B_n(X)^\top\|_F^2
\right]\\
&=
\frac1n
\mathbbm{E}_P\left[
g_P^4\|B_n(X)\|^4
\right]\\
&\le
\frac{\omega_n^2}{n}
\mathbbm{E}_P\left[
g_P^4\|B_n(X)\|^2
\right]\\
&=
\frac{\omega_n^2}{n}
\mathbbm{E}_P\left[
\mathbbm{E}_P(g_P^4\mid X)\|B_n(X)\|^2
\right]\\
&\lesssim
\frac{\omega_n^2J_n}{n},
\end{align*}
where the last inequality follows from Assumptions \ref{Assumption: basis} and \ref{Assumption: nuisance}. Hence
\[
\left\|
(P_n-P)[g_P^2B_nB_n^\top]
\right\|_F
=
O_{\mathcal P}\left(
\omega_n\sqrt{\frac{J_n}{n}}
\right).
\]
Since $\omega_n\lesssim J_n$ and $J_n=o(n^{1/3})$, $\omega_n\sqrt{\frac{J_n}{n}}
\lesssim
\sqrt{\frac{J_n^3}{n}}
=o(1).$
Therefore, 
\begin{equation}
\label{eq:G-term2}
\left\|
(P_n-P)[g_P^2B_nB_n^\top]
\right\|_F
=
o_{\mathcal P_0}(1).
\end{equation}

Finally, consider \eqref{eq:G-decomp3}. For any $u\in\mathbb R^{J_n}$
with $\|u\|=1$, on the event $\widehat\eta_{n,l}\in\mathcal T_n$,
\begin{align*}
&\left|
u^\top
P\left[
(\widehat g_l^2-g_P^2)
B_nB_n^\top
\right]
u
\right|\\
&\qquad
=
\left|
\mathbbm{E}_P\left[
(\widehat g_l-g_P)
(\widehat g_l+g_P)
\{u^\top B_n(X)\}^2
\right]
\right|\\
&\qquad
\le
\left\{
\mathbbm{E}_P[
(\widehat g_l-g_P)^2
(\widehat g_l+g_P)^2]
\right\}^{1/2}
\left\{
\mathbbm{E}_P[
\{u^\top B_n(X)\}^4]
\right\}^{1/2}.
\end{align*}
By Assumption 3, the first factor is bounded by $C\zeta_n$.
For the second factor,
\[
\{u^\top B_n(X)\}^4
\le
\|B_n(X)\|^2\{u^\top B_n(X)\}^2
\le
\omega_n^2\{u^\top B_n(X)\}^2,
\]
and Assumption 1 implies $\mathbbm{E}_P[\{u^\top B_n(X)\}^2]
=
u^\top \mathbbm{E}_P[B_n(X)B_n(X)^\top]u
\lesssim1.$ Therefore,
\[
\left\|
P\left[
(\widehat g_l^2-g_P^2)
B_nB_n^\top
\right]
\right\|
\lesssim
\zeta_n\omega_n.
\]

Since for a $J_n\times J_n$ matrix $A$,
$\|A\|_F\le\sqrt{J_n}\|A\|$, we obtain
\[
\left\|
P\left[
(\widehat g_l^2-g_P^2)
B_nB_n^\top
\right]
\right\|_F
\lesssim
\sqrt{J_n}\zeta_n\omega_n
=o(1)
\]
by Assumption \ref{Assumption: nuisance}(a). Therefore we have,
\begin{equation}
\label{eq:G-term3}
\left\|
\sum_{l=1}^L\frac{n_l}{n}
P\left[
(\widehat g_l^2-g_P^2)
B_nB_n^\top
\right]
\right\|_F
=
o_{\mathcal P_0}(1).
\end{equation}
Combining \eqref{eq:G-term1}--\eqref{eq:G-term3} proves
\eqref{eq:G-convergence}.

We next collect several consequences of the assumptions. By
Assumptions 1 and 3, for every unit vector $u$, $u^\top G_{n,P}u
=
\mathbbm{E}_P\left[
\mathbbm{E}_P\{g_P^2\mid X\}
\{u^\top B_n(X)\}^2
\right].$ Since $\mathbbm{E}_P[g_P^2\mid X]$ is uniformly bounded above and below away
from zero, Assumption 1 implies that for some constants
$0<c<C<\infty$,
\begin{equation}
\label{eq:G-eigen}
cI_{J_n}
\preceq
G_{n,P}
\preceq
CI_{J_n}
\end{equation}
uniformly over $P\in\mathcal P_0$. Thus $\|G_{n,P}\|=O(1),
\|G_{n,P}\|_F=O(\sqrt{J_n})$. Together with \eqref{eq:G-convergence}, $\|\widehat G_n\|
=
O_{\mathcal P}(1),
\|\widehat G_n\|_F
=
O_{\mathcal P}(\sqrt{J_n}).$ Similarly, Assumption 2 gives $\|\widehat\Omega_n-\Omega_{n,P}\|_F
\le
\sqrt{J_n}
\|\widehat\Omega_n-\Omega_{n,P}\|
\le
\delta_n$ with probability at least $1-\Delta_n$. Since
$\delta_n,\Delta_n\to0$,
\begin{equation}
\label{eq:Omega-convergence}
\|\widehat\Omega_n-\Omega_{n,P}\|_F
=
o_{\mathcal P}(1).
\end{equation}
Moreover, $\|\Omega_{n,P}\|=O(1),
\|\Omega_{n,P}\|_F=O(\sqrt{J_n}),$ and therefore $\|\widehat\Omega_n\|
=
O_{\mathcal P_0}(1),
\|\widehat\Omega_n\|_F
=
O_{\mathcal P_0}(\sqrt{J_n}).$

Under the additional eigenvalue condition in Theorem 2,
$\lambda_{\min}(\Omega_{n,P})$ is uniformly bounded away from zero.
Combining this condition with \eqref{eq:G-eigen} yields
\begin{equation}
\label{eq:gamma-order}
\gamma_{n,P}^2
=
\text{trace}(\Omega_{n,P}G_{n,P}\Omega_{n,P}G_{n,P})
\asymp
J_n,
\end{equation}
uniformly over $P\in\mathcal P_0$. In particular, $\gamma_{n,P}\asymp\sqrt{J_n}.$

We now establish the result for $\rho_{n,P}$. By cyclicity of the
trace and the definition of the matrix square root used in the main
text, $\rho_{n,P}
=
\text{trace}(\Sigma_{n,P})
=
\text{trace}(\Omega_{n,P}G_{n,P})$, $\widehat\rho_n
=
\text{trace}(\widehat\Omega_n\widehat G_n)$. Therefore,
\begin{align*}
|\widehat\rho_n-\rho_{n,P}|
&\le
\left|
\text{trace}\{(\widehat\Omega_n-\Omega_{n,P})\widehat G_n\}
\right|
+
\left|
\text{trace}\{\Omega_{n,P}(\widehat G_n-G_{n,P})\}
\right|\\
&\le
\|\widehat\Omega_n-\Omega_{n,P}\|_F
\|\widehat G_n\|_F
+
\|\Omega_{n,P}\|_F
\|\widehat G_n-G_{n,P}\|_F\\
&=
o_{\mathcal P}(\sqrt{J_n}),
\end{align*}
where we used \eqref{eq:G-convergence} and
\eqref{eq:Omega-convergence}. Together with
\eqref{eq:gamma-order}, this gives $\frac{\widehat\rho_n-\rho_{n,P}}{\gamma_{n,P}}
=
o_{\mathcal P}(1).$

It remains to establish consistency of $\widehat\gamma_n$. Since
$\Sigma_{n,P}$ and $\widehat\Sigma_n$ are symmetric, $\gamma_{n,P}^2
=
\|\Sigma_{n,P}\|_F^2
=
\text{trace}(\Omega_{n,P}G_{n,P}\Omega_{n,P}G_{n,P}),$ and similarly, $\widehat\gamma_n^2
=
\text{trace}(\widehat\Omega_n\widehat G_n
\widehat\Omega_n\widehat G_n).$ A telescoping expansion yields
\begin{align*}
\widehat\gamma_n^2-\gamma_{n,P}^2
={}&
\text{trace}\left[
(\widehat\Omega_n-\Omega_{n,P})
\widehat G_n\widehat\Omega_n\widehat G_n
\right]\\
&+
\text{trace}\left[
\Omega_{n,P}
(\widehat G_n-G_{n,P})
\widehat\Omega_n\widehat G_n
\right]\\
&+
\text{trace}\left[
\Omega_{n,P}G_{n,P}
(\widehat\Omega_n-\Omega_{n,P})
\widehat G_n
\right]\\
&+
\text{trace}\left[
\Omega_{n,P}G_{n,P}\Omega_{n,P}
(\widehat G_n-G_{n,P})
\right].
\end{align*}
Using $|\text{trace}(A^\top B)|\le\|A\|_F\|B\|_F$ and
$\|AB\|_F\le\min\{\|A\|\|B\|_F,\|A\|_F\|B\|\}$,
the four terms above are respectively bounded by
\begin{align*}
&\|\widehat\Omega_n-\Omega_{n,P}\|_F
\|\widehat G_n\|
\|\widehat\Omega_n\|
\|\widehat G_n\|_F,\\
&\|\widehat G_n-G_{n,P}\|_F
\|\widehat\Omega_n\|
\|\widehat G_n\|
\|\Omega_{n,P}\|_F,\\
&\|\widehat\Omega_n-\Omega_{n,P}\|_F
\|\widehat G_n\|
\|\Omega_{n,P}\|
\|G_{n,P}\|_F,\\
&\|\widehat G_n-G_{n,P}\|_F
\|\Omega_{n,P}\|^2
\|G_{n,P}\|_F.
\end{align*}
By \eqref{eq:G-convergence}, \eqref{eq:Omega-convergence}, and the
norm bounds above, each term is
$o_{\mathcal P}(\sqrt{J_n})$. Hence
\[
|\widehat\gamma_n^2-\gamma_{n,P}^2|
=
o_{\mathcal P}(\sqrt{J_n}).
\]
Finally,
\[
|\widehat\gamma_n-\gamma_{n,P}|
=
\frac{
|\widehat\gamma_n^2-\gamma_{n,P}^2|
}{
\widehat\gamma_n+\gamma_{n,P}
}
\le
\frac{
|\widehat\gamma_n^2-\gamma_{n,P}^2|
}{
\gamma_{n,P}
}
=
o_{\mathcal P}(1),
\]
where the last equality follows from
$\gamma_{n,P}\asymp\sqrt{J_n}$. This completes the proof.
\end{proof}

\section{Proofs of the main results}
\label{sec:main results proof}
\subsection{Proof of Theorem \ref{Theorem: asymptotic distribution chi-sqaure approximation}}

\begin{proof}
    For notational convenience, write $g_i = g^\perp(O_i;\eta_P),
a_n = \frac{1}{n}\sum_{i=1}^n B_n(X_i)g_i$ and define the oracle statistic $S_n = n a_n^\top \Omega_{n,P} a_n.$ We first show that replacing the population nuisance parameter and
weighting matrix by their estimators has an asymptotically negligible
effect. By Assumption \ref{Assumption: nuisance}(b), $\zeta_n\xi_n=o(J_n^{-1}),$ and, by the additional condition in the theorem, $\delta_n J_n^{1/2}=o(1).$ Therefore, Lemma \ref{lemma: convergence rate of S_n}(1) gives $\widehat S_n-S_n=o_{\mathcal P_0}(1)$.

We next derive the asymptotic chi-square approximation of $S_n$. For each
$P\in\mathcal P_0$, let $\Omega_{n,P}
=
\Gamma_{n,P}^{\top}\Lambda_{n,P}\Gamma_{n,P}$ be its spectral decomposition, where $\Gamma_{n,P}$ is orthogonal and
$\Lambda_{n,P}$ is diagonal with positive diagonal entries. Define $A_{n,P}
=
\Lambda_{n,P}^{1/2}\Gamma_{n,P},
\widetilde B_{n,P}^{\lambda}(X)
=
A_{n,P}B_n(X).$ Since $A_{n,P}^{\top}A_{n,P}
=
\Omega_{n,P},$ we have
\begin{align*}
S_n
&=
n
\left\{
\frac{1}{n}
\sum_{i=1}^n B_n(X_i)g_i
\right\}^{\top}
\Omega_{n,P}
\left\{
\frac{1}{n}
\sum_{i=1}^n B_n(X_i)g_i
\right\}
\\
&=
\left\|
\frac{1}{\sqrt n}
\sum_{i=1}^n
g_i
\widetilde B_{n,P}^{\lambda}(X_i)
\right\|^2.
\end{align*}
Let
\[
W_{n,P}
=
\frac{1}{\sqrt n}
\sum_{i=1}^n
g_i\widetilde B_{n,P}^{\lambda}(X_i).
\]
Then $S_n=W_{n,P}^{\top}W_{n,P}$. Furthermore, the covariance
matrix of each summand is
\[
\Sigma_{n,P}
=
\mathbbm{E}_P\left[
g_i^2
\widetilde B_{n,P}^{\lambda}(X_i)
\{\widetilde B_{n,P}^{\lambda}(X_i)\}^{\top}
\right].
\]
Let $\tau_{1,n,P},\ldots,\tau_{J_n,n,P}$ denote the eigenvalues of
$\Sigma_{n,P}$.

By Lemma \ref{lemma: chi-square approximation} and the assumption $\frac{J_n^{5/4}\omega_n}{\sqrt n}\longrightarrow0,$ we have
\[
\sup_{P\in\mathcal P_0}
\sup_{t\in\mathbb R}
\left|
P(S_n\le t)
-
P\left(
\sum_{j=1}^{J_n}
\tau_{j,n,P}\chi_j^2(1)
\le t
\right)
\right|
\longrightarrow0.
\]

It remains to verify the anti-concentration condition \eqref{eq: equicontinuity} required by
Lemma \ref{lemma: uniform slusky}. Define $G_{n,P}
=
\mathbbm{E}_P\left[
g_i^2 B_n(X_i)B_n(X_i)^\top
\right].$ Under the null hypothesis, $\mathbbm{E}_P[g_i\mid X_i]=0,$ and Assumptions 1 and 3 imply that, uniformly over $P\in\mathcal P_0$,
$
\mathbbm{E}_P[g_i^2\mid X_i]\geq c_1$, for some constant $c_1>0$, and $\mathbbm{E}_P[B_n(X)B_n(X)^\top]\succeq c_2 I_{J_n}$ for some constant $c_2>0$. Therefore
\begin{align*}
G_{n,P}
&=
\mathbbm{E}_P\left[
\mathbbm{E}_P(g_i^2\mid X_i)
B_n(X_i)B_n(X_i)^\top
\right]
\\
&\succeq
c_1
\mathbbm{E}_P\left[
B_n(X_i)B_n(X_i)^\top
\right]
\\
&\succeq
c_1c_2 I_{J_n}.
\end{align*}
Therefore, for some constant $c>0$, $G_{n,P}\succeq cI_{J_n}$ uniformly over $P\in\mathcal P_0$.

Since $\Sigma_{n,P}
=
A_{n,P}G_{n,P}A_{n,P}^{\top},$ it follows that $\Sigma_{n,P}
\succeq
cA_{n,P}A_{n,P}^{\top}
=
c\Lambda_{n,P}.$ Consequently, $\lambda_{\max}(\Sigma_{n,P})
\ge
c\,\lambda_{\max}(\Lambda_{n,P})
=
c\,\lambda_{\max}(\Omega_{n,P}).$ By the assumption $\inf_{P\in\mathcal P_0}
\lambda_{\max}(\Omega_{n,P})
\ge c_0$ for some constant $c_0>0$, we obtain $\inf_{P\in\mathcal P_0}
\max_{1\le j\le J_n}\tau_{j,n,P}
\ge cc_0>0.$ Now define $\widetilde S_{n,P}
=
\sum_{j=1}^{J_n}
\tau_{j,n,P}\chi_j^2(1),$ where the $\chi_j^2(1)$ variables are independent. By 
Remark \ref{remark:sufficient conditions},
\[
\lim_{\epsilon\downarrow0}
\limsup_{n\to\infty}
\sup_{P\in\mathcal P_0}
\sup_{t\in\mathbb R}
P\left(
t-\epsilon
<
\widetilde S_{n,P}
\le
t+\epsilon
\right)
=0.
\]
Thus the sequence of approximating weighted chi-square
distributions satisfies the uniform anti-concentration condition in
Lemma \ref{lemma: uniform slusky}.

Finally, apply Lemma \ref{lemma: uniform slusky} with $X_n=S_n,
Y_n=\widehat S_n-S_n,
\widetilde X_n=\widetilde S_{n,P}$, we have
\[
\sup_{P\in\mathcal P_0}
\sup_{t\in\mathbb R}
\left|
P(\widehat S_n\le t)
-
P\left(
\sum_{j=1}^{J_n}
\tau_{j,n,P}\chi_j^2(1)
\le t
\right)
\right|
\longrightarrow0.
\]
This proves the result.
\end{proof}

\subsection{Proof of Theorem \ref{Theorem: asymptotic distribution under the null}}
Since we assume $\Omega_{n,P}$ to be a real symmetric matrix, it has the decomposition $\Omega_{n,P} = F_{n,P}^\top \Lambda_{n,P} F_{n,P}$, where $F_{n,P}$ is an orthogonal matrix and $\Lambda_{n,P}$ is a diagonal matrix with entries $\lambda_{P,1},...,\lambda_{P,J_n}$. Let $\tilde{b}_{P,j}(X)$ be the $j-$th element of $F_{n,P}B_n(X)$. Let $\tilde{b}^\lambda_{P,j}(X) = {\lambda^{1/2}_{P,j}} \tilde{b}_{P,j}(X)$. Then $\vert\tilde{b}^\lambda_{P,j}(X)\vert^2 = \vert \lambda_{P,j}\vert \vert \tilde{b}_{P,j}(X) \vert^2 = \vert \lambda_{P,j}\vert F^{(j,)}_{n,P}  B_n(X)B^\top _n(X) (F^{(j,)}_{n,P})^{T}$, where $F^{(j,)}_{n,P}$ is the $j-$th row of $F_{n,P}$. Therefore, $\mathbbm{E}_{P}\vert\tilde{b}^\lambda_{P,j}(X)\vert^2 =  \vert \lambda_{P,j}\vert F^{(j,)}_{n,P}  \mathbbm{E}_{P}[B_n(X)B^\top _n(X)] (F^{(j,)}_{n,P})^{T}$ is bounded by some constants since $\mathbbm{E}_{P}[B_n(X)B^\top _n(X)]\preceq C I_{J_n}$ and $ F^{(j,)}_{n,P} (F^{(j,)}_{n,P})^\top  = 1$. Note that we have  
\begin{align*}
    \gamma_{n,P} \leq  \sum_{j,j'}^{J_n}\vert \mathbbm{E}_{P}[g^2(O_i;\eta_{P})\tilde{b}_{P,j}^\lambda(X_i)\tilde{b}_{P,j'}^\lambda(X_i)]\vert = \Vert \text{vec}(\Sigma_{n,P}) \Vert_1\leq {J_n}\gamma_{n,P}.
\end{align*}
Let $\tilde{B}^{\lambda}_n(X) = (\tilde{b}_{1}^\lambda(X),...,\tilde{b}_{J_n}^\lambda(X))$.

Let $Q_{ik} = {2}^{1/2}(\gamma_{n,P} n)^{-1} \sum_{j=1}^{J_n} g(O_i;\eta_{P})g(O_k;\eta_{P})\tilde{b}^{\lambda}_{P,j}(X_i)\tilde{b}^{\lambda}_{P,j}(X_k) $, which can be further written as $ {2}^{1/2}(\gamma_{n,P}n)^{-1}g(O_i;\eta_P)g(O_k;\eta_P)[\Tilde{B}^\lambda_{P}(X_i)]^\top \Tilde{B}^\lambda_{P}(X_k)$. For notational convenience, we suppress the dependence of $Q_{ik}$ on $P$. Let $M_{ni} =\sum_{k=1}^{i-1} Q_{ik} $ for $i = 2,...,n$,  $M_{ni} = 0 $ for $i = 1$ or $i >n$. Let $\mathcal{F}_{ni}$ be the $\sigma$-algebra generated by $O_1,O_2,...,O_i$.
\begin{lemma}
  Under the assumptions of Theorem \ref{Theorem: asymptotic distribution under the null} and null hypothesis, $\{(M_{ni},\mathcal{F}_{ni}), i\geq 1\}$ is a martingale difference sequence for any $n\geq 1$.
\end{lemma}

\begin{proof}
We firstly show that $\mathbbm{E}_{P}[\vert M_{ni} \vert]<\infty$, direct calculation yields
    \begin{align*}
        & \mathbb{E}_{P}[\vert M_{ni} \vert]\\
        = & 2^{1/2}(\gamma_{n,P}n)^{-1}\mathbbm{E}_P\left\vert g(O_i;\eta_P)[\Tilde{B}^\lambda_{P}(X_i)]^\top \sum_{k=1}^{i-1}g(O_k;\eta_P)\Tilde{B}^\lambda_{n,P}(X_k)\right\vert \\
        \leq & 2^{1/2}(\gamma_{n,P}n)^{-1} [\mathbbm{E}_P(g^2(O_i;\eta_P)\Vert \Tilde{B}^\lambda_{n,P}(X_i)\Vert_2^2)]^{1/2}\left(\mathbbm{E}_P\left\Vert\sum_{k=1}^{i-1}g(O_k;\eta_P)\Tilde{B}^\lambda_{n,P}(X_k)\right\Vert_2^2\right)^{1/2} \\
        = & 2^{1/2}(\gamma_{n,P}n)^{-1} [\mathbbm{E}_P(g^2(O_i;\eta_P)\Vert \Tilde{B}^\lambda_{n,P}(X_i)\Vert_2^2)]^{1/2}\left(\sum_{k=1}^{i-1}\mathbbm{E}_Pg^2(O_k;\eta_P)\left\Vert\Tilde{B}^\lambda_{n,P}(X_k)\right\Vert_2^2\right)^{1/2} \\
       =  &  2^{1/2}(\gamma_{n,P}n)^{-1}{(i-1)^{1/2}}\mathbbm{E}_P(g^2(O_i;\eta_P)\Vert \Tilde{B}^\lambda_{n,P}(X_i)\Vert_2^2)\\
       \lesssim &  2^{1/2}(\gamma_{n,P}n)^{-1}{(i-1)^{1/2}}J_n\\
       <&\infty
    \end{align*}
    where the first inequality is due to Cauchy–Schwarz inequality, the second and third equality is due to independence between different observations and conditional zero mean of $g(O_i;\eta_P)$, the last inequality is due to ${(i-1)^{1/2}}<{n}^{1/2}$, $J_n/\gamma_{n,P}\leq J_n $.
    
    Next we aim to show $\mathbbm{E}_{P}[M_{ni}|\mathcal{F}_{ni-1}] = 0$. Direct calculation gives
    \begin{align*}
        &\mathbbm{E}_{P}[M_{ni}|\mathcal{F}_{ni-1}]\\
        =&2^{1/2}(\gamma_{n,P} n)^{-1} \mathbbm{E}_{P}\left[\sum_{j=1}^{J_n}\sum_{k=1}^{i-1} {g}(O_i;\eta_{P}){g}(O_k;\eta_{P})\tilde{b}^{\lambda}_{P,j}(X_i)\tilde{b}^{\lambda}_{P,j}(X_k)|\mathcal{F}_{ni-1}\right]\\
        =&2^{1/2}(\gamma_{n,P} n)^{-1} \sum_{j=1}^{J_n}\sum_{k=1}^{i-1} \mathbbm{E}_{P}\left[g(O_i;\eta_{P})\tilde{b}^{\lambda}_{P,j}(X_i)|\mathcal{F}_{ni-1}\right]g(O_k;\eta_{P})\tilde{b}^{\lambda}_{P,j}(X_k)\\
        =&2^{1/2}(\gamma_{n,P} n)^{-1} \sum_{j=1}^{J_n}\sum_{k=1}^{i-1} \mathbbm{E}_{P}\left[g(O_i;\eta_{P})\tilde{b}^{\lambda}_{P,j}(X_i)\right]g(O_k;\eta_{P})\tilde{b}^{\lambda}_{P,j}(X_k)\\
        =&0,
    \end{align*}
    where the last inequality is due to $\mathbbm{E}_{P}[g(O_i;\eta_{P})|X] = 0$ under the null hypothesis.
\end{proof}

\begin{lemma}
Under the assumptions of Theorem \ref{Theorem: asymptotic distribution under the null} and null hypothesis,
    \begin{align*}
        \sum_{i=1}^\infty \mathbbm{E}_{P}\vert M_{ni} \vert^2 =1-\frac{1}{n}\leq  1 
    \end{align*}
\end{lemma}

\begin{proof}
For $k \neq l$,
    \begin{align*}
        &\mathbbm{E}_{P}[g(O_i;\eta_{P})^2g(O_k;\eta_{P})g(O_l;\eta_{P})\tilde{b}^{\lambda}_{P,j}(X_i)\tilde{b}^{\lambda}_{j'}(X_i)\tilde{b}^{\lambda}_{P,j}(X_k)\tilde{b}^{\lambda}_{P,j'}(X_l)] \\
        =&\mathbbm{E}_{P}[g(O_i;\eta_{P})^2\tilde{b}^{\lambda}_{P,j}(X_i)\tilde{b}^{\lambda}_{P,j'}(X_i)]\mathbbm{E}_{P}[g(O_k;\eta_{P})\tilde{b}^{\lambda}_{P,j}(X_k)]\mathbbm{E}_{P}[g(O_l;\eta_{P})\tilde{b}^{\lambda}_{P,j'}(X_l)] = 0.
    \end{align*}
    Therefore, $\mathbbm{E}_{P}[Q_{ik}Q_{il}] = 0$ when $k\neq l$. 
    \begin{align*}
        &\mathbbm{E}_{P}\vert M_{ni}\vert^2 =\mathbbm{E}_{P}\left\vert \sum_{k=1}^{i-1}Q_{ik}\right\vert^2 = \sum_{k=1}^{i-1}\mathbbm{E}_{P}\vert Q_{ik}\vert^2 = (i-1)\mathbbm{E}_{P}\vert Q_{21}\vert^2\\
        =&\frac{2(i-1)}{\gamma_{n,P}^2n^2} \mathbb{E}_{P}\left\vert\sum_{j=1}^{J_n}  g(O_2;\eta_{P})g(O_1;\eta_{P})\tilde{b}^{\lambda}_{P,j}(X_2)\tilde{b}^{\lambda}_{P,j}(X_1) \right\vert^2 \\
        =&\frac{2(i-1)}{\gamma_{n,P}^2n^2} \sum_{j=1}^{J_n}\sum_{j'=1}^{J_n}  (\mathbbm{E}_{P} [g^2(O;\eta_{P})\tilde{b}^{\lambda}_{P,j}(X)\tilde{b}^{\lambda}_{P,j'}(X)])^2 = \frac{2(i-1)}{n^2}.
    \end{align*}
    Where the last line is due to the definition of $\gamma_{n,P}$. Therefore,
    \begin{align*}
        \sum_{i=1}^\infty \mathbbm{E}_{P}\vert M_{ni}\vert^2 = 1-\frac{1}{n}\leq 1
    \end{align*}
\end{proof}

This Lemma also implies that the martingale we constructed is square-integrable.

\begin{lemma}
\label{lem:martingale_quadratic_variation}
Under the assumptions of Theorem \ref{Theorem: asymptotic distribution under the null}, uniformly over
$P\in\mathcal P_0$, we have $\sum_{i=1}^n \mathbbm{E}_P[M_{ni}^4]=o(1)$ and $\mathbbm{E}_P\left[
\left(
\sum_{i=1}^n M_{ni}^2-1
\right)^2
\right]=o(1).$
\end{lemma}

\begin{proof}
For notational convenience, define $V_i
=
g(O_i;\eta_P)\widetilde B^\lambda_{n,P}(X_i)$.
 Recall that $Q_{ik}
=
\frac{\sqrt{2}}{\gamma n}V_i^\top V_k$, $M_{ni}
=
\sum_{k=1}^{i-1}Q_{ik}$. 
Let $R_m=\sum_{k=1}^m V_k, R_0=0.$  Then $M_{ni}
=
\frac{\sqrt{2}}{\gamma_{n,P} n}V_i^\top R_{i-1}$. Recall $\Omega^{1/2}_{n,P}
=
\Lambda_{n,P}^{1/2}\Gamma_{n,P}$, we have $\Sigma_{n,P}
=
\Omega_{n,P}G_{n,P}\Omega_{n,P}^\top$, $\mathbbm{E}_P\Vert V_i\Vert^2 = \text{tr}(\Sigma_{n,P})\lesssim J_n$. We also have $\mathbbm{E}_P\|V_i\|^4
\lesssim
\omega_n^2J_n$. Indeed, since $\|\Omega_{n,P}\|\lesssim1$,
\begin{align*}
\mathbbm{E}_P\|V_i\|^4
&=
\mathbbm{E}_P\left[
g^4(O_i;\eta_P)
\|\Omega_{n,P}B_n(X_i)\|^4
\right]
\\
&\lesssim
\mathbbm{E}_P\|B_n(X_i)\|^4
\\
&\le
\omega_n^2
\mathbbm{E}_P\|B_n(X_i)\|^2
\lesssim
\omega_n^2J_n,
\end{align*}
where we used Assumptions \ref{Assumption: basis} and \ref{Assumption: nuisance}. Furthermore, for every deterministic $u\in\mathbb R^{J_n}$, we have $\mathbbm{E}_P[(u^\top V_i)^4]
\lesssim
\omega_n^2\|u\|^4$. To see this, let $a=\Omega_{n,P}^\top u$. Then
\begin{align*}
\mathbbm{E}_P[(u^\top V_i)^4]
&=
\mathbbm{E}_P\left[
g^4(O_i;\eta_P)
\{a^\top B_n(X_i)\}^4
\right]
\\
&\lesssim
\mathbbm{E}_P\{a^\top B_n(X_i)\}^4
\\
&\le
\sup_x\{a^\top B_n(x)\}^2
\mathbbm{E}_P\{a^\top B_n(X_i)\}^2
\\
&\lesssim
\omega_n^2\|a\|^4
\lesssim
\omega_n^2\|u\|^4.
\end{align*}

We first establish the fourth-moment condition. Conditional on
$\mathcal F_{n,i-1}$, $R_{i-1}$ is fixed and $V_i$ is independent
of $\mathcal F_{n,i-1}$. Therefore, 
\begin{align*}
\mathbbm{E}_P[M_{ni}^4\mid\mathcal F_{n,i-1}]
&=
\frac{4}{\gamma^4n^4}
\mathbbm{E}_P[
(V_i^\top R_{i-1})^4
\mid
\mathcal F_{n,i-1}]
\\
&\lesssim
\frac{\omega_n^2}{\gamma^4n^4}
\|R_{i-1}\|^4.
\end{align*}
Since the $V_i$'s are independent and mean zero, by Theorem 1 of 
\citet{hansen2015integrated},
\[
\mathbbm{E}_P\|R_m\|^4 \lesssim  \left[\left(\sum_{k = 1}^m\mathbbm{E}_P\Vert V_k\Vert^2\right)^2+\sum_{k=1}^m \mathbbm{E}_P\Vert V_k\Vert^4\right]
\lesssim
m\omega_n^2J_n
+
m^2J_n^2.
\]
Hence
\begin{align*}
\sum_{i=1}^n \mathbbm{E}_P[M_{ni}^4]
&\lesssim
\frac{\omega_n^2}{\gamma^4n^4}
\sum_{m=0}^{n-1}
\left(
m\omega_n^2J_n+m^2J_n^2
\right)
\\
&\lesssim
\frac{\omega_n^4}{J_nn^2}
+
\frac{\omega_n^2}{n} = o(1).
\end{align*}
Thus
\[
\sup_{P\in\mathcal P_0}
\sum_{i=1}^n\mathbbm{E}_P[M_{ni}^4]
\longrightarrow0.
\]

We next establish convergence of the quadratic variation. Write $\sum_{i=1}^nM_{ni}^2-1
=
I_n+(II_n-1),$ where 
\[
I_n
=
\sum_{i=1}^n
\left\{
M_{ni}^2
-
\mathbbm{E}_P[M_{ni}^2\mid\mathcal F_{n,i-1}]
\right\}
\]
and
\[
II_n
=
\sum_{i=1}^n
\mathbbm{E}_P[M_{ni}^2\mid\mathcal F_{n,i-1}].
\]
The summands defining $I_n$ form a martingale difference
sequence. Therefore,
\begin{align*}
\mathbbm{E}_P[I_n^2]
&=
\sum_{i=1}^n
\mathbbm{E}_P\left[
\left\{
M_{ni}^2
-
\mathbbm{E}_P[M_{ni}^2\mid\mathcal F_{n,i-1}]
\right\}^2
\right]
\\
&\le
\sum_{i=1}^n\mathbbm{E}_P[M_{ni}^4]
=o(1)
\end{align*}
uniformly over $P\in\mathcal P_0$.

Moreover, $\mathbbm{E}_P[M_{ni}^2\mid\mathcal F_{n,i-1}]
=
\frac{2}{\gamma^2n^2}
R_{i-1}^\top\Sigma R_{i-1}.$ Therefore, $II_n
=
\frac{2}{\gamma^2n^2}
\sum_{m=1}^{n-1}
R_m^\top\Sigma_{n,P} R_m.$
Since $\mathbbm{E}_P[R_m^\top\Sigma_{n,P} R_m]
=
m\operatorname{tr}(\Sigma_{n,P}^2)
=
m\gamma_{n,P}^2$, we obtain
\[
\mathbbm{E}_P[II_n]
=
\frac{2}{n^2}
\sum_{m=1}^{n-1}m
=
1-\frac1n.
\]

It remains to control the variance of $II_n$. Let $T_n
=
\sum_{m=1}^{n-1}R_m^\top\Sigma_{n,P} R_m.$
Expanding $R_m$ gives $T_n
=
D_n+U_n,$ where
\[
D_n
=
\sum_{k=1}^{n-1}
(n-k)V_k^\top\Sigma_{n,P} V_k
\]
and
\[
U_n
=
2
\sum_{1\le k<l\le n-1}
(n-l)V_k^\top\Sigma_{n,P} V_l.
\]

For the diagonal term $D_n$, independence yields
\begin{align*}
\operatorname{Var}_P(D_n)
&=
\sum_{k=1}^{n-1}
(n-k)^2
\operatorname{Var}_P(V_k^\top\Sigma V_k)
\\
&\le
\sum_{k=1}^{n-1}
(n-k)^2
\mathbbm{E}_P[(V_k^\top\Sigma V_k)^2]
\\
&\lesssim
n^3\omega_n^2J_n,
\end{align*}
where we used $\|\Sigma\|\lesssim1$ and $\mathbbm{E}_P\|V_i\|^4\leq \omega_n^2J_n$. For $k<l$, let $H_{kl}=V_k^\top\Sigma V_l.$ If $(k,l)\neq(k',l')$ are distinct pairs, then
$\mathbbm{E}_P[H_{kl}H_{k'l'}]=0$ by independence and $\mathbbm{E}_P[V_i]=0$.
Furthermore,
\begin{align*}
\mathbbm{E}_P[H_{kl}^2] = 
\mathbbm{E}_P[(V_k^\top\Sigma V_l)^2] = 
\operatorname{tr}(\Sigma^4)\leq
\|\Sigma\|^2\operatorname{tr}(\Sigma^2)
\lesssim
J_n.
\end{align*}
Consequently,
\begin{align*}
\operatorname{Var}_P(U_n)
&=
4
\sum_{l=2}^{n-1}
\sum_{k=1}^{l-1}
(n-l)^2\mathbbm{E}_P[H_{kl}^2]\lesssim
J_n
\sum_{l=2}^{n-1}
(l-1)(n-l)^2
\\
&\lesssim
n^4J_n.
\end{align*}

Using
$\operatorname{Var}(D_n+U_n)
\le
2\operatorname{Var}(D_n)+2\operatorname{Var}(U_n)$, we obtain
\begin{align*}
\operatorname{Var}_P(P_n)
&=
\frac{4}{\gamma_{n,P}^4n^4}
\operatorname{Var}_P(T_n)
\\
&\lesssim
\frac{\omega_n^2}{J_nn}
+
\frac{1}{J_n}
=o(1)
\tag{S.12.10}
\end{align*}
uniformly over $P\in\mathcal P_0$. Here we used
$\omega_n\lesssim J_n$, $J_n=o(n^{1/3})$, and $J_n\to\infty$. Therefore the claimed results have been proved.
\end{proof}

\begin{lemma}
Under the assumptions of Theorem \ref{Theorem: asymptotic distribution under the null},
    $$ \frac{n\Vert\frac{1}{n}\sum_{i=1}^n { g(O_i;\eta_{P})\Omega_{n,P}^{1/2}B_{n}(X_i)}\Vert^2-\rho_{n,P}}{2^{1/2}\gamma_{n,P}}\rightsquigarrow N(0,1)$$ 
    uniformly over $P\in \mathcal{P}_0$.
    \label{lemma: normal limit}
\end{lemma}

\begin{proof}
    \begin{align*}
        &\frac{1}{2^{1/2}\gamma_{n,P}}\left(n\left\Vert\frac{1}{n}\sum_{i=1}^n { g(O_i;\eta_{P})\Omega_{n,P}^{1/2}B_{n,P}(X_i)}\right\Vert^2-\rho_n\right)\\
        =& \frac{1}{2^{1/2}\gamma_{n,P}}\left(n\left\Vert\frac{1}{n}\sum_{i=1}^n { g(O_i;\eta_{P})\tilde{B}^\lambda_{n,P}(X_i)}\right\Vert^2-\rho_{n,P}\right)\\
        =&\frac{1}{2^{1/2}\gamma_{n,P}}\left(\frac{1}{n}\sum_{i=1}^n g^2(O_i;\eta_{P})(\tilde{B}^\lambda_{n,P}(X_i))^\top \tilde{B}^\lambda_{n,P}(X_i)-\rho_{n,P}\right)\\
        &+\frac{1}{2^{1/2}\gamma_{n,P}}\left(\frac{1}{n}\sum_{i\neq i'}^n g(O_i;\eta_{P})g(O_{i'};\eta_{P})(\tilde{B}^\lambda_{n,P}(X_i))^\top \tilde{B}^\lambda_{n,P}(X_{i'})\right).
    \end{align*}
The first term are summations of iid mean-zero terms, therefore we are able to use Markov inequality:
\begin{align*}
    & P\left(\frac{1}{2^{1/2}\gamma_{n,P}}\left\vert\frac{1}{n}\sum_{i=1}^n g^2(O_i;\eta_{P})(\tilde{B}^\lambda_{n,P}(X_i))^\top \tilde{B}^\lambda_{n,P}(X_i)-\rho_{n,P}\right\vert>\epsilon\right)\\
    \leq & \frac{1}{2\gamma^2_{n,P}n\epsilon^2}\text{Var}_P\left(g^2(O;\eta_P)\Vert \tilde{B}^\lambda_{n,P}(X_i)\Vert^2\right) \\
    \leq & \frac{1}{2\gamma^2_{n,P}n\epsilon^2}\mathbbm{E}_P\left(g^4(O;\eta_P)\Vert \tilde{B}^\lambda_{n,P}(X_i)\Vert^4\right) \\
    \lesssim & \frac{J_n^4}{2\gamma^2_{n,P}n\epsilon^2} = o(1).
\end{align*}

 Therefore, 
    \begin{align*}
        \frac{1}{2\gamma^2_{n,P}n\epsilon^2}\text{Var}_P\left(g^2(O;\eta_P)\Vert \tilde{B}^\lambda_{n,P}(X_i)\Vert^2\right) = o_{\mathcal{P}_0}(1).
    \end{align*}

    Furthermore, by Lemma \ref{lemma: martingale bound} and previous lemmas in this subsection, we have that the following weak convergence result holds uniformly over $P\in \mathcal{P}_0$:
    \begin{align*}
        \frac{1}{2^{1/2}\gamma_{n,P}}\left(\frac{1}{n}\sum_{i\neq i'}^n g(O_i;\eta_{P})g(O_{i'};\eta_{P})(\tilde{B}^\lambda_{n,P}(X_i))^\top \tilde{B}^\lambda_{n,P}(X_{i'})\right) \rightsquigarrow N(0,1).
    \end{align*}
    Therefore we have proved the claimed result.
\end{proof}

\begin{lemma}
Under assumptions of Theorem \ref{Theorem: asymptotic distribution under the null},
    \begin{align*}
        \frac{n\Vert\frac{1}{n}\sum_{i=1}^n { g(O_i;\eta_{P})\Omega_{n,P}^{1/2}B_{n,P}(X_i)}\Vert^2-\rho_{n,P}}{2^{1/2}\gamma_{n,P}} - \frac{n\Vert\frac{1}{n}\sum_{i=1}^n { g(O_i;\widehat{\eta})\widehat{\Omega}_{n}^{1/2}B_n(X_i)}\Vert^2-\widehat{\rho}_{n,P}}{2^{1/2}\widehat{\gamma}_n} = o_{\mathcal{P}_0}(1)
    \end{align*}
    \label{lemma: no nuisance impact on test statistic}
\end{lemma}

\begin{proof}
It's a direct result of Lemma \ref{lemma: convergence rate of S_n}, Lemma \ref{lem:centering-scaling} and Lemma \ref{lemma: normal limit}.
\end{proof}

Then Theorem 2 is a direct result of Lemma \ref{lemma: no nuisance impact on test statistic}, Lemma \ref{lemma: normal limit} and Lemma \ref{lemma: uniform slusky}.

\subsection{Proof of Theorem \ref{Theorem: local power analysis}}

\begin{lemma}
\label{lem: lemma for local alternative}
Suppose the assumptions of Theorem 3 hold. Then, under
$P_{1,n}$,
\[
\frac{S_{n,P_{1,n}}-\rho_{P_{1,n},n}}
{\sqrt{2}\gamma_{P_{1,n},n}}
-
\frac{1}{\sqrt{2}}
\sum_{j=1}^{J_n}
\left\{
\mathbbm{E}_{P_{1,n}}
\left[
\phi(X)\widetilde b^\lambda_{P_{1,n},j}(X)
\right]
\right\}^2
\rightsquigarrow N(0,1),
\]
where
\[
S_{n,P}
=
n
\left\|
\frac{1}{n}
\sum_{i=1}^n
g(O_i;\eta_P)
\widetilde B^\lambda_{n,P}(X_i)
\right\|^2.
\]
\end{lemma}

\begin{proof}
For notational convenience, throughout the proof write
$P=P_{1,n}$. Define
$\varepsilon_i
=
g_i-(\gamma_{n,P}/n)^{1/2}\phi(X_i),
Z_i=\varepsilon_i\widetilde B^\lambda(X_i),
H_i=\phi(X_i)\widetilde B^\lambda(X_i),$
and $\mu_n=\mathbbm{E}_P[H_i].$
By the local alternative, $\mathbbm{E}_P[\varepsilon_i\mid X_i]=0,$
and therefore $\mathbbm{E}_P[Z_i]=0$.

Let $\overline Z_n=\frac1n\sum_{i=1}^nZ_i,
\overline H_n=\frac1n\sum_{i=1}^nH_i.$ Since
$g_i\widetilde B^\lambda(X_i)=Z_i+(\gamma_{n,P}/n)^{1/2}H_i$,
\begin{align}
S_{n,P}
&=
n\|\overline Z_n+(\gamma_{n,P}/n)^{1/2}\overline H_n\|^2
\nonumber\\
&=
n\|\overline Z_n\|^2
+
2n(\gamma_{n,P}/n)^{1/2}\overline Z_n^\top\overline H_n
+
n(\gamma_{n,P}/n)\|\overline H_n\|^2.
\end{align}

Therefore,
\begin{align*}
    \frac{S_{n,P}-\rho_{n,P}}{\sqrt{2}\gamma_{n,P}} - \frac{1}{\sqrt{2}}\Vert \overline{H}\Vert^2
&=
\frac{n\|\overline Z_n\Vert^2-\rho_{n,P}}{\sqrt{2}\gamma_{n,P}}+\frac{{(2n)}^{1/2}}{\gamma_{n,P}^{1/2}}\overline{Z}^\top_n\overline{H}_n.
\end{align*}

We aim to establish the following results:
\begin{align*}
    &\frac{1}{\sqrt{2}}\Vert \overline{H}\Vert^2 - \frac{1}{\sqrt{2}}
\sum_{j=1}^{J_n}
\left\{
\mathbbm{E}_{P_{1,n}}
\left[
\phi(X)\widetilde b^\lambda_{P_{1,n},j}(X)
\right]
\right\}^2 = o_p(1),\\
&\frac{n\|\overline Z_n\Vert^2-\rho_{n,P}}{\sqrt{2}\gamma_{n,P}} \rightsquigarrow N(0,1),\\
&\frac{{2n}^{1/2}}{\gamma_{n,P}^{1/2}}\overline{Z}^\top_n\overline{H}_n = o_p(1).
\end{align*}

Those three results together imply the claimed result. We first prove
\begin{equation}
\frac{1}{\sqrt{2}}
\|\overline H_n\|^2
-
\frac{1}{\sqrt{2}}
\sum_{j=1}^{J_n}
\left\{
\mathbbm{E}_{P}
\left[
\phi(X)\widetilde b^\lambda_{P,j}(X)
\right]
\right\}^2
=
o_P(1).
\end{equation}
Recall that $\mu_n
=
\mathbbm{E}_{P}[H_i]
=
\mathbbm{E}_{P}
\left[
\phi(X)\widetilde B^\lambda_{n,P}(X)
\right]$,  Therefore, $\|\mu_n\|^2
=
\sum_{j=1}^{J_n}
\left\{
\mathbbm{E}_{P}
\left[
\phi(X)\widetilde b^\lambda_{P,j}(X)
\right]
\right\}^2.$ It thus suffices to show $\|\overline H_n\|^2-\|\mu_n\|^2=o_P(1).$ By independence,
\begin{align*}
\mathbbm{E}_{P}
\|\overline H_n-\mu_n\|^2
&=
\frac{1}{n}
\mathbbm{E}_{P}\|H_i-\mu_n\|^2\\
&\le
\frac{1}{n}
\mathbbm{E}_{P}\|H_i\|^2\\
&=
\frac{1}{n}
\mathbbm{E}_{P}
\left[
\phi^2(X)
\|\widetilde B^\lambda_{n,P}(X)\|^2
\right]\\
&\le\frac{1}{n}
\left\{
\mathbbm{E}_{P}[\phi^4(X)]
\right\}^{1/2}
\left\{
\mathbbm{E}_{P}
\|\widetilde B^\lambda_{n,P}(X)\|^4
\right\}^{1/2}\\
&\lesssim \frac{\omega_n J_n^{1/2}}{n}
\end{align*}

Since $\omega_n\lesssim J_n$ and $J_n=o(n^{1/3})$, we have
\[
\mathbbm{E}_{P}
\|\overline H_n-\mu_n\|^2
\lesssim
\frac{\omega_nJ_n^{1/2}}{n}
=o(1).
\]

Moreover, by Lemma \ref{lemma: bounded projection} applied with $f=\phi$, we have
\begin{align*}
\left|
\|\overline H_n\|^2-\|\mu_n\|^2
\right|
&=
\left|
2\mu_n^\top(\overline H_n-\mu_n)
+
\|\overline H_n-\mu_n\|^2
\right|\\
&\le
2\|\mu_n\|
\|\overline H_n-\mu_n\|
+
\|\overline H_n-\mu_n\|^2\\
&=
o_P(1).
\end{align*}

Now we show the second term. The moment bounds used in Lemma \ref{lem:martingale_quadratic_variation} give $\mathbbm{E}_P\|V_i\|^4\lesssim \omega_n^2J_n$ and, for every deterministic $u\in\mathbb R^{J_n}$, $\mathbbm{E}_P[(u^\top V_i)^4]
\lesssim
\omega_n^2\|u\|^4.$ Since $Z_i=V_i-(\gamma_{n,P}/n)^{1/2}\phi(X_i)\tilde{B}^\lambda(X_i)$, with similar arguments we used in Lemma \ref{lem:martingale_quadratic_variation}, the same bounds hold for $Z_i$: $\mathbbm{E}_P\|Z_i\|^4\lesssim \omega_n^2J_n,
\mathbbm{E}_P[(u^\top Z_i)^4]
\lesssim
\omega_n^2\|u\|^4.$
Moreover, under the assumptions of Theorem \ref{Theorem: asymptotic distribution under the null} the eigenvalues of
$\Sigma_{n,P}$ are bounded above and away from zero. Define $\Sigma^Z_{n,P} = \mathbbm{E}_P[Z_iZ_i^\top]$, $m_{n,P} = \mathbbm{E}_P[V_i] = (\gamma_{n,P}/n)^{1/2}\mu_n$. Since $\|\Sigma^Z_{n,P}-\Sigma_{n,P}\|
=
\|\gamma_{n,P}/n\mathbbm{E}_P[\phi^2(X)\tilde{B}^\lambda_{n,P}(X)\tilde{B}^\lambda_{n,P}(X)^\top]\|=o(1/{J^{1/2}_n}),$
the eigenvalues of $\Sigma^Z_{n,P}$ are also bounded above and away
from zero for all sufficiently large $n$. Hence the martingale argument applies with
$Z_i$ in place of
$g(O_i;\eta_P)\widetilde B^\lambda_{n,P}(X_i)$, and yields
\[
\frac{
n\left\|n^{-1}\sum_{i=1}^n Z_i\right\|^2
-\rho^Z_{n,P}
}{
\sqrt{2}\gamma^Z_{n,P}
}
\rightsquigarrow N(0,1),
\]
where $\gamma^Z_{n,P} = \Vert \Sigma^Z_{n,P} \Vert_F$, $\rho^Z_{n,P} = \text{tr}(\Sigma^Z_{n,P})$. We can similarly show that $\gamma^Z_{n,P}/\gamma_{n,P}\to1$, $\rho^Z_{n,P}/\sqrt{2}\gamma^Z_{n,P} - \rho_{n,P}/\sqrt{2}\gamma_{n,P} \rightarrow 0$. We omit the details here. Therefore Slutsky's theorem gives
\[
\frac{
n\left\|n^{-1}\sum_{i=1}^n Z_i\right\|^2
-\rho_{n,P}
}{
\sqrt{2}\gamma_{n,P}
}
\rightsquigarrow N(0,1).
\]

Therefore, the second claim has been proved. Now we focus on the third claim and show the cross-term converges to $0$. Write $\frac{{(2n)}^{1/2}}{\gamma_{n,P}^{1/2}}\overline{Z}^\top_n\overline{H}_n = \frac{{(2n)}^{1/2}}{\gamma_{n,P}^{1/2}}\overline{Z}^\top_n(\overline{H}_n-\mu_n)+\frac{{(2n)}^{1/2}}{\gamma_{n,P}^{1/2}}\overline{Z}^\top_n\mu_n$.  Here we focus on the second term since the first term has smaller order (since we have already derived $\Vert \widebar{Z}_n\Vert = O_p((J_n/n)^{1/2})$, $\Vert \widebar{H}_n - \mu_n\Vert = O_p(\omega_n^{1/2}J_n^{1/4}/n^{1/2})$). For the second term, we have
\begin{align*}
\mathbbm{E}_P\left[
2n (m_{n,P})^\top
\left(\frac1n\sum_{i=1}^n Z_i\right)
\right]^2
&=
4n m_{n,P}^\top
\Sigma^Z_{n,P}m_{n,P}\\
&\lesssim
n\|m_{n,P}\|^2\\
&\lesssim
\sqrt{J_n}.
\end{align*}
Since $\gamma_{n,P}^2\asymp J_n$,
\[
\frac{
2n m_{n,P}^\top
\left(n^{-1}\sum_{i=1}^nZ_i\right)
}{
\sqrt{2}\gamma_{n,P}
}
=o_P(1).
\]

Therefore, all three claims have been proved and therefore the claimed results hold.
\end{proof}

\begin{proof}[Proof of Theorem \ref{Theorem: local power analysis}]
    We only need to show the impact of nuisance estimation is negligible. Since Lemma \ref{lem:centering-scaling} is still valid under local alternative, we only need to show $\widehat{S}_n - S_n = o_P(\gamma_n)$. To show this, we only need to show $\Vert a_n\Vert = O_p((J_n/n)^{1/2})$, since the proof of Lemma \ref{lemma: second order impact of nuisance} is still valid under the alternative. 

    Under the alternative, $a_n$ is no longer mean zero, write $\Vert a_n\Vert \leq  \Vert a_n - \mathbbm{E}_P a_n \Vert +\Vert \mathbbm{E}_P a_n\Vert$. $\Vert a_n - \mathbbm{E}_P a_n \Vert$ is still $O_p((J_n/n)^{1/2})$ following the similar proof as Lemma \ref{Lemma: a_n}.

    Under the local alternative, we have $\Vert \mathbbm{E}_P[a_n]\Vert = \Vert (\gamma_n/n)^{1/2}\mathbbm{E}[B_n(X)\phi(X)]\Vert \lesssim (\gamma_n/n)^{1/2}\leq (J_n/n)^{1/2}$, therefore $\Vert a_n\Vert = O_p((J_n/n)^{1/2})$ still holds.
\end{proof}

\subsection{Proof of Theorem \ref{Theorem: consistency}}
\begin{proof} [Proof of Theorem \ref{Theorem: consistency}]

To simplify notation, we let 
\begin{align*}
    d_{n,P} = \sum_{j=1}^{J_n}\left[\mathbbm{E}_P[g(O;\eta_P)\tilde{b}^\lambda_{P,j}(X)]^2\right].
\end{align*}
The set of alternative distributions can be rewritten as
\begin{align*}
    \mathcal{P}_{1,n} = \left\{ P:d_{n,P}>\frac{q\gamma_{n,P}}{n}\right\}.
\end{align*}
Similar to the previous proofs, we still define $V_i = g_i\tilde{B}^\lambda_{n,P}(X_i)$, $m_{n,P} = \mathbbm{E}_P[V_i]$. Then $d_{n,P} = \Vert m_{n,P}\Vert^2$. Now we consider $P\in \mathcal{P}_{1,n}(q)$, which implies $nd_{n,P}>q\gamma_{n,P}$ by the definition of $\mathcal{P}_{1,n}$. Furthermore, the definition of $\Sigma_{n,P}$, $\gamma_{n,P}$ and $\rho_{n,P}$ are still defined as follows:
\begin{align*}
    &\Sigma_{n,P} = \mathbbm{E}_P[V_i V_i^\top],\gamma_{n,P} = \Vert \Sigma_{n,P}\Vert_{F},\rho_{n,P} = \text{tr}(\Sigma_{n,P} ).
\end{align*}

Recall that the oracle uncentered statistic is defined as
\begin{align*}
    S_{n,P} = \left\Vert \frac{1}{\sqrt{n}}\sum_{i=1}^n V_i \right\Vert^2.
\end{align*}

The basic idea is to study the power of the oracle standardized test statistic and then use the fact that the impact of nuisance estimation is small as we derived before. Therefore we focus on studying the oracle power and ignore the nuisance estimation. We  start with the following decomposition:
\begin{align*}
    S_{n,P}-\rho_{n,P} = &\left\Vert \frac{1}{\sqrt{n}}\sum_{i=1}^n (V_i - \mathbbm{E}_P[V_i]) \right\Vert^2 -\text{tr}(\text{Var}(V_i)) \\
    &+ 2\sqrt{n}m_{n,P}^\top \frac{1}{\sqrt{n}}\sum_{i=1}^n (V_i - \mathbbm{E}_P[V_i]) \\
    &+(n-1)d_{n,P}.
\end{align*}

From previous proof we have already known that $\gamma_{n,P} \asymp J_n^{1/2}$. For $d_{n,P}$, we have
\begin{align*}
    d_{n,P} =  \Vert \mathbbm{E}_P[\mathbbm{E}_P[g_i|X]\tilde{B}^\lambda_{n,P}(X)]\Vert^2\lesssim C.
\end{align*}

For $\text{Var}(V_i)$, we have the following inequalities:
\begin{align*}
    &\Vert \text{Var}_P(V_i)\Vert=\Vert \mathbbm{E}_P[(V_i - \mathbbm{E}_P[V_i])(V_i - \mathbbm{E}_P[V_i])^\top] \Vert\leq \Vert \Sigma_{n,P}\Vert + \Vert m_{n,P}m_{n,P}^\top\Vert = O(1),\\
    &\Vert \text{Var}_P(V_i)\Vert_F \leq \Vert \Sigma_{n,P}\Vert_F + \Vert m_{n,P}m_{n,P}^\top\Vert_F = \gamma_{n,P} + d_{n,P}\lesssim \gamma_{n,P}.
\end{align*}

Now we consider controlling the first term in the decomposition. We denote this term as $I_n$:
\begin{align*}
    I_n : = \left\Vert \frac{1}{\sqrt{n}}\sum_{i=1}^n (V_i - \mathbbm{E}_P[V_i]) \right\Vert^2 -\text{tr}(\text{Var}(V_i)).
\end{align*}

First, 
\begin{align*}
    &\mathbbm{E}_{P}[I_n] \\
    =&\mathbbm{E}_{P}\left\Vert \frac{1}{\sqrt{n}}\sum_{i=1}^n (V_i - \mathbbm{E}_P[V_i]) \right\Vert^2 -\text{tr}(\text{Var}(V_i))\\
    =&\mathbbm{E}_{P}\frac{1}{n}\sum_{i=1}^n(V_i - \mathbbm{E}_P[V_i])^\top\sum_{j=1}^n(V_j - \mathbbm{E}_P[V_j])-\text{tr}(\text{Var}(V_i)) \\
    =&\mathbbm{E}_{P}(V_i - \mathbbm{E}_P[V_i])^\top(V_i - \mathbbm{E}_P[V_i])-\text{tr}(\text{Var}(V_i))\\
    =&0.
\end{align*}
Therefore, we only need to derive the variance to determine its order. Direct calculation gives
\begin{align*}
    \text{Var}(I_n)& = \frac{1}{n}\text{Var}_P(\Vert V_i -\mathbbm{E}[V_i]\Vert^2) + \frac{2(n-1)}{n}\text{tr}((\text{Var}(V_i))^2)\\
    & \lesssim \frac{1}{n}\mathbbm{E}_P[\Vert V_i \Vert^4] + \frac{1}{n}\Vert m_{n,P}\Vert^4 + 2\Vert \text{Var}(V_i)\Vert_F^2\\
    & \lesssim \frac{\omega_n^2J_n}{n}+\gamma_{n,P}^2.
\end{align*}

Therefore, we have $\text{Var}(I_n)/\gamma_{n,P}^2 \lesssim C$. Consequently, for every $M>0$, $\sup_{P\in \mathcal{P}}P(\vert I_n\vert>M\gamma_{n,P} )\leq C/M^2$. If $P\in \mathcal{P}_{1,n}(q)$, then $n d_{n,P}>q\gamma_{n,P}$, therefore we have
\begin{align*}
    P\left(\vert I_{n}\vert>1/4nd_{n,P}\right)\leq \frac{16\text{Var}_P(I_n)}{n^2d_{n,P}^2}\lesssim \frac{\gamma_{n,P}^2}{n^2d_{n,P}^2}\leq C/q^2.
\end{align*}

Next, we derive the rate for the second term:
\begin{align*}
    II_n&: = 2\sqrt{n}m_{n,P}^\top \frac{1}{\sqrt{n}}\sum_{i=1}^n (V_i - \mathbbm{E}_P[V_i])
\end{align*}
It's easy to see $\mathbbm{E}_P[II_n] = 0$, $\text{Var}_P[II_n] = 4n m_{n,P}^\top \text{Var}_P[V_i]m_{n,P}$. Therefore,
\begin{align*}
    \text{Var}_P[II_n] = 4n m_{n,P}^\top \text{Var}_P[V_i]m_{n,P}\lesssim n \Vert m_{n,P}\Vert^2 = nd_{n,P}.
\end{align*}

Therefore, we have
\begin{align*}
    P\left(\vert II_n\vert >1/4n d_{n,P}\right) \lesssim \frac{1}{nd_{n,P}}.
\end{align*}

When $P\in \mathcal{P}_{1,n}(q)$, $nd_{n,P}>q\gamma_{n,P}$, therefore we have
\begin{align*}
     P\left(\vert II_n\vert >1/4n d_{n,P}\right) \lesssim 1/(q\gamma_{n,P}) = o(1)
\end{align*}

Combine the previous results, we have with probability at least $1-C/q^2-o(1)$, uniformly over $P\in \mathcal{P}_{1,n}(q)$, we have
\begin{align*}
    S_{n,P}-\rho_{n,P} \geq (n-1)d_{n,P} - \frac{1}{4}nd_{n,P}- \frac{1}{4}nd_{n,P} \geq \frac{1}{4}nd_{n,P}
\end{align*}
for sufficiently large $n$. Since $P\in \mathcal{P}_{1,n}(q)$,
\begin{align*}
    \frac{S_{n,P}-\rho_{n,P}}{2^{1/2}\gamma_{n,P}} \geq q/(4\times 2^{1/2}).
\end{align*}
Choosing $q$ sufficiently large such that $q/(4\times 2^{1/2}) >z_{1-\alpha}+1$ and $C/q^2<\epsilon/4$, then we have
\begin{align}
    \liminf_{n\rightarrow\infty}\inf_{P\in\mathcal{P}_{1,n}(q)}P\left(\frac{S_{n,P}-\rho_{n,P}}{2^{1/2}\gamma_{n,P}}>z_{1-\alpha}+1\right)\geq 1-\epsilon/2 \label{eq: alternative, oracle}
\end{align}

Now we handle the impact of nuisance estimation. Similar to the proof in Lemma \ref{lemma: convergence rate of S_n}, we define $a_n = \frac{1}{n}\sum_{i=1}^nB_n(X_i)g(O_i;\eta_{{P}})$ and $\widehat{a}_n = \frac{1}{n}\sum_{i=1}^nB_n(X_i)g(O_i;{\eta}_{\widehat{P}_n})$. We can show that $\Vert \widehat{a}_n - a_n\Vert = O_{\mathcal{P}}(\zeta_n\xi_n(J_n/n)^{1/2})$. Furthermore, $\Vert a_n\Vert = O_P(d^{1/2}_{n,P}+(J_n/n)^{1/2})$, the additional $d^{1/2}_{n,P}$ is because we are working under alternative. With a decomposition  similar to that in \ref{lemma: convergence rate of S_n}, we can show that $\widehat{S}_n - S_{n}\vert = o_\mathcal{P}(\gamma_{n,P}+nd_{n,P})$, under the alternative $\mathcal{P}_{1,n}$, we further have $\gamma_{n,P}+nd_{n,P}\lesssim nd_{n,P}$. Therefore, $\vert \widehat{S}_n - S_{n}\vert = o_{P}(nd_{n,P})$. The results derived in \ref{lem:centering-scaling} do not depend on the null. Combining this with the results under the oracle setting, we have
\begin{align*}
    \liminf_{n\rightarrow\infty}\inf_{P\in\mathcal{P}_{1,n}(q)}P\left(\frac{\widehat{S}_{n}-\widehat{\rho}_{n}}{2^{1/2}\widehat{\gamma}_{n}}>z_{1-\alpha}\right)\geq 1-\epsilon.
\end{align*}

The second claim is obvious since under the alternative with $q_n \rightarrow \infty$,
\begin{align*}
    \frac{S_{n}-\rho_{n,P}}{\gamma_{n,P}} = \frac{(n-1)d_{n,P}}{\gamma_{n,P}} + \frac{I_n}{\gamma_{n,P}} + \frac{II_n}{\gamma_{n,P}}.
\end{align*}

Now the signal $\frac{(n-1)d_{n,P}}{\gamma_{n,P}}  = (1-1/n)(nd_{n,P}/\gamma_{n,P})\geq (1-1/n)q_n \rightarrow \infty$.

\end{proof}

\subsection{Proofs for Proposition \ref{prop: example 1} and \ref{prop: example 2}}
\label{subsec: proofs for examples}
The proof for Proposition \ref{prop: example 1} is given as below.
\begin{proof}
For notational convenience, fix $a\in\{0,1\}$ and, for
$s\in\{0,1\}$, define $I_s
    =
    I(A=a,S=s),
    \mu_{P,s}(x)
    =
    \mathbbm{E}_P[Y\mid A=a,S=s,X=x],
    \pi_{P,s}(x)
    =
    P(A=a,S=s\mid X=x).$
For a generic nuisance value $\eta
=
\left\{
\mu_s,\pi_s:s\in\{0,1\}
\right\}$, define $\psi_s(O;\eta)
=
\frac{I_s}{\pi_s(X)}
\{Y-\mu_s(X)\}
+
\mu_s(X)$, 
so that
\[
g^\perp_{\mathrm{ME}}(O;\eta)
=
\psi_1(O;\eta)-\psi_0(O;\eta).
\]

We verify the conditions in Assumption~3.

\noindent\textbf{Step 1: Moment conditions and uniform boundedness.}

Since $Y$ is uniformly bounded, there exists a constant $C<\infty$
such that $|\mu_{P,s}(X)|
=
\left|
\mathbbm{E}_P[Y\mid A=a,S=s,X]
\right|
\le C, s\in\{0,1\}.$ 
By the assumed bounds on the realization set
$\mathcal{T}_n$, every candidate $\mu_s$ is also uniformly bounded.
Moreover, by positivity, $\epsilon
\le
\pi_s(X)
\le
1-\epsilon$, for every $\eta\in\mathcal{T}_n$. It follows that
\begin{align*}
|\psi_s(O;\eta)|
&\le
\frac{1}{\epsilon}
\left\{
|Y|+|\mu_s(X)|
\right\}
+
|\mu_s(X)|\le C.
\end{align*}
Therefore, $\sup_{\eta\in\mathcal{T}_n}
|g^\perp_{\mathrm{ME}}(O;\eta)|
\le C.$ Consequently,
\[
\sup_{\eta\in\mathcal{T}_n}
\mathbbm{E}_P
\left[
\left\{
g^\perp_{\mathrm{ME}}(O;\eta)
\right\}^4
\mid X
\right]
\le C.
\]
Furthermore, uniform boundedness implies
\[
\left|
g^\perp_{\mathrm{ME}}(O;\eta)
+
g^\perp_{\mathrm{ME}}(O;\eta_P)
\right|^2
\le C.
\]
Hence
\begin{align*}
&
\mathbbm{E}_P
\left[
\left\{
g^\perp_{\mathrm{ME}}(O;\eta)
-
g^\perp_{\mathrm{ME}}(O;\eta_P)
\right\}^2
\left\{
g^\perp_{\mathrm{ME}}(O;\eta)
+
g^\perp_{\mathrm{ME}}(O;\eta_P)
\right\}^2
\right]
\\
&\qquad\le
C
\mathbbm{E}_P
\left[
\left\{
g^\perp_{\mathrm{ME}}(O;\eta)
-
g^\perp_{\mathrm{ME}}(O;\eta_P)
\right\}^2
\right].
\end{align*}

\noindent\textbf{Step 2: Control of $r_{n,P}$.}

For $s\in\{0,1\}$, define $\Delta_{\mu,s}(X)
=
\mu_s(X)-\mu_{P,s}(X),
\Delta_{\pi,s}(X)
=
\pi_s(X)-\pi_{P,s}(X).$
We have
\begin{align*}
&
\psi_s(O;\eta)
-
\psi_s(O;\eta_P)
\\
&=
I_s
\left[
\frac{Y-\mu_s(X)}{\pi_s(X)}
-
\frac{Y-\mu_{P,s}(X)}{\pi_{P,s}(X)}
\right]
+
\Delta_{\mu,s}(X)
\\
&=
I_s
\{Y-\mu_{P,s}(X)\}
\left\{
\frac{1}{\pi_s(X)}
-
\frac{1}{\pi_{P,s}(X)}
\right\}
-
\frac{I_s}{\pi_s(X)}
\Delta_{\mu,s}(X)
+
\Delta_{\mu,s}(X).
\end{align*}
Since $\pi_s$ and $\pi_{P,s}$ are uniformly bounded below by
$\epsilon$,
\[
\left|
\frac{1}{\pi_s(X)}
-
\frac{1}{\pi_{P,s}(X)}
\right|
=
\frac{
|\Delta_{\pi,s}(X)|
}{
\pi_s(X)\pi_{P,s}(X)
}
\le
C|\Delta_{\pi,s}(X)|.
\]
Moreover, $Y-\mu_{P,s}(X)$ is uniformly bounded. Therefore,
\[
\left|
\psi_s(O;\eta)
-
\psi_s(O;\eta_P)
\right|
\le
C
\left\{
|\Delta_{\mu,s}(X)|
+
|\Delta_{\pi,s}(X)|
\right\}.
\]
Taking the $L_2(P)$ norm yields
\[
\left\|
\psi_s(O;\eta)
-
\psi_s(O;\eta_P)
\right\|_{P,2}
\le
C
\left\{
\|\Delta_{\mu,s}\|_{P,2}
+
\|\Delta_{\pi,s}\|_{P,2}
\right\}.
\]
Since $g^\perp_{\mathrm{ME}}
=
\psi_1-\psi_0,$ the triangle inequality gives
\begin{align*}
&
\left\|
g^\perp_{\mathrm{ME}}(O;\eta)
-
g^\perp_{\mathrm{ME}}(O;\eta_P)
\right\|_{P,2} \qquad\le
C
\sum_{s=0}^1
\left\{
\|\mu_s-\mu_{P,s}\|_{P,2}
+
\|\pi_s-\pi_{P,s}\|_{P,2}
\right\}.
\end{align*}
By the assumed $L_2(P)$ convergence rates,
\[
\sup_{\eta\in\mathcal{T}_n}
\left\|
g^\perp_{\mathrm{ME}}(O;\eta)
-
g^\perp_{\mathrm{ME}}(O;\eta_P)
\right\|_{P,2}
\lesssim
\zeta_n.
\]
Thus, after absorbing an absolute multiplicative constant into
$\zeta_n$, $r_{n,P}\le \zeta_n.$

\noindent\textbf{Step 3: Conditional Neyman orthogonality.}

Fix $s\in\{0,1\}$ and an arbitrary
$\eta_1\in\mathcal{T}_n$. Consider the path $\eta_t
=
(1-t)\eta_P+t\eta_1,
t\in[0,1].$ For the $s$th component, write
\begin{align*}
\mu_{s,t}(X)
&=
\mu_{P,s}(X)
+
t\Delta_{\mu,s}(X),\\
\pi_{s,t}(X)
&=
\pi_{P,s}(X)
+
t\Delta_{\pi,s}(X),
\end{align*}
where now $\Delta_{\mu,s}
=
\mu_{1,s}-\mu_{P,s}$, $\Delta_{\pi,s}
=
\pi_{1,s}-\pi_{P,s}$. For fixed $X=x$, define $m_s(t,x)
=
\mathbbm{E}_P
\left[
\psi_s(O;\eta_t)
\mid X=x
\right].$ Taking the conditional expectation under the true distribution gives
\begin{align*}
m_s(t,x)
&=
\frac{\pi_{P,s}(x)}
{\pi_{s,t}(x)}
\left\{
\mu_{P,s}(x)-\mu_{s,t}(x)
\right\}
+
\mu_{s,t}(x)
\\
&=
-\frac{
t\pi_{P,s}(x)\Delta_{\mu,s}(x)
}{
\pi_{s,t}(x)
}
+
\mu_{P,s}(x)
+
t\Delta_{\mu,s}(x)
\\
&=
\mu_{P,s}(x)
+
t\Delta_{\mu,s}(x)
\left\{
1-
\frac{\pi_{P,s}(x)}
{\pi_{s,t}(x)}
\right\}
\\
&=
\mu_{P,s}(x)
+
\frac{
t^2
\Delta_{\mu,s}(x)
\Delta_{\pi,s}(x)
}{
\pi_{s,t}(x)
}.
\end{align*}
Therefore, $\left.
\frac{\partial}{\partial t}
m_s(t,x)
\right|_{t=0}
=
0.$ Since $g^\perp_{\mathrm{ME}}
=
\psi_1-\psi_0,$, it follows that
\[
\left.
\frac{\partial}{\partial t}
\mathbbm{E}_P
\left[
g^\perp_{\mathrm{ME}}(O;\eta_t)
\mid X
\right]
\right|_{t=0}
=
0.
\]
Hence $g^\perp_{\mathrm{ME}}$ is conditionally Neyman orthogonal.

\noindent\textbf{Step 4: Bound on the second-order Gateaux derivative.}

From the preceding display,
\[
m_s(t,x)
=
\mu_{P,s}(x)
+
\frac{
t^2
\Delta_{\mu,s}(x)
\Delta_{\pi,s}(x)
}{
\pi_{P,s}(x)+t\Delta_{\pi,s}(x)
}.
\]
Differentiating twice with respect to $t$ gives
\[
\frac{\partial^2}{\partial t^2}
m_s(t,x)
=
\frac{
2\pi_{P,s}^2(x)
\Delta_{\mu,s}(x)
\Delta_{\pi,s}(x)
}{
\left\{
\pi_{P,s}(x)+t\Delta_{\pi,s}(x)
\right\}^3
}.
\]
Since $\pi_{P,s}(x)+t\Delta_{\pi,s}(x)
=
(1-t)\pi_{P,s}(x)+t\pi_{1,s}(x)
\ge\epsilon$, for every $t\in[0,1]$, we obtain
\[
\left|
\frac{\partial^2}{\partial t^2}
m_s(t,X)
\right|
\le
C
\left|
\Delta_{\mu,s}(X)
\Delta_{\pi,s}(X)
\right|.
\]
Therefore,
\begin{align*}
&
\mathbbm{E}_P
\left[
\left|
\frac{\partial^2}{\partial t^2}
m_s(t,X)
\right|
\right]
\le
C
\mathbbm{E}_P
\left[
\left|
\Delta_{\mu,s}(X)
\Delta_{\pi,s}(X)
\right|
\right]
\le
C
\|\Delta_{\mu,s}\|_{P,2}
\|\Delta_{\pi,s}\|_{P,2},
\end{align*}
where the last inequality follows from the Cauchy--Schwarz
inequality.

Since $g^\perp_{\mathrm{ME}}
=
\psi_1-\psi_0$, we have
\begin{align*}
r'_{n,P}
&=
\sup_{\substack{
t\in[0,1]\\
\eta_1\in\mathcal{T}_n
}}
\mathbbm{E}_P
\left[
\left|
\frac{\partial^2}{\partial t^2}
\mathbbm{E}_P
\left[
g^\perp_{\mathrm{ME}}(O;\eta_t)
\mid X
\right]
\right|
\right]
\\
&\le
C
\sup_{\eta_1\in\mathcal{T}_n}
\sum_{s=0}^1
\|\mu_{1,s}-\mu_{P,s}\|_{P,2}
\|\pi_{1,s}-\pi_{P,s}\|_{P,2}.
\end{align*}
By the assumed product-rate condition, $\|\mu_{1,s}-\mu_{P,s}\|_{P,2}
\|\pi_{1,s}-\pi_{P,s}\|_{P,2}
\le
\frac{\zeta_n}{\sqrt n},$ uniformly over $\eta_1\in\mathcal{T}_n$. Since the number of
components is fixed, $r'_{n,P}
\leq
\frac{\zeta_n}{\sqrt n}.$
After absorbing an absolute constant into $\zeta_n$,
\[
r'_{n,P}
\le
\frac{\zeta_n}{\sqrt n}.
\]
Therefore,
Assumption~3 holds for
$g=g^\perp_{\mathrm{ME}}$.
\end{proof}

The proof for Proposition \ref{prop: example 2} is given as below.
\begin{proof}
For $j\in\{1,2\}$ and $z\in\{0,1\}$, remember that we define $\pi^{(j)}_{P,j}(x)=
    P(Z_j=1\mid X=x),
    \mu^{Y,(j)}_{P,z}(x)
    =
    \mathbbm{E}_P[Y\mid Z_j=z,X=x],
    \mu^{D,(j)}_{P,z}(x)
    =
    \mathbbm{E}_P[D\mid Z_j=z,X=x].$
Furthermore, define $\Delta_{P}^{Y,(j)}(x)
    =
    \mu^{Y,(j)}_{P,1}(x)-\mu^{Y,(j)}_{P,0}(x)$, $\Delta^{D,(j)}_{P}(x)
    =
    \mu^{D,(j)}_{P,1}(x)
    -
    \mu^{D,(j)}_{P,0}(x)$.
For a generic nuisance value
\[
\eta
=
\left\{
\pi^{j},
\mu^{Y,(j)}_z,
\mu^{D,(j)}_z:
j\in\{1,2\},z\in\{0,1\}
\right\},
\]
define analogously $\Delta^{Y,(j)}
=
\mu^{Y,(j)}_1-\mu^{Y,(j)}_0$ and $\Delta^{D,(j)}
=
\mu^{D,(j)}_1-\mu^{D,(j)}_0$.
For each $j\in\{1,2\}$, write
\begin{align*}
g_{Z_j}(O;\eta)
={}&
\frac{1}{\Delta^{D,(j)}(X)}
\left[
\frac{I(Z_j=1)}{\pi^{(j)}(X)}
\left\{
Y-\mu^{Y,(j)}_1(X)
\right\}
-
\frac{I(Z_j=0)}{1-\pi^{(j)}(X)}
\left\{
Y-\mu^{Y,(j)}_0(X)
\right\}
\right]
\nonumber\\
&-
\frac{\Delta^{Y,(j)}(X)}
{\{\Delta^{D,(j)}(X)\}^2}
\left[
\frac{I(Z_j=1)}{\pi^{(j)}(X)}
\left\{
D-\mu^{D,(j)}_1(X)
\right\}
-
\frac{I(Z_j=0)}{1-\pi^{(j)}(X)}
\left\{
D-\mu^{D,(j)}_0(X)
\right\}
\right]
\nonumber\\
&+
\frac{\Delta^{Y,(j)}(X)}
{\Delta^{D,(j)}(X)}.
\end{align*}
Then $g^\perp_{\mathrm{COM}}(O;\eta)
=
g_{Z_1}(O;\eta)-g_{Z_2}(O;\eta).$
We verify each component of Assumption 3.

\noindent\textbf{Step 1: Moment conditions and uniform boundedness.} This part can be verified in a similar way as the first example, thus omitted.

\noindent\textbf{Step 2: Control of $r_{n,P}$.} 

Unlike the first example, here we avoid explicit calculation. For a fixed $j$, define $u^{(j)}(x)
=
\left(
\pi^{(j)}(X),
\mu^{Y,(j)}_1(x),
\mu^{Y,(j)}_0(x),
\mu^{D,(j)}_1(x),
\mu^{D,(j)}_0(x)
\right)^\top,$
and define $u_P^{(j)}(x)$ analogously using the true nuisance
functions. For fixed $o=(x,z_1,z_2,d,y)$, regard
$g_{Z_j}(o;\eta)$ as a function of $u^{(j)}(x)$.
On the set $\epsilon\le \pi^{(j)}(X)\le 1-\epsilon,
\Delta^{D,(j)}(X)\ge\epsilon$ with all nuisance coordinates uniformly bounded,
the map $u\mapsto g_{Z_j}(o;u)$ is continuously differentiable, and every component of its gradient
is uniformly bounded. Therefore, by the multivariate mean-value
theorem,
\[
\left|
g_{Z_j}(O;\eta)-g_{Z_j}(O;\eta_P)
\right|
\le
C
\left\|
u^{(j)}(X)-u_P^{(j)}(x)
\right\|_2.
\]
Squaring and taking expectation gives
\begin{align*}
&
\left\|
g_{Z_j}(O;\eta)-g_{Z_j}(O;\eta_P)
\right\|_{P,2}^2
\\
&\qquad\le
C
\mathbbm{E}_P
\left[
\left\|
u^{(j)}(X)-u_P^{(j)}(x)
\right\|_2^2
\right]
\\
&\qquad=
C
\Bigg[
\|p_j-p_{P,j}\|_{P,2}^2
+
\sum_{z=0}^1
\left\{
\left\|
\mu^{Y,(j)}_z-\mu^{Y,(j)}_{P,z}
\right\|_{P,2}^2
+
\left\|
\mu^{D,(j)}_z-\mu^{D,(j)}_{P,z}
\right\|_{P,2}^2
\right\}
\Bigg].
\end{align*}
Thus, if every nuisance component has $L_2(P)$ error bounded by
$q_n$, then $\left\|g_{Z_j}(O;\eta)-g_{Z_j}(O;\eta_P)
\right\|_{P,2}
\lesssim q_n,$ and therefore
\[
\left\|
g^\perp_{\mathrm{COM}}(O;\eta)
-
g^\perp_{\mathrm{COM}}(O;\eta_P)
\right\|_{P,2}
\lesssim q_n.
\]
Therefore, provided $q_n\lesssim\zeta_n$,
\[
r_{n,P}
=
\sup_{\eta\in\mathcal{T}_n}
\left\{
\mathbbm{E}_P
\left[
\left|
g^\perp_{\mathrm{COM}}(O;\eta)
-
g^\perp_{\mathrm{COM}}(O;\eta_P)
\right|^2
\right]
\right\}^{1/2}
\lesssim
\zeta_n.
\]

\noindent\textbf{Step 3: Conditional Neyman orthogonality.}

We next show that $g_{Z_j}$ is conditionally Neyman orthogonal
for each $j\in\{1,2\}$. Fix $j$ and, temporarily, suppress the
subscripts $j$ and the argument $x$.

Let $\eta_t
=
(1-t)\eta_P+t\eta_1, t\in[0,1],$ for an arbitrary $\eta_1\in\mathcal{T}_n$.
Write $p_t
=
p_0+t h_p$, $\mu^Y_{t,z}
=
\mu^Y_{0,z}+t h^Y_z$ and $\mu^D_{t,z}
=
\mu^D_{0,z}+t h^D_z$, where the subscript $0$ denotes the true nuisance value under DGP $P$. Define $\Delta_Y(t)
=
\mu^Y_{t,1}-\mu^Y_{t,0},
\Delta_D(t)
=
\mu^D_{t,1}-\mu^D_{t,0},
h_{\Delta Y}
=
h^Y_1-h^Y_0,
h_{\Delta D}
=
h^D_1-h^D_0.$
Define
\begin{align*}
A_Y(t)
&=
\frac{p_0}{p_t}
\left(
\mu^Y_{0,1}-\mu^Y_{t,1}
\right)
-
\frac{1-p_0}{1-p_t}
\left(
\mu^Y_{0,0}-\mu^Y_{t,0}
\right),
\\
A_D(t)
&=
\frac{p_0}{p_t}
\left(
\mu^D_{0,1}-\mu^D_{t,1}
\right)
-
\frac{1-p_0}{1-p_t}
\left(
\mu^D_{0,0}-\mu^D_{t,0}
\right).
\end{align*}
Taking conditional expectation of $g_{Z_j}(O;\eta_t)$ under the
true distribution $P$ gives
\begin{align*}
m_j(t,x)
&:=
\mathbbm{E}_P
\left[
g_{Z_j}(O;\eta_t)
\mid X=x
\right]
\nonumber\\
&=
\frac{A_Y(t)}{\Delta_D(t)}
-
\frac{\Delta_Y(t)A_D(t)}
{\{\Delta_D(t)\}^2}
+
\frac{\Delta_Y(t)}
{\Delta_D(t)}.
\end{align*}

At $t=0$, $A_Y(0)=A_D(0)=0.$ Moreover, $A_Y'(0)
=
-h_{\Delta Y},
A_D'(0)
=
-h_{\Delta D}$.
Indeed, $\frac{p_0}{p_t}
\left(
\mu^Y_{0,1}-\mu^Y_{t,1}
\right)
=
-\frac{tp_0h^Y_1}{p_t},$ 
whose derivative at $t=0$ equals $-h^Y_1$, while $-\frac{1-p_0}{1-p_t}
\left(
\mu^Y_{0,0}-\mu^Y_{t,0}
\right)
=
\frac{t(1-p_0)h^Y_0}{1-p_t},$ whose derivative at $t=0$ equals $h^Y_0$. Hence $A_Y'(0)
=
-h^Y_1+h^Y_0
=
-h_{\Delta Y},$ and the same argument gives $A_D'(0)
=
-h_{\Delta D}.$

Furthermore,
\begin{align*}
\left.
\frac{d}{dt}
\frac{\Delta_Y(t)}{\Delta_D(t)}
\right|_{t=0}
&=
\frac{h_{\Delta Y}}{\Delta_D(0)}
-
\frac{
\Delta_Y(0)h_{\Delta D}
}{
\{\Delta_D(0)\}^2
}.
\end{align*}
Differentiating $m_j$ at $t=0$ and using
$A_Y(0)=A_D(0)=0$, we obtain
\begin{align*}
m_j'(0,x)
={}&
\frac{A_Y'(0)}{\Delta_D(0)}
-
\frac{
\Delta_Y(0)A_D'(0)
}{
\{\Delta_D(0)\}^2
}
+
\frac{h_{\Delta Y}}{\Delta_D(0)}
-
\frac{
\Delta_Y(0)h_{\Delta D}
}{
\{\Delta_D(0)\}^2
}
\\
={}&
-\frac{h_{\Delta Y}}{\Delta_D(0)}
+
\frac{
\Delta_Y(0)h_{\Delta D}
}{
\{\Delta_D(0)\}^2
}+
\frac{h_{\Delta Y}}{\Delta_D(0)}
-
\frac{
\Delta_Y(0)h_{\Delta D}
}{
\{\Delta_D(0)\}^2
}
\\
={}&0.
\end{align*}
Thus
\[
\left.
\frac{\partial}{\partial t}
\mathbbm{E}_P
\left[
g_{Z_j}(O;\eta_t)
\mid X
\right]
\right|_{t=0}
=
0.
\]
Since this holds for both $j=1$ and $j=2$,
\[
\left.
\frac{\partial}{\partial t}
\mathbbm{E}_P
\left[
g^\perp_{\mathrm{COM}}(O;\eta_t)
\mid X
\right]
\right|_{t=0}
=
0.
\]
Hence $g^\perp_{\mathrm{COM}}$ is conditionally Neyman orthogonal.

\noindent\textbf{Step 4: Bound on the second-order Gateaux derivative.} Since the second order Gateaux derivative formula is too complex, unlike example 1, here we verify the condition in a  different way. For fixed $j$ and $x$, define $u
=
\left(
\pi,
\mu^Y_1,
\mu^Y_0,
\mu^D_1,
\mu^D_0
\right)^\top$ and $F_{P,j,x}(u)
=
\mathbbm{E}_P
\left[
g_{Z_j}(O;u)
\mid X=x\right]$. $F_{P,j,x}$ is a rational function of the
coordinates of $u$. Its denominators consist only of powers of $\pi,
1-p,
\mu^D_1-\mu^D_0.$
By assumption, along the entire line segment joining
$u_P^{(j)}(x)$ and any $u^{(j)}_{1}(x)$ corresponding to
$\eta_1\in\mathcal{T}_n$,
these quantities are uniformly bounded away from zero.
All coordinates are also uniformly bounded. Consequently,
$F_{P,j,x}$ is twice continuously differentiable on this set and
there exists an absolute constant $C$ such that $\sup_{P,j,x,u}
\left\|
\nabla^2_uF_{P,j,x}(u)
\right\|
\le C.$

Let $h^{(j)}(X)
=
u^{(j)}_{1}(X)-u_P^{(j)}(x)$ 
and $u^{(j)}_{t}(X)
=
u_P^{(j)}(x)+t h^{(j)}(X).$ By the chain rule,
\[
\frac{d^2}{dt^2}
F_{P,j,X}
\left(
u_{j,t}(X)
\right)
=
h_j(X)^\top
\nabla^2_u
F_{P,j,X}
\left(
u_{j,t}(X)
\right)
h_j(X).
\]
Therefore,
\[
\left|
\frac{d^2}{dt^2}
F_{P,j,X}
\left(
u^{(j)}_{t}(X)
\right)
\right|
\le
C\|h^{(j)}(X)\|_2^2.
\]
Taking expectation gives
\begin{align*}
\mathbbm{E}_P
\left[
\left|
\frac{d^2}{dt^2}
F_{P,j,X}
\left(
u^{(j)}_{t}(X)
\right)
\right|
\right]
&\le
C
\mathbbm{E}_P
\left[
\|h_j(X)\|_2^2
\right]
\\
&=
C
\Bigg[
\|\pi^{(j)}_{1,j}-\pi^{(j)}_{P}\|_{P,2}^2
\\
&\qquad+
\sum_{z=0}^1
\left\{
\left\|
\mu^{Y,(j)}_{1,z}
-
\mu^{Y,(j)}_{P,z}
\right\|_{P,2}^2
+
\left\|
\mu^{D,(j)}_{1,z}
-
\mu^{D,(j)}_{P,z}
\right\|_{P,2}^2
\right\}
\Bigg]
\\
&\lesssim q_n^2.
\end{align*}
Since $g^\perp_{\mathrm{COM}}
=
g_{Z_1}-g_{Z_2},$ we conclude that $r'_{n,P}
\lesssim
q_n^2.$  Under the assumed product-rate condition $q_n^2
\lesssim
\frac{\zeta_n}{\sqrt{n}}$,
it follows that $r'_{n,P}
\lesssim
\frac{\zeta_n}{\sqrt{n}}.$ Therefore Assumption~3 holds for
$g=g^\perp_{\mathrm{COM}}$.
\end{proof}

{\color{black}

\section{Additional simulation results}
\label{section:additional simulation results}

\subsection{Additional simulation results with orthonormal Legendre basis and unstandardized tests}
\label{subsec:additional simulation}
We present the additional simulation results for standardized GP tests with orthonormal Legendre basis (Table \ref{tab:sim_legendre_std}), unstandardized GP tests with orthonormal Fourier basis (Table \ref{tab:sim_fourier_unstd}) and orthonormal Legendre basis (Table \ref{tab:sim_legendre_unstd}).

% 1. Standardized orthonormal Legendre basis
\begin{table}[htbp]
\centering
\tiny
\caption{%
\textbf{Panel A}: Proportion of times the mean exchangeability assumption was
rejected across different parameter settings, based on the fixed-dimension,
linear projection (LP) test, the generalized projection (GP) tests, and the
conditional independence test of \citet{racine2006testing}.
\textbf{Panel B}: Proportion of times the IV compatibility condition was rejected
across different parameter settings, based on the LP test and the GP tests.
For the GP test, results correspond to the standardized tests using the
orthonormal Legendre basis. $J^*$ denotes the number of basis functions for each
covariate. Scenario I corresponds to the null hypothesis case; Scenario II
corresponds to a weak nonlinear misalignment; Scenario III corresponds to a weak
linear misalignment; Scenario IV corresponds to a strong nonlinear misalignment;
Scenario V corresponds to a strong linear misalignment.
}
\label{tab:sim_legendre_std}
\renewcommand{\arraystretch}{0.85}

\resizebox{0.90\textwidth}{!}{%
\begin{tabular}{cclccccc}
\toprule
& \begin{tabular}[c]{@{}c@{}}Sample \\ size\end{tabular}
& \begin{tabular}[c]{@{}c@{}}Testing \\ method\end{tabular}
& \begin{tabular}[c]{@{}c@{}}Scenario \\ I\end{tabular}
& \begin{tabular}[c]{@{}c@{}}Scenario \\ II\end{tabular}
& \begin{tabular}[c]{@{}c@{}}Scenario \\ III\end{tabular}
& \begin{tabular}[c]{@{}c@{}}Scenario \\ IV\end{tabular}
& \begin{tabular}[c]{@{}c@{}}Scenario \\ V\end{tabular} \\
\midrule

\multirow{24}{*}{Panel A} & \multirow{6}{*}{250} & LP test & 0.069 & 0.089 & 0.246 & 0.233 & 0.592 \\
 &  & GP test ($J^*=4$) & 0.036 & 0.109 & 0.094 & 0.438 & 0.347 \\
 &  & GP test ($J^*=8$) & 0.016 & 0.050 & 0.040 & 0.227 & 0.181 \\
 &  & GP test ($J^*=16$) & 0.007 & 0.013 & 0.016 & 0.090 & 0.078 \\
 &  & GP test (Combined) & 0.014 & 0.061 & 0.055 & 0.314 & 0.250 \\
 &  & Racine & 0.092 & 0.206 & 0.367 & 0.833 & 0.909 \\

\cmidrule{2-8}

 & \multirow{6}{*}{500} & LP test & 0.062 & 0.059 & 0.428 & 0.393 & 0.922 \\
 &  & GP test ($J^*=4$) & 0.045 & 0.253 & 0.191 & 0.865 & 0.808 \\
 &  & GP test ($J^*=8$) & 0.032 & 0.151 & 0.101 & 0.731 & 0.631 \\
 &  & GP test ($J^*=16$) & 0.018 & 0.055 & 0.047 & 0.496 & 0.376 \\
 &  & GP test (Combined) & 0.034 & 0.185 & 0.133 & 0.812 & 0.717 \\
 &  & Racine & 0.081 & 0.332 & 0.610 & 0.990 & 0.996 \\

\cmidrule{2-8}

 & \multirow{6}{*}{1000} & LP test & 0.065 & 0.062 & 0.739 & 0.695 & 0.996 \\
 &  & GP test ($J^*=5$) & 0.058 & 0.613 & 0.396 & 0.996 & 0.993 \\
 &  & GP test ($J^*=10$) & 0.042 & 0.420 & 0.236 & 0.979 & 0.956 \\
 &  & GP test ($J^*=20$) & 0.019 & 0.226 & 0.124 & 0.917 & 0.845 \\
 &  & GP test (Combined) & 0.035 & 0.495 & 0.295 & 0.991 & 0.986 \\
 &  & Racine & 0.082 & 0.696 & 0.871 & 1.000 & 1.000 \\

\cmidrule{2-8}

 & \multirow{6}{*}{1500} & LP test & 0.061 & 0.057 & 0.911 & 0.892 & 1.000 \\
 &  & GP test ($J^*=5$) & 0.061 & 0.829 & 0.641 & 1.000 & 1.000 \\
 &  & GP test ($J^*=10$) & 0.042 & 0.656 & 0.453 & 1.000 & 0.996 \\
 &  & GP test ($J^*=20$) & 0.024 & 0.423 & 0.242 & 0.981 & 0.977 \\
 &  & GP test (Combined) & 0.040 & 0.768 & 0.528 & 1.000 & 0.998 \\
 &  & Racine & 0.079 & 0.897 & 0.968 & 1.000 & 1.000 \\

\midrule

\multirow{22}{*}{Panel B} & \multirow{5}{*}{1000} & LP test & 0.039 & 0.036 & 0.072 & 0.076 & 0.157 \\
 &  & GP test ($J^*=5$) & 0.052 & 0.080 & 0.063 & 0.177 & 0.133 \\
 &  & GP test ($J^*=10$) & 0.028 & 0.058 & 0.055 & 0.090 & 0.095 \\
 &  & GP test ($J^*=20$) & 0.028 & 0.029 & 0.029 & 0.044 & 0.052 \\
 &  & GP test (Combined) & 0.036 & 0.040 & 0.046 & 0.121 & 0.100 \\

\cmidrule{2-8}

 & \multirow{5}{*}{2000} & LP test & 0.048 & 0.042 & 0.123 & 0.126 & 0.410 \\
 &  & GP test ($J^*=5$) & 0.058 & 0.122 & 0.097 & 0.514 & 0.394 \\
 &  & GP test ($J^*=10$) & 0.058 & 0.089 & 0.087 & 0.372 & 0.270 \\
 &  & GP test ($J^*=20$) & 0.053 & 0.077 & 0.071 & 0.233 & 0.158 \\
 &  & GP test (Combined) & 0.055 & 0.100 & 0.083 & 0.443 & 0.315 \\

\cmidrule{2-8}

 & \multirow{6}{*}{3000} & LP test & 0.038 & 0.055 & 0.175 & 0.200 & 0.662 \\
 &  & GP test ($J^*=5$) & 0.058 & 0.192 & 0.146 & 0.764 & 0.645 \\
 &  & GP test ($J^*=10$) & 0.058 & 0.140 & 0.106 & 0.604 & 0.486 \\
 &  & GP test ($J^*=20$) & 0.032 & 0.087 & 0.082 & 0.438 & 0.316 \\
 &  & GP test ($J^*=40$) & 0.042 & 0.050 & 0.048 & 0.234 & 0.192 \\
 &  & GP test (Combined) & 0.040 & 0.146 & 0.110 & 0.680 & 0.542 \\

\cmidrule{2-8}

 & \multirow{6}{*}{5000} & LP test & 0.045 & 0.038 & 0.359 & 0.385 & 0.917 \\
 &  & GP test ($J^*=5$) & 0.046 & 0.326 & 0.225 & 0.977 & 0.903 \\
 &  & GP test ($J^*=10$) & 0.048 & 0.244 & 0.157 & 0.913 & 0.817 \\
 &  & GP test ($J^*=20$) & 0.057 & 0.170 & 0.092 & 0.780 & 0.652 \\
 &  & GP test ($J^*=40$) & 0.049 & 0.098 & 0.066 & 0.571 & 0.419 \\
 &  & GP test (Combined) & 0.047 & 0.258 & 0.158 & 0.951 & 0.852 \\

\bottomrule
\end{tabular}%
}
\end{table}

% 2. Unstandardized orthonormal Fourier basis
\begin{table}[htbp]
\centering
\tiny
\caption{%
\textbf{Panel A}: Proportion of times the mean exchangeability assumption was
rejected across different parameter settings, based on the fixed-dimension,
linear projection (LP) test, the generalized projection (GP) tests, and the
conditional independence test of \citet{racine2006testing}.
\textbf{Panel B}: Proportion of times the IV compatibility condition was rejected
across different parameter settings, based on the LP test and the GP tests.
For the GP test, results correspond to the unstandardized tests using the
orthonormal Fourier basis. $J^*$ denotes the number of basis functions for each
covariate. Scenario I corresponds to the null hypothesis case; Scenario II
corresponds to a weak nonlinear misalignment; Scenario III corresponds to a weak
linear misalignment; Scenario IV corresponds to a strong nonlinear misalignment;
Scenario V corresponds to a strong linear misalignment.
}
\label{tab:sim_fourier_unstd}
\renewcommand{\arraystretch}{0.85}

\resizebox{0.90\textwidth}{!}{%
\begin{tabular}{cclccccc}
\toprule
& \begin{tabular}[c]{@{}c@{}}Sample \\ size\end{tabular}
& \begin{tabular}[c]{@{}c@{}}Testing \\ method\end{tabular}
& \begin{tabular}[c]{@{}c@{}}Scenario \\ I\end{tabular}
& \begin{tabular}[c]{@{}c@{}}Scenario \\ II\end{tabular}
& \begin{tabular}[c]{@{}c@{}}Scenario \\ III\end{tabular}
& \begin{tabular}[c]{@{}c@{}}Scenario \\ IV\end{tabular}
& \begin{tabular}[c]{@{}c@{}}Scenario \\ V\end{tabular} \\
\midrule

\multirow{24}{*}{Panel A} & \multirow{6}{*}{250} & LP test & 0.069 & 0.089 & 0.246 & 0.233 & 0.592 \\
 &  & GP test ($J^*=4$) & 0.034 & 0.095 & 0.061 & 0.407 & 0.267 \\
 &  & GP test ($J^*=8$) & 0.013 & 0.041 & 0.037 & 0.238 & 0.145 \\
 &  & GP test ($J^*=16$) & 0.006 & 0.011 & 0.007 & 0.075 & 0.061 \\
 &  & GP test (Combined) & 0.011 & 0.030 & 0.021 & 0.220 & 0.127 \\
 &  & Racine & 0.092 & 0.206 & 0.367 & 0.833 & 0.909 \\

\cmidrule{2-8}

 & \multirow{6}{*}{500} & LP test & 0.062 & 0.059 & 0.428 & 0.393 & 0.922 \\
 &  & GP test ($J^*=4$) & 0.034 & 0.261 & 0.122 & 0.889 & 0.755 \\
 &  & GP test ($J^*=8$) & 0.036 & 0.139 & 0.080 & 0.779 & 0.597 \\
 &  & GP test ($J^*=16$) & 0.012 & 0.060 & 0.036 & 0.556 & 0.388 \\
 &  & GP test (Combined) & 0.012 & 0.126 & 0.054 & 0.799 & 0.568 \\
 &  & Racine & 0.081 & 0.332 & 0.610 & 0.990 & 0.996 \\

\cmidrule{2-8}

 & \multirow{6}{*}{1000} & LP test & 0.065 & 0.062 & 0.739 & 0.695 & 0.996 \\
 &  & GP test ($J^*=5$) & 0.051 & 0.626 & 0.306 & 0.999 & 0.994 \\
 &  & GP test ($J^*=10$) & 0.032 & 0.446 & 0.209 & 0.993 & 0.982 \\
 &  & GP test ($J^*=20$) & 0.012 & 0.235 & 0.114 & 0.961 & 0.890 \\
 &  & GP test (Combined) & 0.020 & 0.448 & 0.181 & 0.994 & 0.979 \\
 &  & Racine & 0.082 & 0.696 & 0.871 & 1.000 & 1.000 \\

\cmidrule{2-8}

 & \multirow{6}{*}{1500} & LP test & 0.061 & 0.057 & 0.911 & 0.892 & 1.000 \\
 &  & GP test ($J^*=5$) & 0.041 & 0.867 & 0.509 & 1.000 & 1.000 \\
 &  & GP test ($J^*=10$) & 0.035 & 0.708 & 0.415 & 1.000 & 0.998 \\
 &  & GP test ($J^*=20$) & 0.028 & 0.497 & 0.249 & 0.998 & 0.993 \\
 &  & GP test (Combined) & 0.021 & 0.742 & 0.363 & 1.000 & 1.000 \\
 &  & Racine & 0.079 & 0.897 & 0.968 & 1.000 & 1.000 \\

\midrule

\multirow{22}{*}{Panel B} & \multirow{5}{*}{1000} & LP test & 0.039 & 0.036 & 0.072 & 0.076 & 0.157 \\
 &  & GP test ($J^*=5$) & 0.047 & 0.056 & 0.057 & 0.162 & 0.100 \\
 &  & GP test ($J^*=10$) & 0.043 & 0.050 & 0.048 & 0.088 & 0.085 \\
 &  & GP test ($J^*=20$) & 0.034 & 0.036 & 0.027 & 0.053 & 0.060 \\
 &  & GP test (Combined) & 0.022 & 0.026 & 0.019 & 0.085 & 0.063 \\

\cmidrule{2-8}

 & \multirow{5}{*}{2000} & LP test & 0.048 & 0.042 & 0.123 & 0.126 & 0.410 \\
 &  & GP test ($J^*=5$) & 0.045 & 0.098 & 0.083 & 0.497 & 0.302 \\
 &  & GP test ($J^*=10$) & 0.046 & 0.074 & 0.062 & 0.370 & 0.258 \\
 &  & GP test ($J^*=20$) & 0.044 & 0.055 & 0.050 & 0.222 & 0.162 \\
 &  & GP test (Combined) & 0.025 & 0.060 & 0.044 & 0.350 & 0.218 \\

\cmidrule{2-8}

 & \multirow{6}{*}{3000} & LP test & 0.038 & 0.055 & 0.175 & 0.200 & 0.662 \\
 &  & GP test ($J^*=5$) & 0.042 & 0.185 & 0.104 & 0.732 & 0.530 \\
 &  & GP test ($J^*=10$) & 0.039 & 0.121 & 0.095 & 0.583 & 0.442 \\
 &  & GP test ($J^*=20$) & 0.029 & 0.085 & 0.077 & 0.423 & 0.303 \\
 &  & GP test ($J^*=40$) & 0.029 & 0.069 & 0.047 & 0.267 & 0.198 \\
 &  & GP test (Combined) & 0.022 & 0.086 & 0.053 & 0.576 & 0.382 \\

\cmidrule{2-8}

 & \multirow{6}{*}{5000} & LP test & 0.045 & 0.038 & 0.359 & 0.385 & 0.917 \\
 &  & GP test ($J^*=5$) & 0.035 & 0.298 & 0.148 & 0.960 & 0.832 \\
 &  & GP test ($J^*=10$) & 0.050 & 0.220 & 0.111 & 0.906 & 0.746 \\
 &  & GP test ($J^*=20$) & 0.039 & 0.174 & 0.092 & 0.769 & 0.620 \\
 &  & GP test ($J^*=40$) & 0.038 & 0.110 & 0.074 & 0.573 & 0.443 \\
 &  & GP test (Combined) & 0.023 & 0.191 & 0.093 & 0.907 & 0.741 \\

\bottomrule
\end{tabular}%
}
\end{table}

% 3. Unstandardized orthonormal Legendre basis
\begin{table}[htbp]
\centering
\tiny
\caption{%
\textbf{Panel A}: Proportion of times the mean exchangeability assumption was
rejected across different parameter settings, based on the fixed-dimension,
linear projection (LP) test, the generalized projection (GP) tests, and the
conditional independence test of \citet{racine2006testing}.
\textbf{Panel B}: Proportion of times the IV compatibility condition was rejected
across different parameter settings, based on the LP test and the GP tests.
For the GP test, results correspond to the unstandardized tests using the
orthonormal Legendre basis. $J^*$ denotes the number of basis functions for each
covariate. Scenario I corresponds to the null hypothesis case; Scenario II
corresponds to a weak nonlinear misalignment; Scenario III corresponds to a weak
linear misalignment; Scenario IV corresponds to a strong nonlinear misalignment;
Scenario V corresponds to a strong linear misalignment.
}
\label{tab:sim_legendre_unstd}
\renewcommand{\arraystretch}{0.85}

\resizebox{0.90\textwidth}{!}{%
\begin{tabular}{cclccccc}
\toprule
& \begin{tabular}[c]{@{}c@{}}Sample \\ size\end{tabular}
& \begin{tabular}[c]{@{}c@{}}Testing \\ method\end{tabular}
& \begin{tabular}[c]{@{}c@{}}Scenario \\ I\end{tabular}
& \begin{tabular}[c]{@{}c@{}}Scenario \\ II\end{tabular}
& \begin{tabular}[c]{@{}c@{}}Scenario \\ III\end{tabular}
& \begin{tabular}[c]{@{}c@{}}Scenario \\ IV\end{tabular}
& \begin{tabular}[c]{@{}c@{}}Scenario \\ V\end{tabular} \\
\midrule

\multirow{24}{*}{Panel A} & \multirow{6}{*}{250} & LP test & 0.069 & 0.089 & 0.246 & 0.233 & 0.592 \\
 &  & GP test ($J^*=4$) & 0.022 & 0.078 & 0.074 & 0.361 & 0.290 \\
 &  & GP test ($J^*=8$) & 0.012 & 0.029 & 0.026 & 0.182 & 0.143 \\
 &  & GP test ($J^*=16$) & 0.003 & 0.008 & 0.007 & 0.052 & 0.056 \\
 &  & GP test (Combined) & 0.002 & 0.021 & 0.018 & 0.181 & 0.133 \\
 &  & Racine & 0.092 & 0.206 & 0.367 & 0.833 & 0.909 \\

\cmidrule{2-8}

 & \multirow{6}{*}{500} & LP test & 0.062 & 0.059 & 0.428 & 0.393 & 0.922 \\
 &  & GP test ($J^*=4$) & 0.036 & 0.205 & 0.152 & 0.836 & 0.757 \\
 &  & GP test ($J^*=8$) & 0.026 & 0.103 & 0.076 & 0.675 & 0.564 \\
 &  & GP test ($J^*=16$) & 0.009 & 0.040 & 0.035 & 0.423 & 0.323 \\
 &  & GP test (Combined) & 0.014 & 0.096 & 0.056 & 0.683 & 0.584 \\
 &  & Racine & 0.081 & 0.332 & 0.610 & 0.990 & 0.996 \\

\cmidrule{2-8}

 & \multirow{6}{*}{1000} & LP test & 0.065 & 0.062 & 0.739 & 0.695 & 0.996 \\
 &  & GP test ($J^*=5$) & 0.040 & 0.544 & 0.330 & 0.993 & 0.988 \\
 &  & GP test ($J^*=10$) & 0.030 & 0.352 & 0.177 & 0.972 & 0.933 \\
 &  & GP test ($J^*=20$) & 0.012 & 0.168 & 0.091 & 0.887 & 0.800 \\
 &  & GP test (Combined) & 0.018 & 0.329 & 0.192 & 0.976 & 0.957 \\
 &  & Racine & 0.082 & 0.696 & 0.871 & 1.000 & 1.000 \\

\cmidrule{2-8}

 & \multirow{6}{*}{1500} & LP test & 0.061 & 0.057 & 0.911 & 0.892 & 1.000 \\
 &  & GP test ($J^*=5$) & 0.046 & 0.798 & 0.584 & 1.000 & 0.998 \\
 &  & GP test ($J^*=10$) & 0.023 & 0.601 & 0.398 & 0.998 & 0.991 \\
 &  & GP test ($J^*=20$) & 0.017 & 0.361 & 0.175 & 0.975 & 0.968 \\
 &  & GP test (Combined) & 0.013 & 0.634 & 0.367 & 0.998 & 0.996 \\
 &  & Racine & 0.079 & 0.897 & 0.968 & 1.000 & 1.000 \\

\midrule

\multirow{22}{*}{Panel B} & \multirow{5}{*}{1000} & LP test & 0.039 & 0.036 & 0.072 & 0.076 & 0.157 \\
 &  & GP test ($J^*=5$) & 0.041 & 0.052 & 0.043 & 0.140 & 0.105 \\
 &  & GP test ($J^*=10$) & 0.019 & 0.037 & 0.039 & 0.060 & 0.068 \\
 &  & GP test ($J^*=20$) & 0.020 & 0.016 & 0.018 & 0.024 & 0.036 \\
 &  & GP test (Combined) & 0.016 & 0.017 & 0.019 & 0.062 & 0.052 \\

\cmidrule{2-8}

 & \multirow{5}{*}{2000} & LP test & 0.048 & 0.042 & 0.123 & 0.126 & 0.410 \\
 &  & GP test ($J^*=5$) & 0.040 & 0.088 & 0.079 & 0.464 & 0.334 \\
 &  & GP test ($J^*=10$) & 0.045 & 0.068 & 0.065 & 0.337 & 0.233 \\
 &  & GP test ($J^*=20$) & 0.036 & 0.063 & 0.052 & 0.189 & 0.128 \\
 &  & GP test (Combined) & 0.029 & 0.054 & 0.037 & 0.314 & 0.212 \\

\cmidrule{2-8}

 & \multirow{6}{*}{3000} & LP test & 0.038 & 0.055 & 0.175 & 0.200 & 0.662 \\
 &  & GP test ($J^*=5$) & 0.041 & 0.169 & 0.110 & 0.723 & 0.599 \\
 &  & GP test ($J^*=10$) & 0.038 & 0.124 & 0.085 & 0.570 & 0.439 \\
 &  & GP test ($J^*=20$) & 0.028 & 0.063 & 0.070 & 0.405 & 0.273 \\
 &  & GP test ($J^*=40$) & 0.034 & 0.040 & 0.032 & 0.195 & 0.158 \\
 &  & GP test (Combined) & 0.017 & 0.088 & 0.055 & 0.546 & 0.403 \\

\cmidrule{2-8}

 & \multirow{6}{*}{5000} & LP test & 0.045 & 0.038 & 0.359 & 0.385 & 0.917 \\
 &  & GP test ($J^*=5$) & 0.032 & 0.283 & 0.188 & 0.967 & 0.889 \\
 &  & GP test ($J^*=10$) & 0.038 & 0.205 & 0.128 & 0.897 & 0.778 \\
 &  & GP test ($J^*=20$) & 0.046 & 0.149 & 0.077 & 0.755 & 0.628 \\
 &  & GP test ($J^*=40$) & 0.041 & 0.085 & 0.052 & 0.539 & 0.387 \\
 &  & GP test (Combined) & 0.030 & 0.171 & 0.089 & 0.899 & 0.786 \\

\bottomrule
\end{tabular}%
}
\end{table}

\subsection{Simulation settings with higher-dimensional covariates}
\label{subsec:high-dimensional simulation}

We consider new simulation settings with higher dimensional covariates. These settings follow the data-generating processes in Sections~\ref{subsec:simulation mean exchangeability} and \ref{subsec:simulation IV}, but augment the covariate vector to include ten observed covariates. Specifically, $X_1,\ldots,X_5$ are independently generated from $\mathrm{Uniform}[-1,1]$, and $X_6,\ldots,X_{10}$ are independently generated from $\mathrm{Bernoulli}(0.5)$. For the high-dimensional mean exchangeability simulation, the source indicator is generated according to
\[
P(S=1\mid X)=\mathrm{expit}(X_1-X_2+0.3X_7-0.3X_9),
\]
and the treatment assignment mechanism is
\[
P(A=1\mid X,S)=S\mathrm{expit}(1.5X_1-0.5X_2)+(1-S)\mathrm{expit}(X_1+0.5X_2).
\]
The potential outcomes are generated as
\begin{align*}
Y(0)
&=\mathbbm{1}\{S=0\}\{X_1+X_2+\mathrm{expit}(X_1)+X_6-X_7\}\\
&\quad+\mathbbm{1}\{S=1\}\{X_1+X_2+\mathrm{expit}(X_1)+X_6-X_7\\
&\quad+\alpha_1[\cos(\pi X_3)+\cos(\pi X_4)]+\alpha_2(X_1+X_2)\}
+0.5\epsilon,\\
Y(1)&=Y(0)+2X_1-2X_2,
\end{align*}
where $\epsilon\sim N(0,1)$, and the observed outcome is $Y=AY(1)+(1-A)Y(0)$. The five scenarios are indexed by $(\alpha_1,\alpha_2)\in\{(0,0),(0.2,0),(0,0.2),(0.4,0.2),(0.2,0.4)\}$, with the same interpretation as in Section~\ref{subsec:simulation mean exchangeability}. The sample sizes are $N=250,500,1000$, and $1500$.

For the high-dimensional IV compatibility simulation, the two instruments are generated independently conditional on $X$ according to
\begin{align*}
P(Z_1=1\mid X)&=\mathrm{expit}\{0.5+0.5\mathbbm{1}(X_1>0)-0.5\mathbbm{1}(X_2>0)+0.3X_7-0.3X_9\},\\
P(Z_2=1\mid X)&=\mathrm{expit}\{0.5+0.5\mathbbm{1}(X_1>0)+0.5\mathbbm{1}(X_2>0)+0.3X_7-0.3X_9\}.
\end{align*}
Let $X^*_j=\mathbbm{1}(X_j>0)$ for $j=1,2$, and generate $U\sim N(-0.3,0.3^2)$. Principal stratum membership is generated from the same multinomial model as in Section~\ref{subsec:simulation IV}, with
\begin{align*}
g_{\mathrm{ANT}}&=\exp\{1-X_2^*+0.3\mathbbm{1}(U>0)\},\\
g_{\mathrm{SCO1}}&=\exp\{3.5+0.5X_1^*+X_2^*+0.3\mathbbm{1}(U>0)\},\\
g_{\mathrm{SCO2}}&=\exp\{3.5+0.5X_1^*+X_2^*+0.3\mathbbm{1}(U>0)\},\\
g_{\mathrm{RCO}}&=\exp\{2+X_1^*+X_2^*+0.3\mathbbm{1}(U>0)\},\\
g_{\mathrm{ECO}}&=\exp\{2+X_1^*+0.5X_2^*+0.3\mathbbm{1}(U>0)\}.
\end{align*}
The observed treatment $D$ is then determined by the realized instruments and the principal stratum. The potential outcomes are
\begin{align*}
Y(0)&=1+X_1+X_2+0.3X_6-0.3X_7+U+\epsilon,\\
Y(1)&=\mathbbm{1}\{S=\mathrm{ANT}\}(1+X_1+X_2+U)\\
&\quad+\mathbbm{1}\{S\in\{\mathrm{SCO1},\mathrm{RCO},\mathrm{ECO}\}\}(1-X_1+X_2+U)\\
&\quad+\mathbbm{1}\{S=\mathrm{SCO2}\}\{1-X_1+X_2+U+\beta_1[\cos(\pi X_3)+\cos(\pi X_4)]\\
&\quad+\beta_2(X_1+X_2)\}+\epsilon,
\end{align*}
where $\epsilon\sim N(0,1)$, and $Y=(1-D)Y(0)+DY(1)$. The five scenarios are indexed by $(\beta_1,\beta_2)\in\{(0,0),(0.3,0),(0,0.3),(0.6,0.3),(0.3,0.6)\}$, where the first scenario satisfies the IV compatibility condition and the remaining scenarios introduce nonlinear, linear, or mixed violations. The sample sizes are $N=1000,2000,3000$, and $5000$.

The simulation results for standardized test are shown in Table \ref{tab:sim_high_fourier_std} (with orthonormal Fourier basis) and Table \ref{tab:sim_high_legendre_std} (with orthonormal Legendre basis). The  simulation results for unstandardized test are shown in \ref{tab:sim_high_fourier_unstd} (with orthonormal Fourier basis) and \ref{tab:sim_high_legendre_unstd} (with orthonormal Legendre basis).

\subsection{Computation time}
\label{subsec:computation time}
The computation time is reported in Table \ref{tab:method-timing} and Table \ref{tab:nuisance-timing}.
% Requires \usepackage{booktabs}. The two original table labels are preserved.

% 4. High-dimensional setting: standardized orthonormal Fourier basis
\begin{table}[htbp]
\centering
\tiny
\caption{%
\textbf{Panel A}: Proportion of times the mean exchangeability assumption was
rejected across different parameter settings in the high-dimensional setting, based on the fixed-dimension,
linear projection (LP) test, the generalized projection (GP) tests, and the
conditional independence test of \citet{racine2006testing}.
\textbf{Panel B}: Proportion of times the IV compatibility condition was rejected
across different parameter settings in the high-dimensional setting, based on the LP test and the GP tests.
For the GP test, results correspond to the standardized tests using the
orthonormal Fourier basis. $J^*$ denotes the number of basis functions for each
covariate. Scenario I corresponds to the null hypothesis case; Scenario II
corresponds to a weak nonlinear misalignment; Scenario III corresponds to a weak
linear misalignment; Scenario IV corresponds to a strong nonlinear misalignment;
Scenario V corresponds to a strong linear misalignment.
}
\label{tab:sim_high_fourier_std}
\renewcommand{\arraystretch}{0.85}

\resizebox{0.95\textwidth}{!}{%
\begin{tabular}{cclccccc}
\toprule
& \begin{tabular}[c]{@{}c@{}}Sample \\ size\end{tabular}
& \begin{tabular}[c]{@{}c@{}}Testing \\ method\end{tabular}
& \begin{tabular}[c]{@{}c@{}}Scenario \\ I\end{tabular}
& \begin{tabular}[c]{@{}c@{}}Scenario \\ II\end{tabular}
& \begin{tabular}[c]{@{}c@{}}Scenario \\ III\end{tabular}
& \begin{tabular}[c]{@{}c@{}}Scenario \\ IV\end{tabular}
& \begin{tabular}[c]{@{}c@{}}Scenario \\ V\end{tabular} \\
\midrule

\multirow{24}{*}{Panel A} & \multirow{6}{*}{250} & LP test & 0.018 & 0.022 & 0.046 & 0.034 & 0.096 \\
 &  & GP test ($J^*=4$) & 0.007 & 0.006 & 0.011 & 0.010 & 0.011 \\
 &  & GP test ($J^*=8$) & 0.002 & 0.002 & 0.002 & 0.000 & 0.002 \\
 &  & GP test ($J^*=16$) & 0.000 & 0.000 & 0.000 & 0.000 & 0.000 \\
 &  & GP test (Combined) & 0.001 & 0.003 & 0.004 & 0.002 & 0.003 \\
 &  & Racine & 0.172 & 0.163 & 0.331 & 0.330 & 0.654 \\

\cmidrule{2-8}

 & \multirow{6}{*}{500} & LP test & 0.063 & 0.055 & 0.200 & 0.159 & 0.629 \\
 &  & GP test ($J^*=4$) & 0.023 & 0.056 & 0.041 & 0.224 & 0.124 \\
 &  & GP test ($J^*=8$) & 0.007 & 0.021 & 0.016 & 0.100 & 0.049 \\
 &  & GP test ($J^*=16$) & 0.001 & 0.006 & 0.003 & 0.028 & 0.010 \\
 &  & GP test (Combined) & 0.011 & 0.019 & 0.019 & 0.124 & 0.068 \\
 &  & Racine & 0.137 & 0.152 & 0.392 & 0.386 & 0.905 \\

\cmidrule{2-8}

 & \multirow{6}{*}{1000} & LP test & 0.066 & 0.077 & 0.452 & 0.364 & 0.974 \\
 &  & GP test ($J^*=5$) & 0.038 & 0.204 & 0.082 & 0.871 & 0.689 \\
 &  & GP test ($J^*=10$) & 0.013 & 0.097 & 0.046 & 0.685 & 0.522 \\
 &  & GP test ($J^*=20$) & 0.006 & 0.025 & 0.016 & 0.351 & 0.239 \\
 &  & GP test (Combined) & 0.016 & 0.114 & 0.050 & 0.774 & 0.560 \\
 &  & Racine & 0.112 & 0.108 & 0.613 & 0.807 & 0.998 \\

\cmidrule{2-8}

 & \multirow{6}{*}{1500} & LP test & 0.076 & 0.064 & 0.690 & 0.594 & 1.000 \\
 &  & GP test ($J^*=5$) & 0.055 & 0.483 & 0.179 & 0.992 & 0.970 \\
 &  & GP test ($J^*=10$) & 0.036 & 0.287 & 0.127 & 0.970 & 0.912 \\
 &  & GP test ($J^*=20$) & 0.014 & 0.134 & 0.053 & 0.842 & 0.725 \\
 &  & GP test (Combined) & 0.030 & 0.347 & 0.114 & 0.981 & 0.934 \\
 &  & Racine & 0.086 & 0.106 & 0.794 & 0.704 & 1.000 \\

\midrule

\multirow{22}{*}{Panel B} & \multirow{5}{*}{1000} & LP test & 0.008 & 0.002 & 0.004 & 0.001 & 0.005 \\
 &  & GP test ($J^*=5$) & 0.003 & 0.001 & 0.003 & 0.001 & 0.007 \\
 &  & GP test ($J^*=10$) & 0.001 & 0.001 & 0.002 & 0.000 & 0.002 \\
 &  & GP test ($J^*=20$) & 0.000 & 0.000 & 0.000 & 0.001 & 0.001 \\
 &  & GP test (Combined) & 0.001 & 0.001 & 0.003 & 0.000 & 0.004 \\

\cmidrule{2-8}

 & \multirow{5}{*}{2000} & LP test & 0.010 & 0.019 & 0.028 & 0.020 & 0.071 \\
 &  & GP test ($J^*=5$) & 0.031 & 0.038 & 0.027 & 0.085 & 0.060 \\
 &  & GP test ($J^*=10$) & 0.025 & 0.021 & 0.027 & 0.065 & 0.039 \\
 &  & GP test ($J^*=20$) & 0.014 & 0.014 & 0.014 & 0.036 & 0.020 \\
 &  & GP test (Combined) & 0.021 & 0.020 & 0.026 & 0.059 & 0.032 \\

\cmidrule{2-8}

 & \multirow{6}{*}{3000} & LP test & 0.019 & 0.023 & 0.050 & 0.056 & 0.203 \\
 &  & GP test ($J^*=5$) & 0.040 & 0.063 & 0.051 & 0.272 & 0.160 \\
 &  & GP test ($J^*=10$) & 0.042 & 0.053 & 0.047 & 0.192 & 0.135 \\
 &  & GP test ($J^*=20$) & 0.033 & 0.039 & 0.040 & 0.112 & 0.089 \\
 &  & GP test ($J^*=40$) & 0.019 & 0.023 & 0.019 & 0.074 & 0.041 \\
 &  & GP test (Combined) & 0.037 & 0.043 & 0.033 & 0.179 & 0.118 \\

\cmidrule{2-8}

 & \multirow{6}{*}{5000} & LP test & 0.030 & 0.031 & 0.132 & 0.114 & 0.582 \\
 &  & GP test ($J^*=5$) & 0.055 & 0.132 & 0.089 & 0.674 & 0.439 \\
 &  & GP test ($J^*=10$) & 0.044 & 0.130 & 0.077 & 0.560 & 0.348 \\
 &  & GP test ($J^*=20$) & 0.043 & 0.082 & 0.065 & 0.373 & 0.253 \\
 &  & GP test ($J^*=40$) & 0.040 & 0.055 & 0.050 & 0.237 & 0.149 \\
 &  & GP test (Combined) & 0.035 & 0.099 & 0.069 & 0.549 & 0.331 \\

\bottomrule
\end{tabular}%
}
\end{table}

% 5. High-dimensional setting: standardized orthonormal Legendre basis
\begin{table}[htbp]
\centering
\tiny
\caption{%
\textbf{Panel A}: Proportion of times the mean exchangeability assumption was
rejected across different parameter settings in the high-dimensional setting, based on the fixed-dimension,
linear projection (LP) test, the generalized projection (GP) tests, and the
conditional independence test of \citet{racine2006testing}.
\textbf{Panel B}: Proportion of times the IV compatibility condition was rejected
across different parameter settings in the high-dimensional setting, based on the LP test and the GP tests.
For the GP test, results correspond to the standardized tests using the
orthonormal Legendre basis. $J^*$ denotes the number of basis functions for each
covariate. Scenario I corresponds to the null hypothesis case; Scenario II
corresponds to a weak nonlinear misalignment; Scenario III corresponds to a weak
linear misalignment; Scenario IV corresponds to a strong nonlinear misalignment;
Scenario V corresponds to a strong linear misalignment.
}
\label{tab:sim_high_legendre_std}
\renewcommand{\arraystretch}{0.85}

\resizebox{0.90\textwidth}{!}{%
\begin{tabular}{cclccccc}
\toprule
& \begin{tabular}[c]{@{}c@{}}Sample \\ size\end{tabular}
& \begin{tabular}[c]{@{}c@{}}Testing \\ method\end{tabular}
& \begin{tabular}[c]{@{}c@{}}Scenario \\ I\end{tabular}
& \begin{tabular}[c]{@{}c@{}}Scenario \\ II\end{tabular}
& \begin{tabular}[c]{@{}c@{}}Scenario \\ III\end{tabular}
& \begin{tabular}[c]{@{}c@{}}Scenario \\ IV\end{tabular}
& \begin{tabular}[c]{@{}c@{}}Scenario \\ V\end{tabular} \\
\midrule

\multirow{24}{*}{Panel A} & \multirow{6}{*}{250} & LP test & 0.018 & 0.022 & 0.046 & 0.034 & 0.096 \\
 &  & GP test ($J^*=4$) & 0.006 & 0.005 & 0.009 & 0.008 & 0.015 \\
 &  & GP test ($J^*=8$) & 0.002 & 0.000 & 0.005 & 0.003 & 0.003 \\
 &  & GP test ($J^*=16$) & 0.000 & 0.000 & 0.000 & 0.000 & 0.000 \\
 &  & GP test (Combined) & 0.001 & 0.002 & 0.005 & 0.001 & 0.002 \\
 &  & Racine & 0.172 & 0.163 & 0.331 & 0.330 & 0.654 \\

\cmidrule{2-8}

 & \multirow{6}{*}{500} & LP test & 0.063 & 0.055 & 0.200 & 0.159 & 0.629 \\
 &  & GP test ($J^*=4$) & 0.021 & 0.045 & 0.032 & 0.212 & 0.144 \\
 &  & GP test ($J^*=8$) & 0.007 & 0.014 & 0.021 & 0.088 & 0.055 \\
 &  & GP test ($J^*=16$) & 0.001 & 0.005 & 0.003 & 0.022 & 0.009 \\
 &  & GP test (Combined) & 0.010 & 0.019 & 0.018 & 0.122 & 0.081 \\
 &  & Racine & 0.137 & 0.152 & 0.392 & 0.386 & 0.905 \\

\cmidrule{2-8}

 & \multirow{6}{*}{1000} & LP test & 0.066 & 0.077 & 0.452 & 0.364 & 0.974 \\
 &  & GP test ($J^*=5$) & 0.029 & 0.167 & 0.099 & 0.846 & 0.741 \\
 &  & GP test ($J^*=10$) & 0.016 & 0.073 & 0.047 & 0.627 & 0.506 \\
 &  & GP test ($J^*=20$) & 0.002 & 0.024 & 0.013 & 0.297 & 0.211 \\
 &  & GP test (Combined) & 0.013 & 0.099 & 0.051 & 0.742 & 0.623 \\
 &  & Racine & 0.112 & 0.108 & 0.613 & 0.807 & 0.998 \\

\cmidrule{2-8}

 & \multirow{6}{*}{1500} & LP test & 0.076 & 0.064 & 0.690 & 0.594 & 1.000 \\
 &  & GP test ($J^*=5$) & 0.053 & 0.433 & 0.225 & 0.985 & 0.977 \\
 &  & GP test ($J^*=10$) & 0.026 & 0.258 & 0.120 & 0.944 & 0.902 \\
 &  & GP test ($J^*=20$) & 0.005 & 0.099 & 0.046 & 0.783 & 0.666 \\
 &  & GP test (Combined) & 0.027 & 0.307 & 0.133 & 0.968 & 0.942 \\
 &  & Racine & 0.086 & 0.106 & 0.794 & 0.704 & 1.000 \\

\midrule

\multirow{22}{*}{Panel B} & \multirow{5}{*}{1000} & LP test & 0.008 & 0.002 & 0.004 & 0.001 & 0.005 \\
 &  & GP test ($J^*=5$) & 0.001 & 0.002 & 0.004 & 0.003 & 0.007 \\
 &  & GP test ($J^*=10$) & 0.002 & 0.001 & 0.002 & 0.000 & 0.003 \\
 &  & GP test ($J^*=20$) & 0.001 & 0.000 & 0.000 & 0.000 & 0.000 \\
 &  & GP test (Combined) & 0.001 & 0.001 & 0.002 & 0.001 & 0.003 \\

\cmidrule{2-8}

 & \multirow{5}{*}{2000} & LP test & 0.010 & 0.019 & 0.028 & 0.020 & 0.071 \\
 &  & GP test ($J^*=5$) & 0.028 & 0.035 & 0.026 & 0.072 & 0.058 \\
 &  & GP test ($J^*=10$) & 0.016 & 0.020 & 0.023 & 0.048 & 0.035 \\
 &  & GP test ($J^*=20$) & 0.010 & 0.009 & 0.007 & 0.026 & 0.012 \\
 &  & GP test (Combined) & 0.016 & 0.020 & 0.015 & 0.049 & 0.036 \\

\cmidrule{2-8}

 & \multirow{6}{*}{3000} & LP test & 0.019 & 0.023 & 0.050 & 0.056 & 0.203 \\
 &  & GP test ($J^*=5$) & 0.043 & 0.064 & 0.061 & 0.263 & 0.170 \\
 &  & GP test ($J^*=10$) & 0.036 & 0.045 & 0.045 & 0.179 & 0.122 \\
 &  & GP test ($J^*=20$) & 0.024 & 0.040 & 0.024 & 0.108 & 0.068 \\
 &  & GP test ($J^*=40$) & 0.019 & 0.015 & 0.013 & 0.049 & 0.022 \\
 &  & GP test (Combined) & 0.032 & 0.043 & 0.030 & 0.155 & 0.100 \\

\cmidrule{2-8}

 & \multirow{6}{*}{5000} & LP test & 0.030 & 0.031 & 0.132 & 0.114 & 0.582 \\
 &  & GP test ($J^*=5$) & 0.042 & 0.120 & 0.094 & 0.678 & 0.536 \\
 &  & GP test ($J^*=10$) & 0.044 & 0.111 & 0.073 & 0.559 & 0.374 \\
 &  & GP test ($J^*=20$) & 0.040 & 0.077 & 0.048 & 0.375 & 0.239 \\
 &  & GP test ($J^*=40$) & 0.026 & 0.050 & 0.042 & 0.190 & 0.128 \\
 &  & GP test (Combined) & 0.031 & 0.087 & 0.060 & 0.553 & 0.368 \\

\bottomrule
\end{tabular}%
}
\end{table}

% 6. High-dimensional setting: unstandardized orthonormal Fourier basis
\begin{table}[htbp]
\centering
\tiny
\caption{%
\textbf{Panel A}: Proportion of times the mean exchangeability assumption was
rejected across different parameter settings in the high-dimensional setting, based on the fixed-dimension,
linear projection (LP) test, the generalized projection (GP) tests, and the
conditional independence test of \citet{racine2006testing}.
\textbf{Panel B}: Proportion of times the IV compatibility condition was rejected
across different parameter settings in the high-dimensional setting, based on the LP test and the GP tests.
For the GP test, results correspond to the unstandardized tests using the
orthonormal Fourier basis. $J^*$ denotes the number of basis functions for each
covariate. Scenario I corresponds to the null hypothesis case; Scenario II
corresponds to a weak nonlinear misalignment; Scenario III corresponds to a weak
linear misalignment; Scenario IV corresponds to a strong nonlinear misalignment;
Scenario V corresponds to a strong linear misalignment.
}
\label{tab:sim_high_fourier_unstd}
\renewcommand{\arraystretch}{0.85}

\resizebox{0.90\textwidth}{!}{%
\begin{tabular}{cclccccc}
\toprule
& \begin{tabular}[c]{@{}c@{}}Sample \\ size\end{tabular}
& \begin{tabular}[c]{@{}c@{}}Testing \\ method\end{tabular}
& \begin{tabular}[c]{@{}c@{}}Scenario \\ I\end{tabular}
& \begin{tabular}[c]{@{}c@{}}Scenario \\ II\end{tabular}
& \begin{tabular}[c]{@{}c@{}}Scenario \\ III\end{tabular}
& \begin{tabular}[c]{@{}c@{}}Scenario \\ IV\end{tabular}
& \begin{tabular}[c]{@{}c@{}}Scenario \\ V\end{tabular} \\
\midrule

\multirow{24}{*}{Panel A} & \multirow{6}{*}{250} & LP test & 0.018 & 0.022 & 0.046 & 0.034 & 0.096 \\
 &  & GP test ($J^*=4$) & 0.001 & 0.003 & 0.007 & 0.002 & 0.006 \\
 &  & GP test ($J^*=8$) & 0.001 & 0.000 & 0.001 & 0.000 & 0.001 \\
 &  & GP test ($J^*=16$) & 0.000 & 0.000 & 0.000 & 0.000 & 0.000 \\
 &  & GP test (Combined) & 0.000 & 0.000 & 0.000 & 0.000 & 0.000 \\
 &  & Racine & 0.172 & 0.163 & 0.331 & 0.330 & 0.654 \\

\cmidrule{2-8}

 & \multirow{6}{*}{500} & LP test & 0.063 & 0.055 & 0.200 & 0.159 & 0.629 \\
 &  & GP test ($J^*=4$) & 0.018 & 0.034 & 0.027 & 0.171 & 0.081 \\
 &  & GP test ($J^*=8$) & 0.002 & 0.013 & 0.011 & 0.064 & 0.034 \\
 &  & GP test ($J^*=16$) & 0.000 & 0.002 & 0.000 & 0.018 & 0.008 \\
 &  & GP test (Combined) & 0.001 & 0.007 & 0.011 & 0.050 & 0.021 \\
 &  & Racine & 0.137 & 0.152 & 0.392 & 0.386 & 0.905 \\

\cmidrule{2-8}

 & \multirow{6}{*}{1000} & LP test & 0.066 & 0.077 & 0.452 & 0.364 & 0.974 \\
 &  & GP test ($J^*=5$) & 0.026 & 0.148 & 0.063 & 0.830 & 0.623 \\
 &  & GP test ($J^*=10$) & 0.010 & 0.069 & 0.038 & 0.607 & 0.442 \\
 &  & GP test ($J^*=20$) & 0.004 & 0.016 & 0.010 & 0.263 & 0.192 \\
 &  & GP test (Combined) & 0.004 & 0.051 & 0.023 & 0.614 & 0.387 \\
 &  & Racine & 0.112 & 0.108 & 0.613 & 0.807 & 0.998 \\

\cmidrule{2-8}

 & \multirow{6}{*}{1500} & LP test & 0.076 & 0.064 & 0.690 & 0.594 & 1.000 \\
 &  & GP test ($J^*=5$) & 0.035 & 0.406 & 0.138 & 0.983 & 0.950 \\
 &  & GP test ($J^*=10$) & 0.026 & 0.235 & 0.096 & 0.954 & 0.882 \\
 &  & GP test ($J^*=20$) & 0.009 & 0.105 & 0.034 & 0.799 & 0.663 \\
 &  & GP test (Combined) & 0.011 & 0.197 & 0.049 & 0.958 & 0.854 \\
 &  & Racine & 0.086 & 0.106 & 0.794 & 0.704 & 1.000 \\

\midrule

\multirow{22}{*}{Panel B} & \multirow{5}{*}{1000} & LP test & 0.008 & 0.002 & 0.004 & 0.001 & 0.005 \\
 &  & GP test ($J^*=5$) & 0.002 & 0.001 & 0.003 & 0.001 & 0.004 \\
 &  & GP test ($J^*=10$) & 0.001 & 0.000 & 0.001 & 0.000 & 0.000 \\
 &  & GP test ($J^*=20$) & 0.000 & 0.000 & 0.000 & 0.000 & 0.000 \\
 &  & GP test (Combined) & 0.000 & 0.000 & 0.000 & 0.000 & 0.000 \\

\cmidrule{2-8}

 & \multirow{5}{*}{2000} & LP test & 0.010 & 0.019 & 0.028 & 0.020 & 0.071 \\
 &  & GP test ($J^*=5$) & 0.024 & 0.028 & 0.022 & 0.065 & 0.040 \\
 &  & GP test ($J^*=10$) & 0.015 & 0.014 & 0.022 & 0.047 & 0.028 \\
 &  & GP test ($J^*=20$) & 0.013 & 0.012 & 0.009 & 0.028 & 0.014 \\
 &  & GP test (Combined) & 0.011 & 0.006 & 0.009 & 0.024 & 0.011 \\

\cmidrule{2-8}

 & \multirow{6}{*}{3000} & LP test & 0.019 & 0.023 & 0.050 & 0.056 & 0.203 \\
 &  & GP test ($J^*=5$) & 0.033 & 0.049 & 0.040 & 0.218 & 0.130 \\
 &  & GP test ($J^*=10$) & 0.031 & 0.042 & 0.042 & 0.164 & 0.112 \\
 &  & GP test ($J^*=20$) & 0.027 & 0.034 & 0.030 & 0.099 & 0.076 \\
 &  & GP test ($J^*=40$) & 0.015 & 0.020 & 0.015 & 0.061 & 0.033 \\
 &  & GP test (Combined) & 0.018 & 0.020 & 0.014 & 0.095 & 0.051 \\

\cmidrule{2-8}

 & \multirow{6}{*}{5000} & LP test & 0.030 & 0.031 & 0.132 & 0.114 & 0.582 \\
 &  & GP test ($J^*=5$) & 0.041 & 0.107 & 0.075 & 0.624 & 0.382 \\
 &  & GP test ($J^*=10$) & 0.033 & 0.107 & 0.068 & 0.518 & 0.311 \\
 &  & GP test ($J^*=20$) & 0.035 & 0.068 & 0.053 & 0.338 & 0.221 \\
 &  & GP test ($J^*=40$) & 0.030 & 0.046 & 0.041 & 0.211 & 0.136 \\
 &  & GP test (Combined) & 0.021 & 0.050 & 0.035 & 0.428 & 0.217 \\

\bottomrule
\end{tabular}%
}
\end{table}

% 7. High-dimensional setting: unstandardized orthonormal Legendre basis
\begin{table}[htbp]
\centering
\tiny
\caption{%
\textbf{Panel A}: Proportion of times the mean exchangeability assumption was
rejected across different parameter settings in the high-dimensional setting, based on the fixed-dimension,
linear projection (LP) test, the generalized projection (GP) tests, and the
conditional independence test of \citet{racine2006testing}.
\textbf{Panel B}: Proportion of times the IV compatibility condition was rejected
across different parameter settings in the high-dimensional setting, based on the LP test and the GP tests.
For the GP test, results correspond to the unstandardized tests using the
orthonormal Legendre basis. $J^*$ denotes the number of basis functions for each
covariate. Scenario I corresponds to the null hypothesis case; Scenario II
corresponds to a weak nonlinear misalignment; Scenario III corresponds to a weak
linear misalignment; Scenario IV corresponds to a strong nonlinear misalignment;
Scenario V corresponds to a strong linear misalignment.
}
\label{tab:sim_high_legendre_unstd}
\renewcommand{\arraystretch}{0.85}

\resizebox{0.90\textwidth}{!}{%
\begin{tabular}{cclccccc}
\toprule
& \begin{tabular}[c]{@{}c@{}}Sample \\ size\end{tabular}
& \begin{tabular}[c]{@{}c@{}}Testing \\ method\end{tabular}
& \begin{tabular}[c]{@{}c@{}}Scenario \\ I\end{tabular}
& \begin{tabular}[c]{@{}c@{}}Scenario \\ II\end{tabular}
& \begin{tabular}[c]{@{}c@{}}Scenario \\ III\end{tabular}
& \begin{tabular}[c]{@{}c@{}}Scenario \\ IV\end{tabular}
& \begin{tabular}[c]{@{}c@{}}Scenario \\ V\end{tabular} \\
\midrule

\multirow{24}{*}{Panel A} & \multirow{6}{*}{250} & LP test & 0.018 & 0.022 & 0.046 & 0.034 & 0.096 \\
 &  & GP test ($J^*=4$) & 0.003 & 0.002 & 0.006 & 0.003 & 0.006 \\
 &  & GP test ($J^*=8$) & 0.001 & 0.000 & 0.000 & 0.000 & 0.001 \\
 &  & GP test ($J^*=16$) & 0.000 & 0.000 & 0.000 & 0.000 & 0.000 \\
 &  & GP test (Combined) & 0.001 & 0.000 & 0.000 & 0.000 & 0.000 \\
 &  & Racine & 0.172 & 0.163 & 0.331 & 0.330 & 0.654 \\

\cmidrule{2-8}

 & \multirow{6}{*}{500} & LP test & 0.063 & 0.055 & 0.200 & 0.159 & 0.629 \\
 &  & GP test ($J^*=4$) & 0.014 & 0.028 & 0.022 & 0.157 & 0.106 \\
 &  & GP test ($J^*=8$) & 0.003 & 0.008 & 0.009 & 0.062 & 0.031 \\
 &  & GP test ($J^*=16$) & 0.001 & 0.002 & 0.002 & 0.014 & 0.006 \\
 &  & GP test (Combined) & 0.001 & 0.003 & 0.007 & 0.050 & 0.025 \\
 &  & Racine & 0.137 & 0.152 & 0.392 & 0.386 & 0.905 \\

\cmidrule{2-8}

 & \multirow{6}{*}{1000} & LP test & 0.066 & 0.077 & 0.452 & 0.364 & 0.974 \\
 &  & GP test ($J^*=5$) & 0.018 & 0.132 & 0.070 & 0.795 & 0.676 \\
 &  & GP test ($J^*=10$) & 0.008 & 0.057 & 0.034 & 0.541 & 0.414 \\
 &  & GP test ($J^*=20$) & 0.001 & 0.012 & 0.007 & 0.212 & 0.160 \\
 &  & GP test (Combined) & 0.002 & 0.038 & 0.022 & 0.571 & 0.446 \\
 &  & Racine & 0.112 & 0.108 & 0.613 & 0.807 & 0.998 \\

\cmidrule{2-8}

 & \multirow{6}{*}{1500} & LP test & 0.076 & 0.064 & 0.690 & 0.594 & 1.000 \\
 &  & GP test ($J^*=5$) & 0.039 & 0.359 & 0.175 & 0.979 & 0.960 \\
 &  & GP test ($J^*=10$) & 0.021 & 0.204 & 0.090 & 0.924 & 0.869 \\
 &  & GP test ($J^*=20$) & 0.002 & 0.077 & 0.037 & 0.725 & 0.598 \\
 &  & GP test (Combined) & 0.008 & 0.159 & 0.063 & 0.933 & 0.888 \\
 &  & Racine & 0.086 & 0.106 & 0.794 & 0.704 & 1.000 \\

\midrule

\multirow{22}{*}{Panel B} & \multirow{5}{*}{1000} & LP test & 0.008 & 0.002 & 0.004 & 0.001 & 0.005 \\
 &  & GP test ($J^*=5$) & 0.001 & 0.001 & 0.003 & 0.001 & 0.005 \\
 &  & GP test ($J^*=10$) & 0.000 & 0.000 & 0.000 & 0.000 & 0.000 \\
 &  & GP test ($J^*=20$) & 0.000 & 0.000 & 0.000 & 0.000 & 0.000 \\
 &  & GP test (Combined) & 0.000 & 0.000 & 0.000 & 0.000 & 0.000 \\

\cmidrule{2-8}

 & \multirow{5}{*}{2000} & LP test & 0.010 & 0.019 & 0.028 & 0.020 & 0.071 \\
 &  & GP test ($J^*=5$) & 0.023 & 0.026 & 0.018 & 0.060 & 0.044 \\
 &  & GP test ($J^*=10$) & 0.013 & 0.010 & 0.018 & 0.035 & 0.025 \\
 &  & GP test ($J^*=20$) & 0.010 & 0.004 & 0.005 & 0.020 & 0.009 \\
 &  & GP test (Combined) & 0.006 & 0.005 & 0.006 & 0.016 & 0.011 \\

\cmidrule{2-8}

 & \multirow{6}{*}{3000} & LP test & 0.019 & 0.023 & 0.050 & 0.056 & 0.203 \\
 &  & GP test ($J^*=5$) & 0.037 & 0.051 & 0.049 & 0.212 & 0.136 \\
 &  & GP test ($J^*=10$) & 0.026 & 0.037 & 0.035 & 0.143 & 0.100 \\
 &  & GP test ($J^*=20$) & 0.018 & 0.030 & 0.018 & 0.087 & 0.052 \\
 &  & GP test ($J^*=40$) & 0.012 & 0.012 & 0.010 & 0.040 & 0.016 \\
 &  & GP test (Combined) & 0.011 & 0.023 & 0.013 & 0.075 & 0.048 \\

\cmidrule{2-8}

 & \multirow{6}{*}{5000} & LP test & 0.030 & 0.031 & 0.132 & 0.114 & 0.582 \\
 &  & GP test ($J^*=5$) & 0.029 & 0.098 & 0.068 & 0.645 & 0.477 \\
 &  & GP test ($J^*=10$) & 0.036 & 0.090 & 0.060 & 0.517 & 0.341 \\
 &  & GP test ($J^*=20$) & 0.035 & 0.065 & 0.042 & 0.338 & 0.221 \\
 &  & GP test ($J^*=40$) & 0.024 & 0.039 & 0.032 & 0.163 & 0.103 \\
 &  & GP test (Combined) & 0.014 & 0.039 & 0.031 & 0.414 & 0.236 \\

\bottomrule
\end{tabular}%
}
\end{table}

\begin{table}[!htbp]
\centering
\caption{Computation time for the testing methods in the usual and high-dimensional settings, in seconds.}
\label{tab:method-timing}
\small
\setlength{\tabcolsep}{3pt}
\begin{tabular}{@{}crlrrrr@{\hspace{12pt}}rrrr@{}}
\toprule
& & & \multicolumn{4}{c}{Usual setting ($d=2$)} & \multicolumn{4}{c}{High-dimensional setting ($d=10$)} \\
\cmidrule(lr){4-7}\cmidrule(l){8-11}
Example & $n$ & Method & Mean & Median & Q25 & Q75 & Mean & Median & Q25 & Q75 \\
\midrule
1 & 250 & GP test (Fourier) & 0.120 & 0.101 & 0.098 & 0.123 & 0.271 & 0.243 & 0.237 & 0.261 \\
 &  & GP test (Legendre) & 0.129 & 0.101 & 0.098 & 0.111 & 0.278 & 0.243 & 0.236 & 0.340 \\
 &  & LP test & 0.004 & 0.001 & 0.001 & 0.002 & 0.004 & 0.001 & 0.001 & 0.002 \\
 &  & Racine test & 1.889 & 1.869 & 1.759 & 1.999 & 15.450 & 15.219 & 13.806 & 16.807 \\
\addlinespace
 & 500 & GP test (Fourier) & 0.105 & 0.099 & 0.096 & 0.100 & 0.265 & 0.239 & 0.235 & 0.246 \\
 &  & GP test (Legendre) & 0.103 & 0.099 & 0.095 & 0.101 & 0.256 & 0.242 & 0.237 & 0.248 \\
 &  & LP test & 0.001 & 0.001 & 0.001 & 0.001 & 0.002 & 0.001 & 0.001 & 0.002 \\
 &  & Racine test & 4.972 & 4.941 & 4.606 & 5.302 & 53.927 & 53.263 & 49.248 & 57.797 \\
\addlinespace
 & 1000 & GP test (Fourier) & 0.130 & 0.119 & 0.117 & 0.121 & 0.335 & 0.296 & 0.289 & 0.308 \\
 &  & GP test (Legendre) & 0.126 & 0.119 & 0.116 & 0.121 & 0.317 & 0.296 & 0.290 & 0.303 \\
 &  & LP test & 0.001 & 0.001 & 0.001 & 0.002 & 0.002 & 0.002 & 0.002 & 0.002 \\
 &  & Racine test & 16.643 & 16.546 & 15.526 & 17.663 & 199.346 & 197.871 & 186.010 & 211.060 \\
\addlinespace
 & 1500 & GP test (Fourier) & 0.135 & 0.120 & 0.118 & 0.122 & 0.337 & 0.299 & 0.293 & 0.311 \\
 &  & GP test (Legendre) & 0.129 & 0.119 & 0.117 & 0.122 & 0.325 & 0.298 & 0.293 & 0.306 \\
 &  & LP test & 0.002 & 0.002 & 0.001 & 0.002 & 0.003 & 0.002 & 0.002 & 0.002 \\
 &  & Racine test & 35.693 & 35.396 & 33.298 & 37.816 & 432.259 & 428.036 & 406.174 & 453.781 \\
\midrule
2 & 1000 & GP test (Fourier) & 0.123 & 0.118 & 0.115 & 0.121 & 0.300 & 0.287 & 0.283 & 0.294 \\
 &  & GP test (Legendre) & 0.124 & 0.118 & 0.115 & 0.122 & 0.295 & 0.287 & 0.282 & 0.293 \\
 &  & LP test & 0.002 & 0.002 & 0.001 & 0.002 & 0.002 & 0.002 & 0.002 & 0.002 \\
\addlinespace
 & 2000 & GP test (Fourier) & 0.126 & 0.120 & 0.118 & 0.124 & 0.305 & 0.295 & 0.290 & 0.301 \\
 &  & GP test (Legendre) & 0.126 & 0.120 & 0.117 & 0.124 & 0.300 & 0.294 & 0.289 & 0.300 \\
 &  & LP test & 0.002 & 0.002 & 0.002 & 0.002 & 0.003 & 0.003 & 0.003 & 0.003 \\
\addlinespace
 & 3000 & GP test (Fourier) & 0.261 & 0.251 & 0.246 & 0.256 & 0.641 & 0.622 & 0.612 & 0.638 \\
 &  & GP test (Legendre) & 0.259 & 0.249 & 0.244 & 0.255 & 0.630 & 0.618 & 0.608 & 0.633 \\
 &  & LP test & 0.003 & 0.002 & 0.002 & 0.003 & 0.004 & 0.003 & 0.003 & 0.004 \\
\addlinespace
 & 5000 & GP test (Fourier) & 0.276 & 0.264 & 0.259 & 0.271 & 0.680 & 0.663 & 0.652 & 0.678 \\
 &  & GP test (Legendre) & 0.271 & 0.262 & 0.256 & 0.269 & 0.666 & 0.657 & 0.646 & 0.671 \\
 &  & LP test & 0.003 & 0.003 & 0.003 & 0.003 & 0.005 & 0.005 & 0.004 & 0.005 \\
\bottomrule
\end{tabular}

\medskip
\begin{minipage}{0.98\linewidth}
\footnotesize
\textit{Note:} The usual and high-dimensional settings contain two and ten covariates, respectively.
For each example and sample size, summaries pool times across five simulation settings and
1,000 replicates per setting in the usual setting (5,000 timing observations), and 999 replicates
per setting in the high-dimensional setting (4,995 timing observations).
One high-dimensional replicate is excluded following a SuperLearner fitting failure.
Each timing observation represents one method block in one setting and one replicate;
times are not summed over settings. Q25 and Q75 are the 25th and 75th percentiles.
Method times exclude nuisance estimation. For each GP basis, the timed block includes basis
construction, all reported values of $J$, both standardized and unstandardized tests,
eigenvalue computation and 10,000 weighted chi-square draws per value of $J$ for the
unstandardized tests, and both combined tests.
\end{minipage}
\end{table}

\begin{table}[!htbp]
\centering
\caption{Computation time for nuisance estimation in the usual and high-dimensional settings, in seconds.}
\label{tab:nuisance-timing}
\small
\setlength{\tabcolsep}{4pt}
\begin{tabular}{@{}crrrrr@{\hspace{12pt}}rrrr@{}}
\toprule
& & \multicolumn{4}{c}{Usual setting ($d=2$)} & \multicolumn{4}{c}{High-dimensional setting ($d=10$)} \\
\cmidrule(lr){3-6}\cmidrule(l){7-10}
Example & $n$ & Mean & Median & Q25 & Q75 & Mean & Median & Q25 & Q75 \\
\midrule
1 & 250 & 0.466 & 0.442 & 0.418 & 0.494 & 0.646 & 0.625 & 0.593 & 0.684 \\
 & 500 & 0.834 & 0.808 & 0.776 & 0.855 & 1.175 & 1.154 & 1.107 & 1.223 \\
 & 1000 & 1.766 & 1.769 & 1.651 & 1.853 & 2.467 & 2.464 & 2.358 & 2.569 \\
 & 1500 & 2.883 & 2.867 & 2.728 & 3.033 & 3.955 & 3.934 & 3.797 & 4.093 \\
\midrule
2 & 1000 & 18.853 & 18.562 & 17.910 & 19.586 & 29.448 & 29.221 & 28.572 & 30.215 \\
 & 2000 & 40.823 & 40.502 & 39.719 & 41.718 & 65.222 & 64.937 & 63.997 & 66.300 \\
 & 3000 & 64.981 & 64.824 & 63.870 & 65.933 & 104.112 & 103.786 & 102.535 & 105.522 \\
 & 5000 & 119.966 & 119.782 & 118.039 & 121.675 & 190.896 & 190.221 & 188.330 & 192.952 \\
\bottomrule
\end{tabular}

\medskip
\begin{minipage}{0.98\linewidth}
\footnotesize
\textit{Note:} The usual and high-dimensional settings contain two and ten covariates, respectively.
Times cover SuperLearner fitting and cross-fitted pseudo-outcome construction.
For each example and sample size, summaries pool times across five simulation settings and
1,000 replicates per setting in the usual setting.
One high-dimensional replicate is excluded following a SuperLearner fitting failure.
There is one timing observation per setting and replicate, counted once across all testing methods.
Q25 and Q75 are the 25th and 75th percentiles.
\end{minipage}
\end{table}}
\end{document}